\pdfoutput=1
\RequirePackage{ifpdf}
\ifpdf 
\documentclass[pdftex]{sigma}
\else
\documentclass{sigma}
\fi

\usepackage{mathrsfs, mathtools}
\usepackage{frcursive, enumitem}
\definecolor{lightblue}{rgb}{0.68, 0.85, 0.9}
\usepackage{faktor}
\usepackage{slashbox}

\usepackage{dynkin-diagrams}
\usepackage{tikz-cd}
\usepackage{tikz}
\usetikzlibrary{braids}

\usepackage{quiver}
\usepackage{mathrsfs}
\usepackage{lscape}

\usepackage{physics}

\usepackage{multicol}
\usepackage{multirow}
\usepackage{booktabs}

\definecolor{amber}{rgb}{1.0, 0.49, 0.0}
\definecolor{Green}{rgb}{0.0, 0.5, 0.0}
\definecolor{purple}{rgb}{0.7,0,0.7}

\newcommand{\ceff}[0]{c_{\rm{eff}}}
\newcommand{\ex}[0]{\ensuremath{\mathrm{e}}}

\newcommand{\ZZ}{\ensuremath{\mathbb{Z}}}
\newcommand{\zhat}{\ensuremath{\widehat{Z}}}
\newcommand{\brs}[3]{\overline{\Sigma({#1},{#2},{#3})}}

\newcommand*\xbar[1]{
 \hbox{
 \vbox{
 \hrule height 0.5pt 
 \kern0.5ex
 \hbox{%
 \kern-0.1em
 \ensuremath{#1}%
 \kern-0.1em
 }%
 }%
 }%
}

\renewcommand{\vec}[1]{\mathbf{#1}}

\numberwithin{equation}{section}

\newtheorem{Theorem}{Theorem}[section]
\newtheorem*{Theorem*}{Theorem}
\newtheorem{Corollary}[Theorem]{Corollary}
\newtheorem{Lemma}[Theorem]{Lemma}
\newtheorem{Proposition}[Theorem]{Proposition}
\newtheorem{coj}[Theorem]{Conjecture}

\theoremstyle{definition}
\newtheorem{Definition}[Theorem]{Definition}

\newtheorem{Example}[Theorem]{Example}
\newtheorem{Remark}[Theorem]{Remark}

\definecolor{dotInt}{RGB}{180,28,30} 
\definecolor{dotSoft}{RGB}{236,140,130} 

\DeclareRobustCommand{\softdot}{\tikz[baseline=-0.45ex]{\fill[dotSoft] (0,0) circle (0.4ex);}}
\DeclareRobustCommand{\intensedot}{\tikz[baseline=-0.45ex]{\fill[dotInt] (0,0) circle (0.4ex);}}

\begin{document}
\allowdisplaybreaks

\newcommand{\arXivNumber}{2508.10087}

\renewcommand{\PaperNumber}{090}

\FirstPageHeading

\ShortArticleName{$c_{\rm eff}$ from Surgery and Modularity}

\ArticleName{$\boldsymbol{c_{\rm eff}}$ from Surgery and Modularity}

\Author{Shimal HARICHURN~$^{\rm a}$, Mrunmay JAGADALE~$^{\rm b}$, Dmitry NOSHCHENKO~$^{\rm c}$\newline and Davide PASSARO~$^{\rm b}$}

\AuthorNameForHeading{S.~Harichurn, M.~Jagadale, D.~Noshchenko and D.~Passaro}

\Address{$^{\rm a)}$~Durban, KwaZulu-Natal, South Africa}
\EmailD{\mail{sharichurn.research@gmail.com}}

\Address{$^{\rm b)}$~Walter Burke Institute for Theoretical Physics, California Institute of Technology,\\
\hphantom{$^{\rm b)}$}~Pasadena, CA 91125, USA}
\EmailD{\mail{mjagadal@caltech.edu}, \mail{dpassaro@caltech.edu}}

\Address{$^{\rm c)}$~School of Theoretical Physics, Dublin Institute for Advanced Studies,\\
\hphantom{$^{\rm c)}$}~10 Burlington Road, Dublin~4, D04 C932, Ireland}
\EmailD{\mail{dsnoshchenko@stp.dias.ie}}

\ArticleDates{Received November 29, 2025, in final form August 23, 2026; Published online September 22, 2026}

\Abstract{$\zhat$ invariants, rigorously defined for negative definite plumbed 3-manifolds, are expected -- on physical grounds -- to exist for every closed, oriented 3-manifold (Gukov et al., 2020). Several prescriptions have been proposed to extend their definition to generic plumbings by reversing the orientation of a~negative definite plumbing, thus turning it into a~positive definite one (Cheng et al., 2019). Two existing proposals are relevant for this paper: (i) the regularized $+1/r$-surgery conjecture combined with the false-mock modular conjecture (Park, 2021 and Cheng et al., 2024), and (ii) a~construction based on resurgence and a~false theta function duality (Costin et al., 2023). In this note, we compare these proposals on the class of Brieskorn homology spheres $\brs{s}{t}{rst\pm1}$ and find that they are incompatible in general. Our diagnostic is the effective central charge, $c_{\rm eff}$, which governs the asymptotic growth of coefficients of $\zhat$ (Gukov and Jagadale, 2024). First, we prove that the upper bound on $\ceff$ from prescription (i) is governed by the Ramanujan theta function, which regularizes the surgery formula. Second, we develop numerical and modular tools that deliver the lower bounds as well as exact values via mixed mock-modular analysis. Complementing this, we also study $\ceff$ for negative definite plumbed 3-manifolds which allow for a~better comparison of pairs of 3-manifolds related by orientation reversal. As a~result, we find that for some Brieskorn spheres the surgery and false-mock prescriptions violate the expected relation between $c_{\rm eff}$, Chern--Simons invariants and non-abelian flat connections.
These findings underscore $c_{\rm eff}$ as a~sensitive probe of the ``positive side'' of $\zhat$-theory.}


\Classification{57K31; 57K16; 11F37; 11F27; 58J28; 81T30; 40G10}

\section{Introduction}

$\widehat{Z}$ invariants are quantum $q$-series invariants for closed, oriented $3$-manifolds.
They were originally introduced by Gukov, Pei, Putrov and Vafa~\cite{gukov2020bps}.
The $\zhat$ invariants were shown to enjoy two complementary realizations: one physical and one mathematical.
The physical perspective starts with a~six-dimensional $\mathcal N=(2,0)$ theory on $D^{2}\!\times\!S^{1}\!\times\!M_{3}$.
Compactifying this theory over~$M_{3}$ yields a~three-dimensional $\mathcal{N}\!=\!2$ theory $T[M_{3}]$ on the solid torus $D^2 \times S^1$;
the half-index of~$T[M_{3}]$ on $D^{2}\!\times_{q}\!S^{1}$ with boundary condition $b\in \operatorname{Tor}H_1(M_{3})/\mathbb Z_2$ is precisely $\zhat_{b}(M_3;\tau)$.
Compactifying instead over $D^{2}\!\times\!S^{1}$ gives complex Chern--Simons theory on $M_{3}$.
The defining relation between the level-$k$ partition function, a.k.a.\ Witten--Reshetikhin--Turaev (WRT) invariant, and~$\widehat{Z}$’s can then be seen as a~finite-dimensional $S$-transform
\begin{equation}\label{eq:Zhat-WRT}
Z_{\text{CS}}(M_3;k)
 =\bigl({\rm i}\sqrt{2k}\bigr)^{ b_{1}(M_{3})-1}
 \sum_{a,b\in\operatorname{Tor}H_{1}(M_{3})/\mathbb Z_{2}}
 \mathrm{e}^{2\pi {\rm i} k \ell k(a,a)} S_{ab}
 \widehat{Z}_{b}(M_3;q)\Big|_{q = \mathrm{e}^{\frac{2\pi {\rm i}}{k}}}.
\end{equation}
Here $b_1(M_3)$ is the first Betti number, and the summation on the right side is over abelian flat connections on $M_3$, which are one-to-one with the boundary conditions of $T[M_3]$.
The authors of~\cite{gukov2020bps} also supplied a~purely mathematical construction of $\zhat_{b}$ whenever $M_{3}$ is a~negative definite plumbed manifold, i.e., it is presented by a~weighted graph $\Gamma$ whose adjacency plumbing matrix~$M$, obtained by adding the weights of the nodes of $\Gamma$ to the diagonal, is negative definite.

Various structural properties of $\widehat{Z}$ invariants (as well as generalizations thereof) for such manifolds were further studied in~\cite{cheng20193d, Cheng:2024vou, Cheng:2023row, Gukov:2023srx, park2020higher, Ri:2022bxf}.
Let $L=\#\!\operatorname{Vert}(\Gamma)$ and $\deg(v)$ degree of a~vertex $v\in\operatorname{Vert}(\Gamma)$; one obtains for each boundary condition $b \simeq (2\operatorname{Coker}M+\delta)/\mathbb Z_{2}$ the convergent $q$-series
\begin{equation}
\label{eq:zhat-positive-plumbing}
\widehat{Z}_{b}(M_3;q)=q^{-\frac{3L+\sum_{v}a_{v}}{4}}
\underset{\text{v.p.}}{\oint_{|z_{v}|=1}}
 \prod_{v}\frac{{\rm d} z_{v}}{2\pi {\rm i} z_{v}}
\bigl(z_{v}-z_{v}^{-1}\bigr)^{ 2-\deg(v)}
\Theta_{-M,b}(z;q),
\end{equation}
where the lattice theta-function is
\smash{$\Theta_{-M,b}(z;q)=\sum_{\ell\in 2M\mathbb Z^{L}+b}
 q^{-\frac{(\ell,M^{-1}\ell)}{4}}
 \prod_{i=1}^{L} z_{i}^{\ell_{i}}$}, $a_v$ is the weight associated with vertex $v$ and $q=\mathrm{e}^{2\pi {\rm i}\tau}$.
Crucially, this series is not convergent when the manifold in question is a~positive definite plumbed manifold.

$\zhat$ invariants can thus be thought of as a~$q$-series refinement of WRT invariants. In particular, when $M_3$ is a~negative definite plumbed homology sphere, the linear combination of radial limits of $\widehat{Z}_b(M_3;\tau)$ at roots of unity recovers the corresponding WRT invariant (in other words, the equation~\eqref{eq:Zhat-WRT} holds, as shown in~\cite{Murakami:2022cjk}).
One motivation for this refinement is due to the search for a~categorified invariant of a~$3$-manifold, rooted in the categorification programme initiated by Crane--Frenkel in four-dimensional TQFT~\cite{crane1994four}.
Since the discovery of $\zhat_{b}$, notable progress on categorified invariants of 3-manifolds has been achieved (e.g., in~\cite{akhmechet2023lattice,akhmechet2025knot,Gukov:2024vbr,Moore:2025pvq}) where new categorified invariants have arisen from certain deformations of $\zhat_{b}$.
The construction of a~proper mathematical definition for the $\zhat_{b}$ for a~general 3-manifold remains an open problem and serves as a~broader motivation for our work.

Continuing the line of research of~\cite{gukov2020bps}, Cheng, Chun, Ferrari, Gukov and Harrison~\cite{cheng20193d} approached the positive-plumbing problem by combining results from modularity and intuition coming from the relation between the $\zhat_{b}$ and Chern--Simons theory.
From the modularity point of view, the authors of~\cite{cheng20193d} show that for the simplest non-trivial negative plumbings -- star-shaped graphs giving Seifert homology spheres with three singular fibers -- the $\zhat_{b}$ is explicitly written as a~linear combination of false theta-functions, Eichler integrals of unary theta functions.
These functions are known to be quantum modular forms, modern extensions of classical modularity~\cite{zagier2010quantum}.
In Chern--Simons theory, orientation reversal sends the level $k$ to its negative
\[
Z_{\text{CS}}(M_3;k) = Z_{\text{CS}}\bigl(\overline{M_3};-k\bigr).
\]
Therefore, via
$
q = \ex^{2\pi {\rm i}\tau}=\ex^{2\pi {\rm i}/k}$,
orientation reversal should implement the inversion $q\leftrightarrow q^{-1}$.
However, applying this transformation naively to the plumbing formula~\eqref{eq:zhat-positive-plumbing} generically leads to divergent series.
Nonetheless, the authors of~\cite{cheng20193d} show that for some manifolds the $\zhat$ invariants can be rewritten as $q$-hypergeometric sums that converge on both sides of the unit circle, thus admitting a~well-defined~${q\leftrightarrow q^{-1}}$ transformation.
Consider, for example, the Poincar\'e homology sphere $\Sigma(2,3,5)$.
Its $\zhat$ invariant can be shown to be equal (up to a~factor of $q$) to a~linear combination of false theta functions
\[
\zhat_0\bigl(\Sigma(2,3,5);\mathrm{e}^{2\pi {\rm i}\tau}\bigr) \sim \tilde{\theta}^1_{30,1}(\tau) + \tilde{\theta}^1_{30,11}(\tau) + \tilde{\theta}^1_{30,19}(\tau) + \tilde{\theta}^1_{30,29}(\tau).
\]
The authors of~\cite{cheng20193d} show that this can also be written in a~$q$-hypergeometric form
\[
\zhat_0(\Sigma(2,3,5);q) \sim q^{\frac{1}{120}}\left(2-\sum_{n\geq 0} \frac{(-1)^n q^{\frac{n(3n-1)}{2}}}{\bigl(q^{n+1};q\bigr)_n}\right).
\]
The hypergeometric series converges for both $\abs{q}<1$ and $\abs{q}>1$.
Therefore, following the physical expectation from Chern--Simons theory, we can apply $q\leftrightarrow q^{-1}$ transformation. Remarkably, it leads to a~relation to the order 5 mock theta function of Ramanujan $\chi_0(q)$
\begin{gather*}
q^{\frac{1}{120}}\left(2-\sum_{n\geq 0} \frac{(-1)^n q^{\frac{n(3n-1)}{2}}}{\bigl(q^{n+1};q\bigr)_n}\right) \xleftrightarrow{ \quad q \leftrightarrow q^{-1}} q^{-\frac{1}{120}}\left(2-\sum_{n\geq 0} \frac{ q^{n^2}}{\bigl(q^{n+1};q\bigr)_n}\right)
= q^{-\frac{1}{120}}(2-\chi_0(q)).
\end{gather*}

Examples like these, where linear combinations of false theta functions are exchanged by mock theta functions, shed light on the modular structure of the $q$-series associated to a~3-manifold upon orientation reversal. The above had motivated the false-mock conjecture for $\zhat$ invariants.

\begin{coj}[the false-mock conjecture \cite{Cheng:2024vou,cheng2020three}]\label{con:False-Mockv2}
Let $M_3$ be a~$3$-manifold for which the~$\zhat$ invariants take the form
$
 \zhat_b\bigl(M_3;\mathrm{e}^{2\pi {\rm i}\tau}\bigr) = \mathrm{e}^{2\pi {\rm i}\tau c} \bigl(\tilde{\vartheta}(\tau)+p(\tau)\bigr)$,
where $\tau\in\mathbb{H}^+$, $c\in \mathbb{Q}$, $\tilde{\vartheta}(\tau)$ is the Eichler integral of a~unary theta function $\vartheta(\tau)$ of weight~$3/2$ and $p(\tau)$ is a~polynomial in~${q=\mathrm{e}^{2\pi {\rm i}\tau}}$. Then the $\zhat$ invariant of the manifold with reversed orientation is
$
 \zhat_b\bigl(\overline{M_3};\mathrm{e}^{2\pi {\rm i}\tau}\bigr) = \mathrm{e}^{-2\pi {\rm i}\tau c} (f(\tau)+p(-\tau))$,
where $f(\tau)$ is a~weight $1/2$ weakly holomorphic $($mixed$)$ mock modular form.
Moreover,
its completion $\hat f$, transforming as a~modular form of weight $1/2$ under a~certain congruence subgroup of~${\rm SL}_2(\ZZ)$, has the following properties. It is given by
\[
\hat f = f - \vartheta^\ast - \sum_{i\in I} g_i \vartheta_i^\ast,
\]
where $I$ is a~finite set, $\vartheta_i$ is a~theta function, and $g_i$ with any $i\in I$ is a~modular function for some discrete subgroup of ${\rm SL}_2(\ZZ)$ that either vanishes or has an exponential singularity at all cusps.\looseness=1
\end{coj}
The false-mock conjecture predicts the modular properties of $\zhat$ invariants for positively-definite plumbed manifolds and provides one possible way to constrain the space of the invariants of positively plumbed Brieskorn spheres through modularity.

Let $p_1, p_2, p_3 \in \mathbb{N}_+$ be pairwise coprime integers, and denote by $\overline{\Sigma(p_1,p_2,p_3)}$ the Brieskorn homology sphere represented by a~positive definite plumbing.\footnote{In this paper, we use Brieskorn spheres, Brieskorn homology spheres and integer homology spheres interchangeably.}
In this paper, we consider Brieskorn spheres of the class $\overline{\Sigma(s,t,rst\pm 1)}$.
These are obtained by performing $+1/r$-surgery on torus knot complements $S^{3}\setminus T_{s,\mp t}$, where $T_{s,\mp t}$ denotes a~torus knot with $s$ twists around the meridian and $t$ twists around the longitude ($-t$ indicates negative twists).
With the exception of Section~\ref{sec:ceff on the negative side}, all the manifolds we consider are integer homology spheres, which admit a~unique abelian flat connection and therefore a~single $\widehat{Z}$ invariant $\widehat{Z}_0(q)$.
In what follows, to simplify the notation we will omit the subscript, simply writing $\widehat{Z}$ in place of $\zhat_{b}(q)$.

For $\overline{\Sigma(s,t,rst\pm 1)}$, the $\zhat$ invariants admit a~conjectural expression, obtained from their surgery description via the $F_K$ invariant of 3-manifolds with boundary~\cite{gukov2021two}.
This expression is known as the regularized $+1/r$-surgery formula, suggested by Park~\cite[Conjecture 5]{park2021inverted}, whose details we recall in Section~\ref{sec:preliminaries}.
The regularized $+1/r$ formula is also known to have a~close connection to mock modularity and regularized indefinite theta series.
In particular, Cheng, Coman, Kucharski, Sgroi and Passaro~\cite{Cheng:2024vou} leveraged the false-mock conjecture~\cite{cheng2020three} to recast the expression of the regularized $+1/r$-surgery formula in a~linear combination of indefinite theta series, well-known mock modular forms~\cite{Zwegers2008MockTF}.

Parallel developments on $\widehat{Z}$ focused on resurgence and Borel resummation.
An important conjecture about $\widehat{Z}$ for any 3-manifold is the trans-series expansion conjecture~\cite{Gukov:2016njj}.
It states that $\widehat{Z}$ can be obtained by Borel resummation of finitely many perturbative contributions, each associated with a~saddle of the Chern--Simons path integral
\begin{equation}\label{eq:zhat trans-series expansion}
\widehat{Z}\bigl(M_3;\mathrm{e}^{-\hbar}\bigr) \stackrel{\mathrm{resum.}}{=}
\sum_{\alpha\in\mathcal{M}_{\rm flat}} n_{\alpha} \mathrm{e}^{-\frac{2\pi {\rm i}}{\hbar} \widetilde{\mathrm{CS}}_{\alpha}(M_3)} Z^{\alpha}_{\rm pert}(\hbar),
\end{equation}
where \smash{$\hbar:=-\frac{2\pi {\rm i}}{k}$} and \smash{$Z^{\alpha}_{\rm pert}(\hbar)$} denotes the perturbative expansion of the Chern--Simons path integral around the non-abelian flat connection $\alpha$ with lifted Chern--Simons invariant ${\widetilde{\mathrm{CS}}_{\alpha}(M_3)\in\mathbb{C}}$, for some choice of a~lift of $\alpha$ to the universal cover of the space of all flat connections $\mathcal{M}_{\mathrm{flat}}$ (see~\cite{Costin:2023kla} for details).

The moduli space of flat $\mathrm{SL}(2,\mathbb{C})$ connections on Seifert homology spheres (of which Brieskorn spheres are a~special case) is well understood~\cite{andersen2020resurgenceanalysisquantuminvariants,boden2006sl,fintushel1990instanton, kirk1991representation,munoz2025counting}.
Denote $[x] = x+ \mathbb{Z}$ for $x\in\mathbb{Q}$.
Then the set of Chern--Simons invariants
is given by
\begin{equation}\label{eq:CS invariants for Brieskorn spheres}
 \mathrm{CS}\bigl(\overline{\Sigma(p_1,p_2,p_3)}\bigr) = \{[0]\}\cup \left\{\left[\frac{m^2}{4p_1p_2p_3}\right]\mid \text{$m$ not divisible by any of $p_1$, $p_2$, $p_3$}\right\}.
\end{equation}

The main characteristic of $\widehat{Z}$ invariant that we focus on is the growth rate of its coefficients, which we encode in a~quantity denoted by $\ceff$.
Physically, $\ceff$ corresponds to the notion of effective central charge for 2d CFTs given by the Cardy formula~\cite{cardy1986operator}, and was recently generalised to~3d~${\mathcal{N}=2}$ theories by Gukov and Jagadale in~\cite{Gukov:2023cog}.
It can also be interpreted as a~measure of entanglement~\cite{karch2024universal}.

In our case, we consider $\ceff$ of 3d $\mathcal{N}=2$ theory $T[M_3]$ associated to a~3-manifold.
The $\ceff$ encodes the asymptotic growth of the coefficients of \smash{$\widehat{Z}(M_3;q) = \sum_{n} a_nq^n$}
\begin{equation}\label{eq:general c_eff}
 a_n \sim \Re \left[ \exp\left( \sqrt{
 \frac{2\pi^2}{3}c_{\rm eff}n
 }+2\pi {\rm i} \omega n \right) \right]
\end{equation}
for some $\omega\in \mathbb{Q}$.
While the general meaning of $\ceff$ for a~3-manifold and its connection to known topological invariants largely remains mysterious, it is expected that $\ceff$ should be associated with a~distinguished non-abelian complex flat connection on $M_3$, when $M_3$ is an integer homology sphere~\cite{Gukov:2023cog}.
From the resurgent structure, the trans-series conjecture of $\widehat{Z}$ invariant~\cite{Gukov:2016njj} and expression~\eqref{eq:CS invariants for Brieskorn spheres}, it is expected that
\begin{equation}\label{coj:z-hat at leading order}
 \ceff\big[\zhat\bigl(\brs{s}{t}{rst\pm 1};\tau\bigr)\big] = 24\left(\frac{m^2}{4st(rst\pm 1)}-l\right),
\end{equation}
for some integers $l$ and $m$ such that $m$ is not divisible by any of $s$, $t$ and $rst\pm1$.

This expectation has been confirmed for $\overline{\Sigma(2,3,5)}$ and $\overline{\Sigma(2,3,7)}$ in~\cite{Gukov:2023cog} using the modular properties of $\widehat{Z}$, assuming the regularized $+1/r$ surgery conjecture. We aim to investigate this for a~larger class of Brieskorn spheres.

\subsection*{A brief summary of the results}

The goal of this paper is to clarify the above conjectures by studying the compatibility of the regularized $+1/r$-surgery formula for $\widehat{Z}$ invariants with the trans-series expansion~\eqref{eq:zhat trans-series expansion} and the expression for $\ceff$ \eqref{coj:z-hat at leading order}, as well as $q$-series arising from resurgence. We aim to investigate in detail instances of Brieskorn spheres and to verify whether the $c_{\text{eff}}$ computed using the $+1/r$ surgery-indefinite theta series proposal can be expressed in terms of some Chern--Simons invariant of a~flat connection. Our main result is the following theorem.
\begin{Theorem}\label{thm:conjecture fails}
 The conjectural regularized $+ \frac{1}{r}$-surgery formula for $\widehat{Z}$ and the expected form of~$\ceff$ \eqref{coj:z-hat at leading order} are generically incompatible with each other.
\end{Theorem}
The first such inconsistency appears for $\overline{\Sigma(3,4,13)}$, for which we find the exact values $l=0$ and \smash{$m=3\sqrt{\frac{13}{5}}$}, which do not correspond to any Chern--Simons invariant of a~flat connection.

Another noteworthy outcome emerged alongside a~parallel study of
$\ceff$ by Adams, Costin, Dunne, Gukov, and \"Oner~\cite{ACDGO}, which aimed to compute $\ceff$ using the resurgent analysis framework proposed in~\cite{Costin:2023kla}. A~direct comparison between their results and ours shows that the two approaches generally produce different values of $\ceff$.

Our take on $\widehat{Z}$ invariants and $\ceff$ can be summarized via the triangle in Figure~\ref{fig:three-part comparison}.

\begin{figure}
\begin{center}
\begin{tikzcd}[column sep=0]
	& {\text{\textsc{Modularity}}} \\
	{\text{\textsc{Resurgence}}} && {\text{\textsc{Numerics}}}
	\arrow[tail reversed, from=1-2, to=2-1]
	\arrow[tail reversed, from=1-2, to=2-3]
	\arrow[tail reversed, from=2-3, to=2-1]
\end{tikzcd}
\end{center}
\caption{The study of $\ceff$ via numerical analysis, modularity and resurgence will shed light on the existing conjectures about \smash{$\widehat{Z}$} on the other side. For us, Brieskorn spheres \smash{$\overline{\Sigma(s,t,rst\pm 1)}$} are the most suitable class of examples, since they can be studied from various angles independently.}
\label{fig:three-part comparison}
\end{figure}
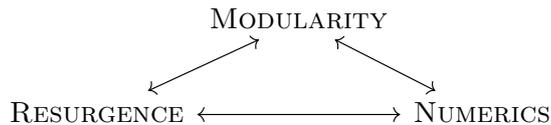

In order to study the other side of Brieskorn spheres, one can use three independent approaches: mock modularity, resurgence, and numerical analysis.
When combined, all three provide a~very powerful set of tools at our disposal, which will shed light on the general structural properties of $\widehat{Z}$ invariants.
The two cases $\overline{\Sigma(s,t,rst- 1)}$ and $\overline{\Sigma(s,t,rst+ 1)}$ turn out to be quite different.
The $\zhat$ invariants of $\overline{\Sigma(s,t,str+1)}$ obtained by the regularized $+1/r$ formula have expressions as mock modular forms~\cite{Cheng:2024vou}.
In this case, $\ceff$ is computed by using mock modularity properties of $\widehat{Z}$.
On the other hand, mock modular expressions for the same invariants of $\overline{\Sigma(s,t,str-1)}$ are not known.
In these cases, we rely on estimates for the bounds of $\ceff$ and we can also compare these bounds with $\ceff$ found from resurgence~\cite{ACDGO}.
 This approach turns out to be sufficient to provide counterexamples to~\eqref{coj:z-hat at leading order}.

The paper is organized as follows. After recalling preliminary notions and necessary conjectures in Sections~\ref{sec:preliminaries} and~\ref{sec:Bounds on ceff}, we
establish an upper bound of $\ceff$ for our class of Brieskorn spheres.
We find that this upper bound comes entirely from the regularization factor in the conjectural $+1/r$-surgery formula and thus depends only on the surgery slope, but not on the choice of a~knot.
In Section~\ref{sec:Numerical analysis of ceff}, we provide the numerical analysis of $\widehat{Z}$ coming from the regularized $+1/r$ formula. First, we analyze the structure of singularities of $\widehat{Z}$ on the unit circle and conjecture that for $+1$-surgeries they are situated at $q=1$. Second, we describe a~numerical method of finding the lower bounds for $\ceff$, and combine this with the results from Section~\ref{sec:Bounds on ceff}.
In Section~\ref{sec:Mixed mock modularity and resurgence}, we proceed with the mock modular analysis of $\widehat{Z}$ for positive definite manifolds. We define the modular completion of $\widehat{Z}$ invariants by combining the modular completions of the respective summands. We then provide a~comprehensive asymptotic analysis of every part of a~modular completion of $\widehat{Z}$ in order to extract information about $\ceff$ and to compare with the trans-series expansion conjecture.
In Section~\ref{sec:ceff on the negative side}, we provide proofs for $\ceff$ on the negative definite side, which nicely complements the above results and makes the overall picture more coherent.
Finally, in Section~\ref{sec:Case studies}, we investigate the concrete examples of Brieskorn spheres and apply all three techniques in order to compute and compare~$\ceff$.

\section{Preliminaries}\label{sec:preliminaries}

We begin by discussing some preliminaries.
First, we introduce the general notion of $\ceff$ for a~$q$-series, which measures the growth rate of its coefficients.
Second, we review the notions of mock modularity and the properties of indefinite theta series which we will leverage. Third, we recall the conjectural expression for the $\widehat{Z}$ invariant obtained from the regularized surgery formula and review its modular properties in the case of reverse-oriented Brieskorn spheres.
In the following sections, these will enable us to produce the exact results on $\ceff$ and to address the main points of this paper.

\subsection[ceff for q-series and the basic example]{$\boldsymbol{\ceff}$ for $\boldsymbol{q}$-series and the basic example}

Let $P(q)\in \ZZ[[q]]\setminus\{0\}$ be a~formal non-zero $q$-series with integer coefficients. The following is a~slightly revised definition from~\cite{Gukov:2023cog}.
\begin{Definition}[$\ceff$ for $q$-series]\label{ceff-func}
 Let
 \[
 P(q) = \sum_{n\geq 0} a_n q^{n}, \qquad a_n\in \ZZ,
 \]
 such that $\lim_{n\to \infty}a_n \neq 0$.
 We define
 \[
 \ceff(P(q)) \coloneq \frac{3}{2\pi^{2}}\limsup_{n\to\infty}\frac{(\log|a_n|)^{2}}{n}.
 \]
\end{Definition}
In other words, $\ceff$ measures the growth rate of the coefficients of a~$q$-series. According to~\cite[Conjecture 1.2]{Gukov:2023cog}, if $P(q)$ is a~supersymmetric index of a~three-dimensional $\mathcal{N} = 2$ theory, $\ceff(P(q))$ measures the growth of degeneracies of supersymmetric (BPS) states, and is therefore analogous to the notion of effective central charge in 2d CFTs.
 A~class of $q$-series we consider in this paper comes from complex-valued functions $f(q)$ defined on the unit disk. In order to formulate the precise correspondence between $\ceff$ and the asymptotic behavior of $f(q)$, we first need to understand its singularities.

\begin{Definition}\label{dfn:dominant singularity}
Let $f(q)$ be a~complex function defined in $|q|\leq 1$ such that $f(q)$ has essential singularities on the unit circle and is analytic inside the unit disc.
We say $q= \mathrm{e}^{2 \pi {\rm i} \omega}$ is a~dominant singularity of $f(q)$ if for all $\phi \in \mathbb{Q}/\mathbb{Z}$,
\[
 \limsup_{ \hbar \rightarrow 0^{+} } \big[ \hbar \log|f\bigl(\mathrm{e}^{- \hbar + 2 \pi {\rm i} \omega }\bigr)| - \hbar \log|f\bigl(\mathrm{e}^{- \hbar + 2 \pi {\rm i} \phi }\bigr)| \big] \geq 0.
\]
\end{Definition}

Using this notion, we can relate $\ceff$ to the growth rate of the coefficients of \smash{$f(q)\stackrel{q\to 0}{=}\sum a_n q^n$}. Namely, assume that $f(q)$ has a~dominant singularity at $q=\mathrm{e}^{-2\pi {\rm i}\omega}$ for some $\omega\in\mathbb{Q}\backslash \mathbb{Z}$. Then, as was argued in~\cite{Gukov:2023cog} using the saddle point analysis,
\[
 a_n \sim \Re \left[ \exp\left( \sqrt{
 \frac{2\pi^2}{3}c_{\rm eff}n
 }+2\pi {\rm i} \omega n \right) \right].
\]

\begin{Remark}
 Note that in Definition~\ref{ceff-func} we use a~simplifying assumption that $\ceff\in\mathbb{R}$,
 which is less general than the one from~\cite{Gukov:2023cog}. However, this assumption holds for $q$-hypergeometric functions whose essential singularities are at roots of unity, in particular, for the regularized~$+1/r$-surgery conjecture for
 $\widehat{Z}\bigl(\overline{\Sigma(s, t, st\pm 1)};q\bigr)$. One can shift the dominant singularity to~${q=1}$ by a~suitable normalization.
\end{Remark}

\begin{Remark}\label{rem:discrepancies}
Recall that in general $\widehat{Z}_b(M_3;q) \in 2^{-c}q^{\Delta_b}\mathbb{Z}[[q]]$ where $\Delta_b$ is a~rational number defined as the minimal $q$-exponent of $\widehat{Z}_b$ and $c \geq 0$ is an integer. When $M_3 = \overline{\Sigma(s, t, st\pm 1)}$ we have $c=0$, Section~\ref{sec:ceff on the negative side} on the other hand contains instances when $c$ is non-zero. To maintain consistency with Definition~\ref{ceff-func}, if $c \geq 0$ is an integer, $\Delta_b \in \mathbb{Q}$ and $P(q) \in \mathbb{Z}[[q]]$ we shall write $\ceff\bigl(2^{-c}q^{\Delta_b}P(q)\bigr)$ to mean $\ceff(P(q))$ throughout this paper. This allows us to consider~\smash{$\ceff\bigl(\widehat{Z}_b(M_3;q)\bigr)$} and \smash{$\ceff\bigl(\widehat{Z}_b(M_3;q)/(q)_\infty\bigr)$}.
\end{Remark}

\begin{Example}\label{exmp:q-pochhammer ceff}
 Consider a~prototype example which will be relevant for us in the following. Define the infinite $q$-Pochhammer symbol as
 \[
 (x;q)_\infty = \prod_{k=0}^{\infty}\bigl(1-xq^k\bigr).
 \]
 Denote $(q)_{\infty} := (q;q)_{\infty}$, and take the function $f(q)=(q)_{\infty}^{-1}$,\footnote{From the CFT/VOA perspective, this function gives the character of a~free boson VOA. It is also related to Dedekind $\eta$ function by a~factor: \smash{$\eta(q)=q^{\frac{1}{24}}(q)_\infty$}.} whose series expansion around the origin is given by
 \[
 \frac{1}{(q)_{\infty}} \stackrel{q\to 0}{=} 1 + q + 2 q^2 + 3 q^3 + 5 q^4 + 7 q^5 + \cdots.
 \]
 Note that $f(q)$ has essential singularities at all roots of unity and is analytic inside the unit disc $|q|<1$.
 By analyzing the radial limits of $\log|f(q)|$ as $q$ approaches various roots of unity and applying Definition~\ref{dfn:dominant singularity}, we deduce that the location of the most dominant singularity is at~${q=1}$.

 Asymptotic analysis of its coefficients can be carried out easily. Let $q=\mathrm{e}^{-\hbar}$ and consider~${\hbar \to 0^{+}}$. We can write
 \[
 \frac{1}{(x;q)_{\infty}} =
 \exp \left( \frac{1}{\hbar}\sum_{k=0}^{\infty}\frac{(-1)^kB_k\hbar^k}{k!}\mathrm{Li}_{2-k}(x) \right), \qquad |x|<1,
 \]
 where $B_k$ is the $k$-th Bernoulli number and $\operatorname{Li}_2(z)$ is the dilogarithm function
 \[
 \operatorname{Li}_{s}(z) = \sum_{k=1}^{\infty}\frac{z^{k}}{k^{s}}.
 \]

 By substituting $x=\mathrm{e}^{-\hbar}$, we get
 \[
 \frac{1}{(q)_{\infty}} \sim \exp \left( \frac{1}{\hbar} [\mathrm{Li}_2(1)+O(\hbar)] \right).
 \]
 Noting that $\mathrm{Li}_2(1)=\frac{\pi^2}{6}$ and using the relationship between the asymptotics of a~$q$-series and the growth rate of its coefficients,
 \[
 f(q) \sim \exp \left(\frac{C^2}{4h}+O(h) \right) \quad \iff \quad a_n\sim \exp \left( C\sqrt{n} \right),
 \]
 we finally obtain
 \[
 a_n \sim \exp\left( \sqrt{
 \frac{2\pi^2}{3}n
 }\right),
 \] i.e., $\ceff\left((q)_{\infty}^{-1}\right)=1$.
 We can also plot $\log|a_n|$ against $n$ and see how fast the asymptotic value is approached, Figure~\ref{fig:Denominator}.
 \begin{figure}
 \centering \includegraphics[width=0.75\linewidth]{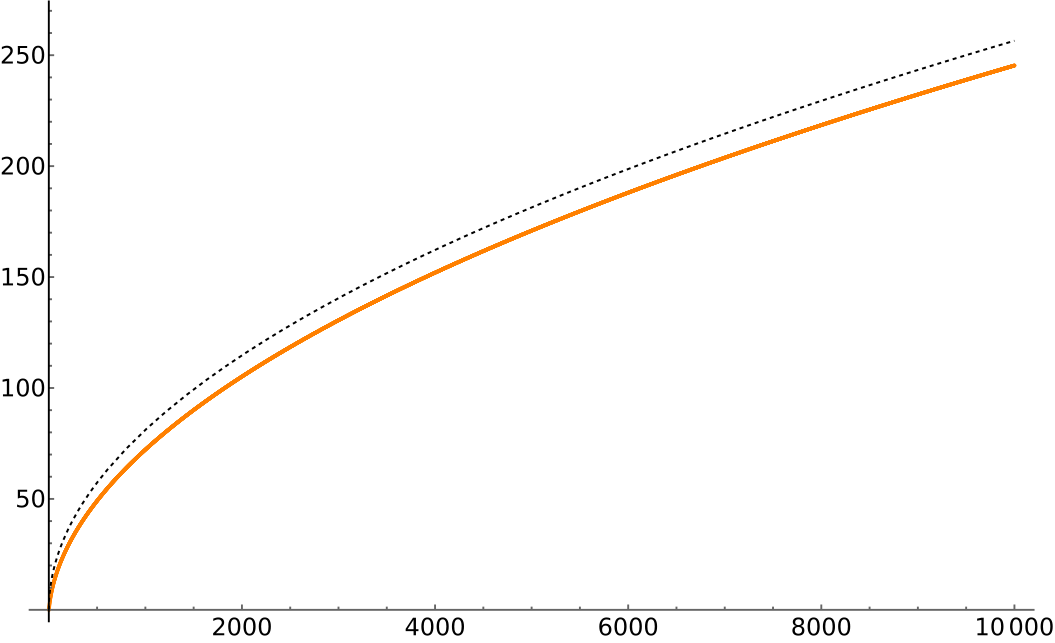}
 \caption{The growth $\log |a_n|$ shown by the \textcolor{orange}{orange} line, and $\log$ of the asymptotic, \smash{$\sqrt{
 \frac{2\pi^2}{3}n
 }$} shown by dashed line, for the first $10^4$ coefficients of the function $(q)_{\infty}^{-1}$.}
 \label{fig:Denominator}
 \end{figure}
\end{Example}

\subsection{Mock modularity}

Modular forms have featured predominantly in the fields of mathematics and physics in the past century, appearing in many, seemingly disconnected areas.
In mathematics, for example, modular forms play a~crucial role in the study of elliptic curves.
In this context, the Taniyama--Shimura--Weil conjecture (now a~theorem) states that every rational elliptic curve is associated with a~modular form~\cite{knapp2000forms}. This theorem was instrumental in the proof of Fermat's last theorem by Andrew Wiles.
Modular forms play a~key role also in representation theory.
The famous Monstrous Moonshine conjecture states that the coefficient of the j-invariant, a~modular function, encodes information about the representations of the Monster group, the largest of the sporadic simple groups; this connection was proven by Richard Borcherds using vertex operator algebras~\cite{Borcherds:1992jjg}.
In theoretical physics modular forms appear in the computation of black hole entropy, which in some contexts is given by the Cardy formula, involving modular forms like the Dedekind $\eta$ function and Jacobi forms~\cite{Chattopadhyaya:2021rdi,deBoer:2006vg,Goldstein:2020yvj,karch2024universal,Korpas:2017qdo}.

It is well known that modular forms are complex functions of the upper half complex plane~$\mathbb{H}$ which benefit from a~certain symmetry by the action of congruence subgroups of $\mathrm{SL}_{2}\left( \mathbb{Z} \right)$.
This symmetry reflects the discrete symmetries of the upper half complex plane and it is often referred to as modular symmetry.
Given the wide application of modular forms, it is interesting to investigate occurrences where the modular symmetry of modular forms is broken in meaningful ways.

One class of functions where the modular symmetry is ``minimally'' broken is that of mock modular forms.
The first examples of mock modular forms were discovered by Ramanujan and are the primary focus of his last notebook.
Mock modular forms have a~fascinating history in that examples of mock modular forms were known before a~proper definition was established.
To this day, Ramanujan's classification of mock modular forms into mock theta function orders remains largely mysterious.

 A~concrete definition of mock modular forms was only given in 2002 by Zwegers in his PhD thesis~\cite{Zwegers2008MockTF}. His groundbreaking work showed how mock theta functions fit into the broader context of modular forms. A~key aspect of mock modular forms is the trade-off between a~modular anomaly and a~holomorphic anomaly.
Mock modular forms are holomorphic functions of the upper half plane which are not modular, but which admit a~modular ``completion'' -- the addition of a~non-holomorphic Eichler integral of a~modular form.
In working with mock modular forms, then one needs to balance the modular anomaly of the mock modular form with the holomorphic anomaly of the completion. More specifically, we will take mock modular forms to be the following.

\begin{Definition}\label{defn:mock modular form}
 Let $k\in \frac{1}{2}\mathbb{Z}$. A~mock modular form $h(\tau)$ of weight $k$ is the first member of a~pair of functions $(h,g)$, where
 \begin{enumerate}\itemsep=0pt
 \item[(1)] $h$ is a~holomorphic function of the upper complex half plane with at most exponential growth at all cusps.
 \item[(2)] The function $g( \tau)$, called the \emph{shadow} of $h$, is a~holomorphic modular form of weight $2-k$.
 \item[(3)] The sum $\hat{h}=h+\tilde{g}^{*}$, called the completion of $h$, is a~non-holomorphic modular form of weight $k$, where $\tilde{g}^{*}$ is the non-holomorphic Eichler integral of the shadow.
 \end{enumerate}
\end{Definition}
Here we take the non-holomorphic Eichler integral $\tilde{g}^{*}( \tau)$ of a~modular form $g( \tau)$ of weight $2-k$ to be the solution of
\[
 ( 4\pi \operatorname{Im}(\tau) )^{k}\frac{\partial \tilde{g}^{*}}{\partial\bar{\tau}}= - 2\pi {\rm i} \overline{g( \tau)}.
\]
If the modular form $g( \tau)$ has expansion $g( \tau)= \sum_{n>0}b_{n}q^{n}$, then we fix the solution of the non-holomorphic Eichler integral to be
\[
 \tilde{g}^{*}( \tau)= \bar{b}_{0}\frac{( 4\pi \operatorname{Im}(\tau) )^{1-k}}{k-1}+\sum_{n>0}n^{k-1} \bar{b}_{n}\Gamma( 1-k, 4\pi n\tau_{2} )q^{-n}.
\]
Here $b$'s are the coefficients of $g(\tau)$, implicitly defined in the line above.
 A~common choice is to restrict $g( \tau)$ to the space of cusp forms, in which case the Eichler integral can also be defined as%
\[
 \tilde{g}^{*}( \tau) = \left( \frac{\rm i}{2\pi} \right)^{k-1}\int_{-\bar{\tau}}^{\infty}( z + \tau )^{-k}\overline{g( -\overline{z} )}{\rm d}z.
\]

Notably, mock modular forms are not closed under multiplication by a~modular form. That is, if $(f,g)$ is a~pair $(\text{mock modular form}, \text{shadow})$ and $\eta$ is a~modular form, then, to construct the completion of $\eta(\tau)f(\tau)$ one would need to multiply the shadow contribution by $\eta(\tau)$ to get something modular.
On the other hand, mixed mock modular forms generalize mock modular forms in such a~way as to be closed under multiplication by modular forms of a~given weight $\ell$.
\begin{Definition}\label{defn:mixed mock modular form}
 Let $k,\ell\in\tfrac12\ZZ$. A~mixed mock modular form $h(\tau)$ of weight $k|\ell$ is the first member of the triplet $( h,\{ f_{0},\dots,f_{N} \},\{g_{0},\dots,g_{N}\} )$, where
 \begin{enumerate}\itemsep=0pt
 \item[(1)] $h$ is a~holomorphic function of the upper half plane with at most exponential growth at all cusps.
 \item[(2)] The functions $g_{i}( \tau)$ are holomorphic modular forms of weight $2-k+\ell$.
 \item[(3)] The functions $f_{i}( \tau)$ are holomorphic modular forms of weight $\ell$.
 \item[(4)] The sum \smash{$\hat{h}\coloneq h+\sum_{j}f_{j}\tilde{g}^{*}_{j}$}, called the completion of $h$, is a~holomorphic modular form of weight $k$.
 \end{enumerate}
\end{Definition}
To avoid confusion, we shall refer to classical mock modular forms as being ``pure'' mock modular forms.

\subsubsection{Indefinite theta series}\label{indefinietthetaseriessubsec}
Indefinite theta series are examples of mock modular forms.
In the following, we will leverage proposals for $\zhat$ invariants for positive definite plumbed manifolds which hinge on indefinite theta series.
To that end, we will briefly summarize their theory.

Theta series on negative definite lattices constitute a~large class of examples of holomorphic modular forms.
When the signature of the lattice is indefinite the sum over the whole lattice fails to converge and the sum must be regularized. A~useful regularization was found by G{\"o}ttsche and Zagier~\cite{gottsche1996jacobi} whereby the sum reduced to a~holomorphic theta series by restricting the summation to cones of signature $(1,n-1)$.
In Zwegers's thesis~\cite{Zwegers2008MockTF}, he was able to determine the modular properties of these objects by constructing a~modular completion and demonstrating these as examples of mock modular forms. A~more generic case of lattices of signature $(2,n-2)$ was studied by Alexandrov et al.\ \cite{Alexandrov_2018}.

Let $A$ be a~$2 \times 2$ symmetric, non-degenerate matrix with integer entries and signature $(1,1)$. This matrix gives us a~quadratic form $Q$ and a~bilinear form $B$, defined as follows,
\[
 Q\colon\ \ZZ^{2} \rightarrow \ZZ ,\qquad n \mapsto n^{\mathsf{T}} A n, \qquad B\colon\ \ZZ^{2} \times \ZZ^{2} \rightarrow \ZZ ,\qquad (x,y) \mapsto x^{\mathsf{T}}A y.
\]
The quadratic form $Q$ of signature $(1,1)$ is extended from $L=\ZZ^{2}$ to $L\otimes_\ZZ \mathbb{R}$ by
$
 \norm{x}^{2} \coloneq x^{\mathsf{T}} Ax$, $ x \in L\otimes_\ZZ \mathbb{R}$.
$Q$ splits the two-dimensional plane into two pairs of conical regions distinguished by the sign of~$Q(c)$ for $c\in\ZZ^{2}$ and separated by rays of zero norm (see Figure~\ref{fig:indefinite_signature} for an illustration).

\begin{figure}
 \centering\includegraphics[width=0.35\linewidth]{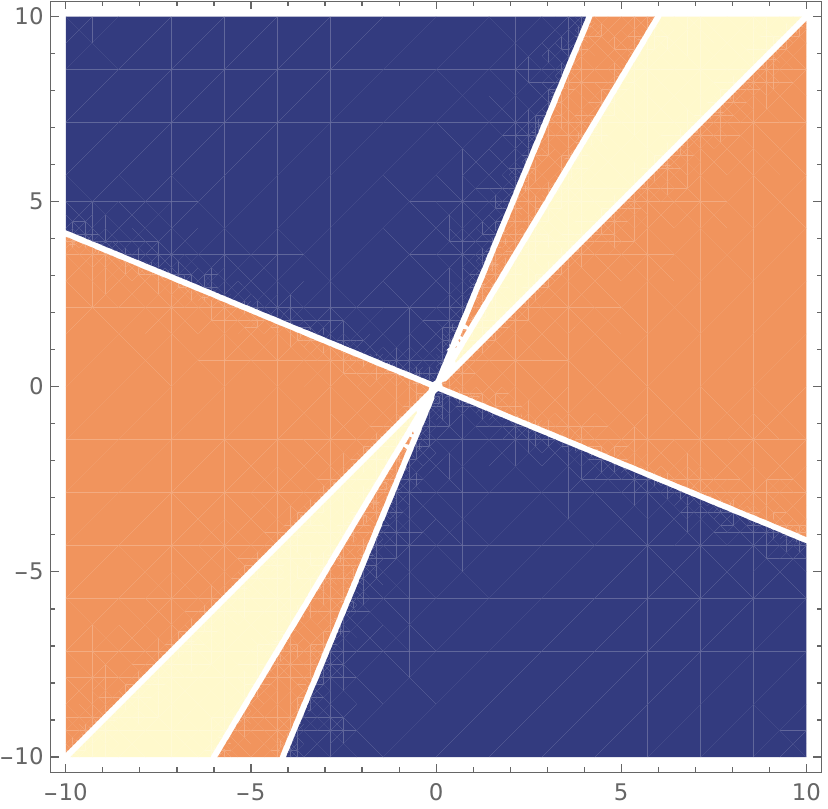}
 \caption{Splitting of the plane into conical regions determined by the indefinite quadratic form defined by \smash{$A=\bigl(\begin{smallmatrix}
 1 & 1 \\
 1 & -1
 \end{smallmatrix}\bigr)$} and the regularization factor $\rho^{\mathsf{c}_1, \mathsf{c}_2}$, with \smash{$\mathsf{c}_1=\bigl(\begin{smallmatrix} 0 \\ 1\end{smallmatrix}\bigr)$} and \smash{$\mathsf{c}_2=\bigl(\begin{smallmatrix} 1 \\ 4\end{smallmatrix}\bigr)$}. The two connected components with $Q(c)<0$ and $Q(c)>0$ are shown in blue and orange, respectively. The yellow shaded region is the region where \smash{$\rho^{\mathsf{c}_1, \mathsf{c}_2}(n)\neq 0$}.} \label{fig:indefinite_signature}
\end{figure}

Given a~vector $c_{0}$ with $Q(c_{0})<0$, we can specify the connected component containing $c_{0}$ as
\begin{equation}\label{eq:indef-theta-cone}
 C_{B}(c_0)= \big\{c \in \mathbb{R}^{2}\mid Q(c)<0,\, B(c,c_0)<0\big\}.
\end{equation}

Note that the resulting set of vectors depends only on the region where $c_0$ lies in, but not on its particular choice. For example, we get the upper (lower) blue cone in Figure~\ref{fig:indefinite_signature} when choosing~$c_0$ to be $(-1,1)$ ($(1,-1)$).

\begin{Definition}
For any $a, b \in \mathbb{R}^{2}$, $c_{1}, c_{2} \in C_{B}$,\footnote{In~\cite{Zwegers2008MockTF}, Zwegers defines the regularized indefinite theta function for slightly more general parameters $a$,~$b$,~$c_{1}$,~$c_{2}$. However, all the cases relevant to the present paper have parameters from the set specified above.} one can define regularized indefinite theta function
\[
 \Theta_{A,a,b,c_{1}, c_{2}}(\tau) \coloneq \sum_{n \in a~+ \ZZ^{2}} \rho^{c_{1},c_{2}}(n) \mathrm{e}^{2 \pi {\rm i} B(n,b)} q^{\frac{\norm{n}^{2}}{2}},
\]
where the regularization factor $\rho^{c_{1},c_{2}}(n)$ is defined as,
\[
 \rho^{c_{1},c_{2}}(n) \coloneq \operatorname{sgn}(B(c_{1},n)) - \operatorname{sgn}(B(c_{2},n)).
\]
\end{Definition}
The choice of regularization is borrowed from~\cite{park2021inverted}. The $q$-series $\Theta_{A,a,b,c_{1}, c_{2}}(\tau)$ is a~mixed mock modular form and it is convergent for $|q|<1$. $\Theta_{A,a,b,c_{1}, c_{2}}(\tau)$ is the holomorphic part of~\smash{$\widehat{\Theta}_{A,a,b,c_{1}, c_{2}}(\tau, \bar{\tau})$} defined below.

\begin{Example}
Choosing $A=\bigl(\begin{smallmatrix}1 & 1 \\ 1 & -1\end{smallmatrix}\bigr)$, and $\mathsf{c}_1 = \bigl(\begin{smallmatrix}
 0 \\ 1
\end{smallmatrix}\bigr)$ and $\mathsf{c}_2= \bigl(\begin{smallmatrix}
 1 \\ 4
\end{smallmatrix}\bigr)$, the summands in the indefinite theta series do not vanish when $n+a$ for $n\in \mathbb{Z}^2$ and $a \in \mathbb{R}^2$ are in the yellow shaded region in figure~\ref{fig:indefinite_signature}. With
\[
a=\begin{pmatrix}\tfrac{1}{2} \\ \tfrac{1}{3}\end{pmatrix}\qquad \text{and} \qquad b = \begin{pmatrix}0\\0\end{pmatrix},
\]
 the indefinite theta series can be expanded to
\[
\widehat{\Theta}_{A,a,b,c_{1}, c_{2}}(\tau, \bar{\tau}) = -2 q+2 q^2-2 q^7+2 q^8-2 q^{14}+2 q^{18}+O\bigl(q^{21}\bigr).
\]
\end{Example}

\begin{Definition}
We call $\widehat{\Theta}_{A,a,b,c_{1}, c_{2}}(\tau, \bar{\tau})$ the modular completion of $\Theta_{A,a,b,c_{1}, c_{2}}(\tau)$
\[
 \widehat{\Theta}_{A,a,b,c_{1}, c_{2}}(\tau, \bar{\tau}) = \sum_{n \in a~+ \ZZ^{2}} \hat{\rho}^{c_{1},c_{2}}(n, \tau) \mathrm{e}^{2 \pi {\rm i} B(n,b)} q^{\frac{\norm{n}^{2}}{2}},
\]
where, for $c_{1}, c_{2}, c\in C_{Q}$,
$
 \hat{\rho}^{c_{1},c_{2}}(n, \tau) \coloneq \hat{\rho}^{c_{1}}(n, \tau) -\hat{\rho}^{c_{2}}(n, \tau)$,
where
\[
 \hat{\rho}^{c}(n, \tau) \coloneq E\left(B(c,n) \sqrt{\frac{2 \operatorname{Im}(\tau)}{-|c|^{2}}} \right), \qquad
 E(z) \coloneq 2 \int_{0}^{z} \mathrm{e}^{- \pi u^{2}} \dd u.
\]
\end{Definition}

Using \cite[Lemma 1.7]{Zwegers2008MockTF}, we can write $\hat{\rho}^{c_{1},c_{2}}(n, \tau)$ as a~sum of three terms,
\begin{align*}
 \hat{\rho}^{c_{1},c_{2}}(n, \tau) ={}& \rho^{c_{1}, c_{2}}(n) - \operatorname{sgn}( B(c_{1},n)) \beta \left( \frac{2 B(c_{1},n)^{2} \operatorname{Im}(\tau)}{ - |c_{1}|^{2}} \right) \\
 &+ \operatorname{sgn}( B(c_{2},n) ) \beta \left( \frac{2 B(c_{2},n)^{2} \operatorname{Im}(\tau)}{ - |c_{2}|^{2}} \right),
\end{align*}
where for $x\geq 0$
 \[
 \beta(x) = \int_{x}^{\infty} u^{-\frac{1}{2}} \mathrm{e}^{- \pi u} \dd u.
 \]
This allows us to write the modular completion $\widehat{\Theta}_{A,a,b,c_{1},c_{2}}(\tau, \bar{\tau})$ as a~sum of $\Theta_{A,a,b,c_{1}, c_{2}}(\tau)$ and two other terms,
\[
\widehat{\Theta}_{A,a,b,c_{1},c_{2}}(\tau, \bar{\tau}) = \Theta_{A,a,b,c_{1}, c_{2}}(\tau) + G_{A,a,b,c_{1}}(\tau, \bar{\tau}) - G_{A,a,b,c_{2}}(\tau, \bar{\tau}),
\]
where
\[
 G_{A,a,b,c}(\tau, \bar{\tau}) \coloneq - \sum_{n \in a~+ \ZZ^{2}} \mathrm{e}^{2 \pi {\rm i} B(n,b)} \operatorname{sgn}( B(c,n)) \beta \left( \frac{2 B(c,n)^{2} \operatorname{Im}(\tau)}{ - |c|^{2}} \right) q^{\frac{\norm{n}^{2}}{2}}.
\]
The functions \smash{$G_{A,a,b,c_{1}}(\tau, \bar{\tau})$} and \smash{$G_{A,a,b,c_{2}}(\tau, \bar{\tau})$} are non-holomorphic functions that we need to add to $\Theta_{A,a,b,c_{1}, c_{2}}(\tau)$ to get the modular completion $\widehat{\Theta}_{A,a,b,c_{1},c_{2}}(\tau, \bar{\tau})$ which transforms like a~weight one modular form. In particular, for $ 2 b \in L^{*} $ and $\mu \in L^{*}$, where $L$ is the lattice on $\ZZ^{2}$ with a~quadratic form $A$ and $L^{*} \coloneq A^{-1} \ZZ^{2}/\ZZ^{2}$, the modular $S$-transformation of $\widehat{\Theta}_{A,a,b,c_{1},c_{2}}(\tau, \bar{\tau})$ is given by (see Appendix~\ref{stransformofnonholomorphicsec} for details),
\begin{align}
 \mathrm{e}^{- 2\pi {\rm i} B(\mu, b )} \widehat{\Theta}_{A, \mu + b, b, c_{1}, c_{2} }(\tau, \bar{\tau}) = {}&\frac{- \tau^{-1} }{\sqrt{|\det(A)|}} \sum_{ \lambda \in L^{*} } \mathrm{e}^{- 2\pi {\rm i} B(\mu + b, \lambda + b)} \mathrm{e}^{-2 \pi {\rm i} B( \lambda, b)}\nonumber\\
 &\times\widehat{\Theta}_{A,\lambda + b, b, c_{1}, c_{2}} \left(-\frac{1}{\tau},- \frac{1}{\bar{\tau}}\right).\label{modularStransformindivtheta}
\end{align}

Based on~\cite[Propositions 4.2 and 4.3]{Zwegers2008MockTF}, we can write $G_{A,a,b,c}(\tau,\bar{\tau})$ as
\begin{align}
 G_{A,a,b,c}(\tau,\bar{\tau}) ={}& -{\rm i} \sqrt{-|c|^{2}} \sum_{I \in P} \left( \sum_{\ell \in (I+a)^{\perp} + c_{\perp}\ZZ } \mathrm{e}^{2 \pi {\rm i} B(\ell, b^{\perp})} q^{\frac{|\ell|^{2}}{2}} \right)\nonumber\\
 &\times\int_{-\bar{\tau}}^{{\rm i} \infty} \frac{g_{\frac{B(I+a,c)}{|c|^{2}},B(c,b)}(-|c|^{2}y)}{\sqrt{-{\rm i}(y+ \tau)}} \dd y, \label{nonholparteichint}
\end{align}
where $x^{\perp}$ is the projection of $x$ on the line orthogonal to $c$, \smash{$x^{\perp} = x - \frac{B(x,c)}{B(c,c)} c$}, $c_{\perp}$ is a~generator of the set $\big\{ \lambda \in \ZZ^{2} \mid B(c, \lambda) =0 \big\}$, $P = \ZZ^{2} / c \ZZ \oplus c_{\perp}\ZZ $, $g_{x,y}$ is the unary theta function,
\[
 g_{x,y} \coloneq \sum_{ n \in x + \ZZ} n \mathrm{e}^{2 \pi {\rm i} n y} q^{\frac{n^{2}}{2}}.
\]

To summarize, above we laid out various notations which would help us with our further study of $\widehat{Z}$. The key point for us is that the holomorphic completion of a~regularized indefinite theta function is a~sum of three parts: holomorphic part, non-holomorphic part that depends only on $c_{1}$, and non-holomorphic part that depends only on $c_{2}$. Equation~\eqref{nonholparteichint} helps us to relate the non-holomorphic part with the Eichler integral, as in Definition~\ref{defn:mixed mock modular form}. With this setup, we will approximate $G$s with definite theta functions, which would help us to compute $\ceff$ coming from non-holomorphic parts.

\subsection[Conjectural form of zhat on the other side]{Conjectural form of $\boldsymbol{\widehat{Z}}$ on the other side}\label{sec:regularization conjecture for zhat}

As was mentioned in the introduction,
Brieskorn spheres \smash{$\overline{\Sigma(s,t,rst\pm1)} = S^3_{+1/r}(T_{s,t})$} offer a~particularly interesting class of examples for which the mathematical definition of $\widehat{Z}$ is not currently known. A~candidate for such a~definition in the case when $M_3$ is a~surgery along some knot in~$S^3$
has been proposed in~\cite{park2021inverted}.

\begin{coj}[the regularization conjecture for $\widehat{Z}$ on the other side]\label{conj:regform}
 Let \smash{$M_3=S^3_{+1/r}(K)$} be a~manifold obtained by performing $+1/r$-surgery on a~knot $K$. Then
 \begin{align}
 \widehat{Z}(M_3;q) = {}& q^\frac{r+r^{-1}}{4} \sum_{j\geq 0} f_j(K;q)\bigl(q^{-r(j + \frac{1}{2} - \frac{1}{2r})^2} - q^{-r(j + \frac{1}{2} + \frac{1}{2r})^2}\bigr)\nonumber\\
 &\times
 \left(1- \frac{\sum_{|k| \leq j}(-1)^k q^{\frac{k((2r+1)k+1)}{2}}}{\theta\bigl(-q^r, -q^{r+1}\bigr)}\right),\label{eq:Zhat regularized +1/r surgery}
 \end{align}
 where $f_j(K;q)$ are coefficients of the $F_K$ invariant of a~knot complement $S^3\setminus K$ $($see~{\rm\cite{gukov2021two})}, and
 Ramanujan theta function $\theta(a,b)$ is given by
 \[
 \theta(a, b) \coloneq \sum_{n \in \mathbb{Z}} a^{\frac{n(n+1)}{2}}b^{\frac{n(n-1)}{2}} = (-a; ab)_\infty(-b; ab)_\infty(ab;ab)_\infty.
 \]

 In particular, for $r=1$ this simplifies to
 \begin{equation}\label{eq:Zhat regularized +1 surgery}
 \widehat{Z}(M_3;q) = q^\frac{1}{2}\frac{1}{(q)_{\infty}} \sum_{j\geq 0,|k|>j} f_j(K;q)\bigl(q^{-j^2} - q^{-(j + 1)^2}\bigr)(-1)^kq^{\frac{k(3k+1)}{2}}.
 \end{equation}

\end{coj}

\begin{Remark}
 We stress that throughout the paper we take Conjecture~\ref{conj:regform} as a~working definition for $\widehat{Z}(\brs{s}{t}{rst\pm 1};q)$. Note that other definitions are possible, for example, the one in~\cite{ACDGO} using the resurgent analysis. One of the main goals of our work is to test the compatibility of the conjectural form~\eqref{eq:Zhat regularized +1 surgery} with other proposals.
\end{Remark}

Note that if $K$ is fibered, then $f_j(K;q)$ are bounded polynomials in $q$, as was discussed in~\cite{park2021inverted}. In particular, this property holds for all torus knots and thus all the case studies in this paper, where $f_j(T_{s,t};q)$ have coefficients $\{-1,0,1\}$.
Equipped with this knowledge, in the following section we will analyze the coefficients of \smash{$\widehat{Z}\big(\overline{\Sigma(s,t,rst\pm1)};q\big)$} given by the above formulas and derive the upper bounds on its $\ceff$.

To highlight its modularity properties for $M_3=\Sigma(s,t,rst+1)$, the regularized $+1/r$ surgery formula~\eqref{eq:Zhat regularized +1/r surgery} for this class of manifolds was rewritten in~\cite{Cheng:2024vou} as
\begin{equation} \label{Zhatindefinitetheta}
 \widehat{Z}
 \bigl(\overline{\Sigma(s,t,str+1)};\mathrm{e}^{2\pi {\rm i} \tau}\bigr) = \frac{ C_{\Gamma}\bigl(q^{-1}\bigr)}{2f_{2r+1,1}(\tau)} \sum_{\substack{\hat\epsilon=(\epsilon_1,\epsilon_2,\epsilon_3)\\\in (\ZZ/2\ZZ)^{\otimes 3}}}
(-1)^{\hat\epsilon} \mathrm{e}^{-2 \pi {\rm i} B(\mathsf{a}_{\hat\epsilon},\mathsf{b})} \Theta_{A,\mathsf{a}_{\hat\epsilon},\mathsf{b},\mathsf{c}_1,\mathsf{c}_2}(\tau),
\end{equation}
where
\begin{gather*}
 \mathsf{a}_{\hat\epsilon} \coloneq \begin{pmatrix}
 \frac{1}{2} - \frac{1}{2}\left( \frac{(-1)^{\epsilon_{1}}}{s} + \frac{(-1)^{\epsilon_{2}}}{t} + \frac{(-1)^{\epsilon_{3}}}{r s t +1} \right) \\ -\frac{1}{2(2r+1)}
\end{pmatrix}, \qquad \mathsf{b} \coloneq \begin{pmatrix}
 0 \\ \frac{1}{2(2r+1)}
\end{pmatrix}, \\
 A \coloneq \begin{pmatrix}
 -2 s t (r s t+1) & 0 \\
 0 & 2r+1
\end{pmatrix},
\end{gather*}
$C_{\Gamma}$ is a~monomial in $q$, the function $f_{x,\chi}(\tau)$ is defined as follows:
\[
 f_{x,\chi}(\tau) \coloneq \sum_{k \in \mathbb{Z}} (-1)^{k} q^{\frac{x}{2} (k- \frac{\chi}{2x})^{2} },
\]
and $\mathsf{c}_1$, $ \mathsf{c}_2$ are
\[
 \mathsf{c}_1 = \left(\begin{array}{c}
 1\\
 0
 \end{array}\right), \qquad
 \mathsf{c}_2 = \left(\begin{array}{c}
 2r+1\\
 2(str+1)
 \end{array}\right).
\]
The authors of~\cite{Cheng:2024vou} found that when $s=2$ and $t=3$, as well as when $s=2$, $t=5$ and $r=3$, the invariant is a~pure mock modular form, whereas in the generic case it is a~mixed mock modular form.
In Section~\ref{sec:Mixed mock modularity and resurgence}, we will leverage this by analyzing the $S$-transformation of this object.

\section[Upper bounds on ceff for sigmastrpm1]{Upper bounds on $\boldsymbol{\ceff}$ for $\boldsymbol{\overline{\Sigma(s,t,rst\pm 1)}}$}\label{sec:Bounds on ceff}

We begin our analysis by establishing the upper bounds on
$\ceff$ for manifolds $\overline{\Sigma(s,t,rst\pm1)}$.
We will show that
$\ceff(P)$ is bounded above by the contribution from the denominator in the regularized formula for $\widehat{Z}$.
This comparison yields explicit upper bounds on $c_{\mathrm{eff}}$ for the
$\widehat{Z}$-invariant of both orientations of $\overline{\Sigma(s,t,rst\pm1)}$ when the
3-manifold is obtained by a~$\tfrac{1}{r}$ surgery on torus knots. We begin with a~definition.

\begin{Definition}\label{def:polynomially-bounded}
 Given an infinite sequence of integers $\{a_n\}_{n=0}^\infty$ we say that $a_n$ is polynomially bounded if there exists a~polynomial $p(x) \in \mathbb{R}[x]$ such that $a_n \leq p(n)$ for all $n \geq 0$.
\end{Definition}

\begin{Remark}\label{polynomially-bounded-equivalence}
 One can equivalently define $a_n$ to be polynomially bounded if there exist constants $C, D \in \mathbb{R}$ and an integer $k \geq 0$ such that $a_n \leq Cn^k + D$ for all $n \geq 0$. One direction of this equivalence is immediate, to prove the other observe if $p(x) = \sum_{i=0}^k c_i x^i$ satisfies $a_n \leq p(n)$ for all $n \geq 0$, then \smash{$a_n \leq \bigl(\sum_{i=0}^k |c_i|\bigr)n^k + |p(0)|$} for all $n \geq 0$.
\end{Remark}

We continue with a~lemma that will be useful, for our purposes, in deriving estimates of $\ceff$ for $\widehat{Z}$.

\begin{Lemma}\label{helper-lemma}
Let $P(q)=\sum_{n=0}^{\infty}a_nq^n$ and $Q(q)=\sum_{n=0}^{\infty}b_nq^n$ be $q$-series such that $c_{\mathrm{eff}}(P(q))<\infty$ and $|b_n|$ is polynomially bounded. Then $c_{\mathrm{eff}}(P(q)Q(q))
\leq c_{\mathrm{eff}}(P(q))$.
\end{Lemma}

\begin{proof}
Write $P(q)Q(q)=\sum_{n=0}^{\infty}
\bigl(\sum_{k=0}^{n}a_kb_{n-k}\bigr)q^n$. and set
\[
L:=
\limsup_{n\to\infty}
\frac{(\log|a_n|)^2}{n}
=
\frac{2\pi^2}{3}c_{\mathrm{eff}}(P(q)).
\]
Fix $\varepsilon>0$. By the definition of the limsup, there exists
$N_\varepsilon$ such that, for every $n\geq N_\varepsilon$,
\[
|a_n|
\leq
\exp\left(\sqrt{(L+\varepsilon)n}\right).
\]
After increasing a constant $C_\varepsilon>0$ to account for the
finitely many terms with $n<N_\varepsilon$, we therefore have
\[
|a_n|
\leq
C_\varepsilon
\exp\bigl(\sqrt{(L+\varepsilon)n}\bigr)
\]
for all $n \geq 0$. Since $|b_n|$ is polynomially bounded, there exist constants $C>0$
and an integer $d\geq0$ such that
\[
|b_n|\leq C(n+1)^d
\]
for every $n\geq0$. Hence
\[
\begin{aligned}
\left|
\sum_{k=0}^{n}a_kb_{n-k}
\right| &\leq
\sum_{k=0}^{n}|a_k||b_{n-k}|\leq
CC_\varepsilon(n+1)^d
\sum_{k=0}^{n}
\exp\bigl(\sqrt{(L+\varepsilon)k}\bigr)\\
&\leq
CC_\varepsilon(n+1)^{d+1}
\exp\bigl(\sqrt{(L+\varepsilon)n}\bigr).
\end{aligned}
\]
It follows that
\[
\log
\left|
\sum_{k=0}^{n}a_kb_{n-k}
\right|
\leq
\sqrt{(L+\varepsilon)n}
+
(d+1)\log(n+1)
+
O(1).
\]
Therefore,
\[
\limsup_{n\to\infty}
\frac{
\bigl(
\log\bigl|
\sum_{k=0}^{n}a_kb_{n-k}
\bigr|
\bigr)^2
}{n}
\leq
L+\varepsilon.
\]
Since $\varepsilon>0$ is arbitrary,
\[
\limsup_{n\to\infty}
\frac{
\bigl(
\log\bigl|
\sum_{k=0}^{n}a_kb_{n-k}
\bigr|
\bigr)^2}{n}
\leq L.
\]
Multiplying both sides by $\frac{3}{2\pi^2}$ gives
$c_{\mathrm{eff}}(P(q)Q(q))
\leq
c_{\mathrm{eff}}(P(q))$.
\end{proof}

For the benefit of the reader, we have stated the above lemma in greater generality than needed. As we shall see, for our purposes it will be applied in the case when $P(q)$ has strictly non-negative coefficients.

\subsection[Estimates of ceff for T S(s,t,+/-1) ]{Estimates of $\boldsymbol{\ceff}$ for $\boldsymbol{T[\Sigma(s, t, st \pm 1)]}$}

In this section, we find the upper bounds of $\ceff$ for $T[\Sigma(s, t, st \pm 1)]$.
These manifolds correspond to $+1$ surgeries on torus knots of type $T(s,t)$ and $T(s,-t)$.
For $M_3 = \overline{\Sigma(s, t, st \pm 1)}$, the conjectural form of $\widehat{Z}$ \eqref{eq:Zhat regularized +1 surgery} takes the form
\begin{align}
 \widehat{Z}(M_3;q) ={}& \frac{q^{\frac{1}{2}}}{(q)_\infty}\sum_{j\geq 0, |k| > j} \varepsilon_{2j+1} q^{\mp \frac{j(j+1)}{st}\mp \frac{(t^2-1)(s^2-1)}{4st}}\bigl(q^{-j^2}-q^{-(j+1)^2}\bigr)(-1)^kq^{\frac{k(3k+1)}{2}} \nonumber\\
 ={}& \frac{q^{\frac{1}{2}}}{(q)_\infty}\sum_{j\geq 0} \varepsilon_{2j+1} q^{\mp \frac{j(j+1)}{st}\mp \frac{(t^2-1)(s^2-1)}{4st}}\bigl(q^{-j^2}-q^{-(j+1)^2}\bigr)\nonumber\\
 &\times\sum_{k=j+1}^\infty \bigl((-1)^k q^{\frac{3k^2+k}{2}} + (-1)^k q^{\frac{3k^2-k}{2}}\bigr),\label{base-eq}
\end{align}
where
\begin{equation}\label{eq:epsilon}
\varepsilon_{2j+1} =
\begin{cases}
 -1 & \text{if $2j+1 \equiv st \pm (s+t) \mod 2st$},\\
 \hphantom{-}1 & \text{if $2j+1 \equiv st \pm (s-t) \mod 2st$}, \\
 \hphantom{-}0 & \text{otherwise}.
\end{cases}
\end{equation}
It is worth noting that the $\mp$ signs appearing in the exponents in~\eqref{base-eq} depend on the $\pm$ signs from $\overline{\Sigma(s, t, st \pm 1)}$ whereas the $\pm$ signs in~\eqref{eq:epsilon} are independent.

To find the upper bounds, the main idea is to apply Lemma~\ref{helper-lemma} to the above expression for $\widehat{Z}$ expanded around $q=0$, where -- up to a~power of $q$ -- we treat the denominator \smash{$\frac{1}{(q)_{\infty}}$} and the numerator of~\eqref{base-eq} as $P(q)$ and $Q(q)$, respectively. Let us first see that all conditions of the aforementioned lemma are satisfied. First of all, $P(q)$ is a~well-defined function whose asymptotics and the growth of coefficients are known ($\ceff=1$). Its expansion around $q=0$ also satisfies monotonic growth. That is, we can take
\[
 P(q) = 1 + q + 2 q^2 + 3 q^3 + 5 q^4 + 7 q^5 + \cdots.
\]
On the other hand, the physical definition of $\widehat{Z}$ implies a~well-defined coefficient growth captured by $\ceff$, as discussed in~\cite{Gukov:2023cog}. However, we still need to confirm that $\widehat{Z}$ as a~$q$-series contains only integer exponents of $q$, up to some common prefactor.

\begin{Proposition}\label{delta-prop-r-1}
 Let $M_3 = \overline{\Sigma(s, t, st\pm1)}$. The minimal $q$-exponent of $\widehat{Z}(M_3;q)$ is given by
 \begin{equation} \label{eq:delta}
 \Delta = \frac{1}{2} + \frac{(s-1)(t-1)(st-s-t-1)}{8} \mp \frac{(s-1)(t-1)}{2}.
 \end{equation}
 Moreover, $\widehat{Z}(M_3;q) \in q^{\Delta}\mathbb{Z}[[q]]$.
\end{Proposition}

\begin{proof}
By inspecting~\eqref{base-eq}, we see that the minimal $q$-exponent of $\widehat{Z}(M_3;q)$ is given by
\[
\frac{1}{2} \mp \smash{\frac{j(j+1)}{st} \mp \frac{(t^2-1)(s^2-1)}{4st} - (j+1)^2 + \frac{3(j+1)^2-(j+1)}{2}},
\]
where $j \geq 0$ is the first such natural number such that $\varepsilon_{2j+1} \neq 0$. From the definition of $\varepsilon_{2j+1}$, we see that $j = \frac{1}{2}(st-s-t-1)$ and one arrives at~\eqref{eq:delta} after simplification. To show that \smash{$\widehat{Z}(M_3;q) \in q^{\Delta}\mathbb{Z}[[q]]$}, it suffices to note that the difference of any pair of $q$-powers in~\eqref{base-eq} is an integer. This is in part due to the fact that for~${j_1, j_2 > 0}$ the definition of $\varepsilon_{2j+1}$ ensures that $\frac{1}{st}[j_2(j_2+1) - j_1(j_1+1)]$ is an integer.
\end{proof}

\begin{Remark}\label{math-physics-ceff-connection}
The study of the minimal $q$-exponent $\Delta$ of $\widehat{Z}$ for negative definite plumbed manifolds was initiated in~\cite{cobordism_bps}.
Formulas for $\Delta$ for the negative definite family of Brieskorn spheres $M_3 = \Sigma(p, q, r)$ appeared in~\cite{delta_brieskorn,delta_harichurn_nemethi_svoboda}.
It is interesting to note that, comparing the value of $\Delta$ in~\eqref{eq:delta} for $\overline{\Sigma(s, t, st\pm1)}$ with that of $\Delta$ for $\Sigma(s, t, st\pm1)$, we see that $\Delta$ is not negated under orientation reversal as was expected in~\cite{cobordism_bps}.
\end{Remark}

The most non-trivial check which remains to be done is that the numerator of $\widehat{Z}$ is polynomially bounded. If this is true, we can then directly apply Lemma~\ref{helper-lemma} to get the upper bound on $\ceff$.
Notice in the numerator of the expression in equation~\eqref{base-eq} that all of the coefficients in front of $q$-powers are $\pm 1$.
However, after performing the summation in the numerator, not all coefficients will be unary, therefore we need a~careful analysis for them.

The idea of the proof is to first split the expression in~\eqref{base-eq} into separate infinite sums~$Q^{\pm}_i(q)$ such that each $Q^{\pm}_i(q)$ contains powers of $q$, written as \smash{$q^{\varphi_i^{\pm}(j, k)}$} with coefficients \smash{$c_{\varphi_i^{\pm}(j, k)}$} satisfying \smash{$|c_{\varphi_i^{\pm}(j, k)}| \leq 1$}. (Note that the $\pm$ signs in \smash{$Q^{\pm}_i$} and \smash{$\varphi_i^{\pm}$} are included to show dependence on the signs in~$\overline{\Sigma(s, t, st\pm1)}$.) Then for a~specific $q$-power, say $q^{d}$, we count the number of repetitions of this $q$-power in the expression of \smash{$Q^{\pm}_i(q)$}, that is the number of $j, k \in \mathbb{N}$ such that \smash{$\varphi_i^{\pm}(j, k) = d$}, and use this to give an upper bound for the absolute value of the corresponding coefficient $c_d$ associated to $q^d$. The number of repetitions of a~given $q$-power gives an upper bound for the absolute value of the corresponding coefficient by virtue of the fact that the coefficients in the expression for \smash{$Q^{\pm}_i(q)$} satisfy \smash{$|c_{\varphi_i^{\pm}(j, k)}| \leq 1$}.

Once the $q$-series
\[
Q^{\pm}_i(q) = \sum_{j = 0}^\infty \sum_{k=j+1}^\infty c_{\varphi_i^{\pm}(j, k)}q^{\varphi_i^{\pm}(j, k)}
\]
 are identified, the sum is split into
\smash{$\widehat{Z}=q^{\Delta}P(q)Q(q)$}, where \smash{$Q(q)=q^{-\Delta}\sum_{i=1}^{4}Q^{\pm}_i(q)$} and $P(q)=\frac{1}{(q)_\infty}$.
The proof is then completed by showing that the absolute value of the coefficients of $Q(q)$ are polynomially bounded, therefore the upper bound on the $\ceff$ comes from the $1/(q)_\infty$ factor.

\begin{Lemma}\label{lem:bounding-box}
Let
\begin{gather*}
 \varphi_1^{\pm}(j, k) = \frac{1}{2}\mp \frac{j(j+1)}{st} \mp \frac{\bigl(t^2-1\bigr)\bigl(s^2-1\bigr)}{4st}-j^2 + \frac{3k^2+k}{2},\qquad
 \varphi_2^{\pm}(j, k) = \varphi_1^{\pm}(j, k)-k ,\\
 \varphi_3^{\pm}(j, k) = \varphi_1^{\pm}(j, k)-2j-1 ,\qquad
 \varphi_4^{\pm}(j, k) = \varphi_2^{\pm}(j, k) - 2j-1
\end{gather*}
and for $1 \leq i \leq 4$ let
\begin{equation}\label{eq:bb-expression}
 S_n^{i, \pm} = \bigl\{(j, k) \in \mathbb{N}^2 \mid k > j \geq 0 \text{ and } \varphi_i^{\pm}(j, k) = \Delta + n \bigr\}.
\end{equation}
Then $\bigl|S_n^{i, \pm}\bigr|$ is polynomially bounded for all $i$.
\end{Lemma}
\begin{proof}
If we consider $\varphi_i^{\pm}(j, k)$ as a~real-valued function $\varphi_i^{\pm} \colon\mathbb{R}^2 \to \mathbb{R}$ then we have that $S^{i, \pm}_n \subseteq \big\{(j, k) \in \mathbb{R}^2 \mid \varphi_i^{\pm}(j, k) = \Delta + n\big\} \cap \{j, k \in \mathbb{R} \mid k > j \geq 0\} \cap \mathbb{Z}^2$. From the definition of the~$\varphi_i^{\pm}$, we see that the graph of the set $H^{\pm}_i = \big\{(j, k) \in \mathbb{R}^2 \mid \varphi_i^{\pm}(j, k) = \Delta + n\big\} $ is a~hyperbola in~$\mathbb{R}^2$. See Figure~\ref{fig:schematic} for a~schematic. Without loss of generality we will focus now on proving the statement for $S^{1, +}_n$, with all remaining cases being similar with appropriate modifications made. We start with $\varphi_1^+$. In Figure~\ref{fig:schematic}, as $j \to \infty$, one can see that the line $k =j$ grows faster than the upper curve of the hyperbola. One can further verify that $H^+_i$ has asymptotes at
\[
k = \pm \sqrt{\frac{2}{3}\left(1 + \frac{1}{st}\right)}\left(j+\frac{1}{2(st+1)}\right) - \frac{1}{6}.
\]
 Thus we just have to compute the intersection point $X$ of $H^{+}_1$ with the line $k = j$ in the upper-right-hand quadrant of $\mathbb{R}^2$. This amounts to making the substitution of $k =j$ into the equation $\varphi_1^+(j, k) = \Delta +n$ and then applying the quadratic formula. We get $X = (x, x)$, where
\[
x = \frac{1}{2}\left(-1+ \sqrt{(st-s-t)^2 + \frac{8stn}{st-2}}\right).
\]
For simplicity in what follows, define $m \coloneq (4st)^2(1+n)$. One can verify that $m > x$, and if we define the bounding box $B \coloneq (0, m) \times (0, m)$, then $B \cap \mathbb{Z}^2$ will contain $S^{1, +}_n$ and the cardinality of the set $B \cap \mathbb{Z}^2$ is effectively $m^2$ which is a~degree $2$ polynomial in $n$. Thus $|S^{1, +}_n|$ is polynomially bounded. The proofs of $\varphi_i^{\pm}$ and $S^{i, \pm}_n$ in all remaining cases follow similarly to the above arguments, except different values of $x$ and $m$ will be chosen as the $H^{\pm}_i$ will be modified hyperbolas to $H^{+}_1$.
\end{proof}

\begin{figure}[!ht]
 \centering
 \includegraphics[width=0.5\linewidth]{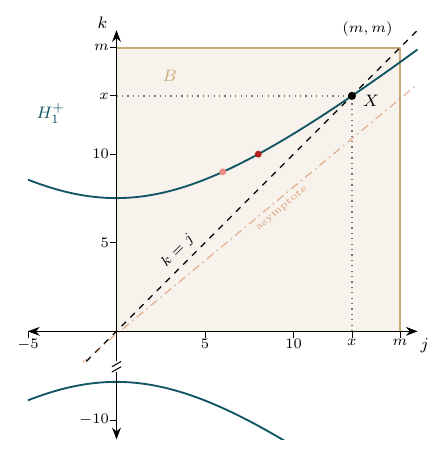}
 \caption{Schematic of hyperbola $H^{+}_1$ intersecting line $k = j$ (here $s=3$, $t=10$, $n=55$ and $H^{+}_1$ is given by $\varphi_1^+(j,k) = \Delta + n$). We focus on the upper right quadrant. If $X$ is the intersection point, the bounding box under consideration is the box $B = (0, m) \times (0, m)$. Then we see that $S^{1, +}_n = B \cap H^{+}_1 \cap \mathbb{Z}^2$. In this example $S^{1, +}_{55}$ consists of two points depicted by a~light red dot, \softdot{}, and dark red dot, \intensedot{}. The dark red dot depicts the only point in this example in which makes a~contribution to~\eqref{base-eq}. The dashed orange line shows the asymptote of $H_1^+$.}
 \label{fig:schematic}
\end{figure}

\begin{Proposition}\label{prp:upper_bound}
 For $M_3 = \overline{\Sigma(s, t, st\pm 1)}$, we get $\ceff\bigl(\widehat{Z}(M_3;q)\bigr) \leq 1$.
\end{Proposition}

\begin{proof}

\noindent For $1 \leq i \leq 4$, define an infinite sum \[
 Q^{\pm}_i(q) = \sum_{j = 0}^\infty \sum_{k=j+1}^\infty c_{\varphi_i^{\pm}(j, k)}q^{\varphi_i^{\pm}(j, k)},
\]
where $\varphi_i^{\pm}(j, k)$ is defined as in Lemma~\ref{lem:bounding-box} and \smash{$c_{\varphi_i^{\pm}(j, k)}$} is defined to be \smash{$(-1)^k\varepsilon_{2j+1}$} for $i=1, 2$ and \smash{$(-1)^{k+1}\varepsilon_{2j+1}$} for $i=3, 4$.
Comparing with~\eqref{base-eq}, we find
\[
 \widehat{Z}(M_3;q) = \frac{1}{(q)_\infty}\sum_{i=1}^4 Q^{\pm}_i(q).
\]
Letting $\Delta$ denote the minimal $q$-exponent of $\widehat{Z}(M_3;q)$, Proposition~\ref{delta-prop-r-1} implies that for each ${1 \leq i \leq 4}$ and for each $j$, $k$ the $q$-powers of $Q^{\pm}_i(q)$ satisfy $\varphi_i^{\pm}(j, k) = \Delta + n$ for some $n \in \mathbb{N}$. Additionally, note that for each $i$ and $j$, $k$ one has \smash{$|c_{\varphi_i^{\pm}(j, k)}| \leq 1$}. Combining these observations, if we collect like terms in \smash{$Q^{\pm}_i(q)$}, we can write
\[
q^{-\Delta}Q^{\pm}_i(q) = \sum_{n=0}^\infty b^{i, \pm}_nq^n \in \mathbb{Z}[[q]]
\] for some coefficients \smash{$b^{i, \pm}_n$} which satisfy \smash{$\big|b^{i, \pm}_n\big| \leq \big|S^{i, \pm}_n\big|$}. Lemma~\ref{lem:bounding-box} then shows that \smash{$\big|b^{i, \pm}_n\big|$} is polynomially bounded.

Now set $P(q) = 1/(q)_\infty$ and \smash{$Q(q) = q^{-\Delta}\sum_{i=1}^4Q^{\pm}_i(q)$}, so that \smash{$q^{-\Delta}\widehat{Z}(M_3;q) = P(q)Q(q)$}.
The absolute value of the $n$-th coefficient of $Q(q)$ is given by \smash{$\big|\sum_{i=1}^4 b^{i, \pm}_n \big|$}. Noting that
\[
\biggl|\sum_{i=1}^4 b^{i, \pm}_n \biggr| \leq \sum_{i=1}^4 \big|b^{i, \pm}_n\big|
\]
 a~routine verification shows that \smash{$\big|\sum_{i=1}^4 b^{i, \pm}_n \big|$} is polynomially bounded. So we can apply Lem\-ma~\ref{helper-lemma} to see that $\ceff(P(q)Q(q)) \leq \ceff(P(q))$.
Given $\ceff(P)=1$, this completes the proof.
\end{proof}

\subsection[Generalization to higher values of r]{Generalization to higher values of $\boldsymbol{r}$}

We recall the $+1/r$-surgery~\eqref{eq:Zhat regularized +1/r surgery} in the case of $M_3 = \overline{\Sigma(s, t, rst\pm 1)}$ below
 \begin{align}
 \widehat{Z}(M_3;q) = {}& q^\frac{r+r^{-1}}{4} \sum_{j\geq 0} \varepsilon_{2j+1} q^{\mp \frac{j(j+1)}{st}\mp \frac{(t^2-1)(s^2-1)}{4st}}\bigl(q^{-r(j + \frac{1}{2} - \frac{1}{2r})^2} - q^{-r(j + \frac{1}{2} + \frac{1}{2r})^2}\bigr) \nonumber\\
 & \times\left(1- \frac{\sum_{|k| \leq j}(-1)^k q^{\frac{k((2r+1)k+1)}{2}}}{\theta\bigl(-q^r, -q^{r+1}\bigr)}\right),\label{z-hat-start}
\end{align}
where $\varepsilon_{2j+1}$ is defined as in~\eqref{eq:epsilon}.

We would like to use the same idea as in Proposition~\ref{prp:upper_bound} to study the upper bounds for $\ceff$ for $\overline{\Sigma(s,t, rst\pm 1)}$, however we encounter two issues in doing so. The first issue is a~subtlety in locating the most dominant singularity of $\widehat{Z}(\overline{\Sigma(s,t,rst\pm 1)})$, a~point we will address in more detail in the next section.
Namely, numerous examples suggest that for $r > 1$, $\widehat{Z}(\overline{\Sigma(s, t, rst\pm 1)})$ has the most dominant singularity near $q = \mathrm{e}^{2\pi {\rm i}\omega}$, where $\omega=\frac{a}{2r+1}$ for some integer $a$. However, if we multiply \smash{$\widehat{Z}$} by \smash{$\eta(q)^{-1}$}, the most dominant singularity shifts to $q=1$.
Based on this evidence, we make the assumption throughout this section that $\omega = 0$.

The second issue is that we cannot directly apply the techniques used to prove $\ceff$ for the~$+1$ surgery formula above as~\eqref{z-hat-start} involves a~summation over the cone $\mathcal{C} = \big\{(j, k) \in \mathbb{Z}^2 \mid {j \geq 0} \allowbreak \text{and } |k| \leq j\big\}$. When we intersect $\mathcal{C}$ with the hyperbola (see Figure~\ref{fig:schematic}) generated by the set of $(j, k) \in \mathbb{Z}^2$ which satisfy the expression~\eqref{eq:bb-expression} for a~given $q$-power, we see that there are infinitely many points in $\mathbb{Z}^2$ in the intersection.
To remedy this, we can rewrite~\eqref{eq:Zhat regularized +1/r surgery} in the following form:
\begin{align}
 \widehat{Z}(M_3;q) = {}&\frac{q^\frac{r+r^{-1}}{4}}{\theta\bigl(-q^r, -q^{r+1}\bigr)} \sum_{j\geq 0, |k| > j} f_j(K;q)\bigl(q^{-r(j + \frac{1}{2} - \frac{1}{2r})^2} - q^{-r(j + \frac{1}{2} + \frac{1}{2r})^2}\bigr)\nonumber\\
 &\times(-1)^k q^{\frac{k((2r+1)k+1)}{2}}.\label{z-hat-better-form}
\end{align}

To confirm this, it suffices to show that
\begin{equation}\label{minimized-problem}
 1- \frac{\sum_{|k| \leq j}(-1)^k q^{\frac{k((2r+1)k+1)}{2}}}{\theta\bigl(-q^r, -q^{r+1}\bigr)} = \frac{\sum_{|k| > j}(-1)^k q^{\frac{k((2r+1)k+1)}{2}}}{\theta\bigl(-q^r, -q^{r+1}\bigr)}.
\end{equation}
Note that we have
\[
 \theta\bigl(-q^r, -q^{r+1}\bigr) = \sum_{k \in \mathbb{Z}} \bigl(-q^r\bigr)^{\frac{k(k+1)}{2}}\bigl(-q^{r+1}\bigr)^{\frac{k(k-1)}{2}} = \sum_{k \in \mathbb{Z}}(-1)^{\frac{k(k+1)}{2}}(-1)^{\frac{k(k-1)}{2}}q^{2rk^2+k^2-k}.
\]
Examining the factors of $-1$ and using $\frac{k(k+1)}{2}+ \frac{k(k-1)}{2} = k^2$ and $(-1)^{k^2} = (-1)^k$ for all $k$, we can simplify the above equation into
\[
 \theta\bigl(-q^r, -q^{r+1}\bigr) = \sum_{k \in \mathbb{Z}}(-1)^{k}q^{2rk^2+k^2-k}.
\]
Note that within the summation there is a~symmetry in terms with $k$ and $-k$. Namely, we have~${(-1)^{-k}q^{2r(-k)^2+(-k)^2-(-k)} = (-1)^{k}q^{2rk^2+k^2+k}}$. Thus we may change the index of summation from $k$ to $-k$ and we end up with
\[
 \theta\bigl(-q^r, -q^{r+1}\bigr) = \sum_{k \in \mathbb{Z}}(-1)^{k}q^{2rk^2+k^2+k} = \sum_{k \in \mathbb{Z}}(-1)^{k}q^{\frac{k((2r+1)k+1)}{2}}.
\]
From this~\eqref{minimized-problem} follows immediately. With the $+1/r$-surgery formula now given by~\eqref{z-hat-better-form}, we can apply the bounding techniques used to find the bound for $\ceff$ for the $+1$ surgery formula to the $+1/r$ surgery formula. One has the following updated proposition.

\begin{Proposition}
 Let $M_3 = \overline{\Sigma(s, t, rst\pm1)}$. The minimal $q$-exponent of $\widehat{Z}(M_3;q)$ is given by
 \begin{equation} \label{eq:delta-r}
 \Delta = (r \mp 1)\frac{(s-1)(t-1)}{2} + \frac{(st-s-t-1)(st-s-t-3)}{8} - \frac{1}{2}.
 \end{equation}
 Moreover, $\widehat{Z}(M_3;q) \in q^{\Delta}\mathbb{Z}[[q]]$.
\end{Proposition}

\begin{proof}
 The proof is the same as that of Proposition~\ref{delta-prop-r-1} except this time, inspecting~\eqref{z-hat-better-form} shows that the minimal $q$-exponent is given by
 \[
 \frac{r+r^{-1}}{4} \mp \frac{j(j+1)}{st} \mp \frac{\bigl(t^2-1\bigr)\bigl(s^2-1\bigr)}{4st} - r\left(j + \frac{1}{2} + \frac{1}{2r}\right)^2  + \frac{(2r+1)(j+1)^2 -(j+1)}{2}
 \]
 for $j = \frac{1}{2}(st-s-t-1)$. After simplification, one arrives at~\eqref{eq:delta-r}.
\end{proof}

The proof of the bound for $\ceff$ is then completed by using the asymptotics of $\theta\bigl(-q^r, -q^{r+1}\bigr)$ near $q = 1$ in place of the asymptotics for $1/(q)_\infty$, which can be derived analogously to Example~\ref{exmp:q-pochhammer ceff}. In the end, we obtain the following key result:

\begin{Proposition}\label{prop:general upper bounds} Assuming that $\omega = 0$ holds in~\eqref{eq:general c_eff} for $\overline{\Sigma(s, t, rst\pm 1)}$, we have
 \[\ceff\big[\widehat{Z}(\overline{\Sigma(s,t,rst\pm 1)};q)\big] \leq \frac{3}{2r+1}.\]
\end{Proposition}

Having established the upper bounds of $\ceff$ for Brieskorn spheres, we now proceed to estimate the lower bounds.

\section[Lower bounds on ceff for strstp1]{Lower bounds on $\boldsymbol{\ceff}$ for $\boldsymbol{\overline{\Sigma(s,t,rst\pm 1)}}$}\label{sec:Numerical analysis of ceff}

In this section, we discuss the numerical analysis of $\widehat{Z}$ for Brieskorn spheres $\overline{\Sigma(s,t,rst\pm 1)}$.
We begin by determining the dominant singularities of the conjectural expression~\eqref{eq:Zhat regularized +1/r surgery}.
To this end, we examine essential singularities of $\widehat{Z}$ as $q$ approaches the unit circle radially, and estimate their relative strength.
Having located the dominant singularities, we proceed with calculating the lower bounds for $\ceff\big[\zhat(\overline{\Sigma(s,t,rst\pm 1)};q)\big]$ using the first $N$ terms in the $q\to 0^+$ expansion of the regularization formula~\eqref{eq:Zhat regularized +1/r surgery}.
Combining this with the upper bounds from Proposition~\ref{prop:general upper bounds}, we will be able to analyze the relationship between $\ceff$ and Chern--Simons invariants of flat connections.

\subsection[Singularity structure of zhat]{Singularity structure of $\boldsymbol{\zhat}$}

As was argued in~\cite{Gukov:2023cog} using the saddle point analysis, supersymmetric indices of 3d $\mathcal{N}=2$ theories (in particular, $\widehat{Z}$ invariants) in general are expected to have
\[
 a_n \sim \Re \left[ \exp\left( \sqrt{
 \frac{2\pi^2}{3}c_{\rm eff}n
 }+2\pi {\rm i} r n \right) \right],
\]
where $r$ is the position of the most dominant singularity (see Definition~\ref{dfn:dominant singularity}) of the index on the unit circle. Since the behavior of coefficients $a_n$ depends heavily on the location of the dominant singularities, it is essential to locate them beforehand.

{\bf Position of dominant singularities.}
The strength of a~singularity reflects how quickly the function approaches infinity as it nears the pole (or, in our case, a~pole of infinite order, i.e., the essential singularity).\footnote{Although the behavior of a~general function around the essential singularity may be erratic, in our case we will use the radial limits of $\widehat{Z}$ invariant which are well behaved.} The same applies for $\widehat{Z}$ invariants on the other side, when considered as a~function of $q\in\mathbb{C}$. According to the regularization conjecture discussed in Section~\ref{sec:regularization conjecture for zhat}, $\widehat{Z}$~invariants on the other side of Brieskorn spheres take the form of a~combination of indefinite theta series which are convergent in the unit disk. Since the denominator of the regularized form (i.e., Ramanujan theta function) has zeroes at roots of unity, one would expect that they determine all singularities of $\widehat{Z}$. However, it is difficult to predict which singularities will be more dominant.

As an example, consider $\big|\widehat{Z}\big|$ inside the unit circle for Brieskorn spheres $\overline{\Sigma(2,3,5)}$ and $\overline{\Sigma(2,3,}\allowbreak \smash{\overline{11)}}$, Figure~\ref{fig:singular examples}. The white regions (i.e., very large absolute values) are formed around a~singularity on the unit circle, and we can estimate its strength by comparing the areas of such regions and computing numerically the radial limits, matching them with Definition~\ref{dfn:dominant singularity}.

\begin{figure}[!ht]
 \centering
 \includegraphics[width=0.8\linewidth]{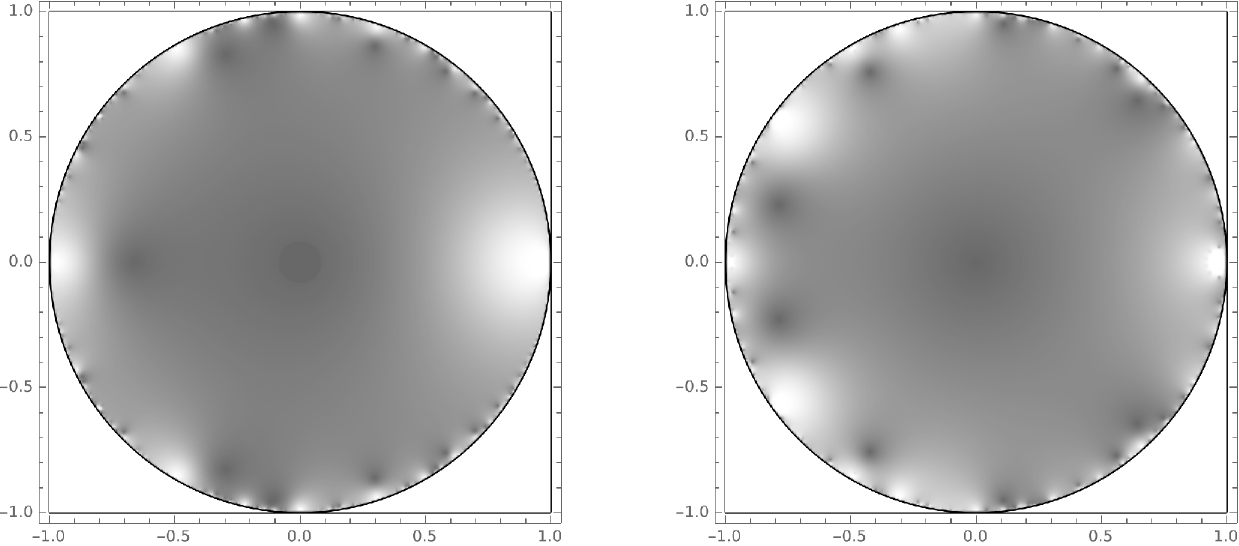}
 \caption{$\big|\widehat{Z}\big|$ as a~function of $q\in\mathbb{C}$ inside the unit circle: $\overline{\Sigma(2,3,5)}$ with the dominant singularity at~${q=1}$ (left), and $\overline{\Sigma(2,3,11)}$ with the dominant singularities at $q=\mathrm{e}^{-2\pi {\rm i} \frac{2}{5}}$ and $q=\mathrm{e}^{-2\pi {\rm i} \frac{3}{5}}$ (right), which can be determined by comparing the radial limits.}
 \label{fig:singular examples}
\end{figure}

Similarly, for $\overline{\Sigma(2,3,13)}$, we find the dominant as well as the sub-dominant singularities whose structure turns out to be exactly the same as for $\overline{\Sigma(2,3,11)}$ (Figure~\ref{fig:singular examples}, right).

Our analysis shows that for $s,t\leq 20$, $\zhat\bigl(\overline{\Sigma(s,t,st\pm 1)};q\bigr)$ has the dominant singularity at~${q=1}$. On the other hand, for $r>1$ there are a~number of examples where the dominant singularity is not at $q=1$, such as $\overline{\Sigma(2,3,12\pm 1)}$ (and the number of such cases grows with $r$).\footnote{Nonetheless, dividing by the Dedekind eta function $\eta(\tau)$ rotates the dominant singularity to $q=1$ for all $r$, $s$ and $t$.} However, if $r$ is fixed and $s$, $t$ are taken to be large enough, then it is again at $q=1$.

\subsection{Estimates of the lower bounds}

By using the truncated series, we can obtain lower bounds on $\ceff$. To this end, consider
\begin{equation}\label{eq:Zhat truncated}
 \widehat{Z}|_N = \sum_{k\in I_N} a_kq^k,
\end{equation}
where $I_N$ denotes the set of exponents of the first $N$ non-zero terms in $\widehat{Z}$.
Any estimated $\ceff$ coming from~\eqref{eq:Zhat truncated} is automatically a~lower bound because of the property called \emph{log-concavity}, see~\cite{desalvo2015log}.
In this way, more terms in the truncated series $\widehat{Z}|_N$ will guarantee a~better bound, and for our studies we set $N=10^4$.

Our next step is to find the range of admissible values for $c_{\rm eff}$.
From the above analysis, it is safe to assume that $q=1$ is the dominant singularity for all $+1$-surgeries.
We can then compare the upper bound given by Propositions~\ref{prp:upper_bound} and~\ref{prop:general upper bounds} with the numerical estimates providing the lower bound. In this way, we can find all possible integer values of $m$ (recall the conjectural relation~\eqref{coj:z-hat at leading order}) lying within the range $[m_{\min}:m_{\max}]$, where $m_{\min}$ is given by the numerical estimate, and $m_{\max}$ is the exact bound from Proposition~\ref{prop:general upper bounds}.
In practice, this can significantly narrow down the range of admissible $c_{\rm eff}$ values.
Based on this, we can compute the numerical estimates of $m$ for $+1$-surgeries on various torus knots.

To illustrate the idea, consider the reverse-oriented Poincar\'e homology sphere \smash{$\overline{\Sigma(2,3,5)}$}. From Figure~\ref{fig:singular examples} we already know that its dominant singularity is at $q=1$. Therefore, we can compare the leading behavior with some Chern--Simons invariant.
There are three flat connections, $\alpha_0$ (abelian), $\alpha_1$ and $\alpha_2$, as listed in Table~\ref{tab:Invariants of flat connections on Poincare sphere}.
In this case, the corresponding value $m=1$ (in other words, $m=1$ and $l=0$) is known from modularity~\cite{Gukov:2023cog}, but we would nevertheless like to confirm the consistency with Figure~\ref{fig:three-part comparison}.
We find the estimated value of $m$ to be $\sim 0.9$, by considering the first $10^4$ terms in $\widehat{Z}$. We find it useful to plot $\log|a_n|$ as a~function of $n$ and compare with the upper bound given by the regularization factor $(q)_{\infty}^{-1}$ in the $+1$-surgery formula. Additionally, we include the logarithmic lines for $m$ taking all integer values, starting from $m=1,2,3,\dots$. Combining all these plots into one, we obtain the following picture (Figure~\ref{fig:ceff235}).

\begin{figure}
 \centering
 \includegraphics[width=0.55\linewidth]{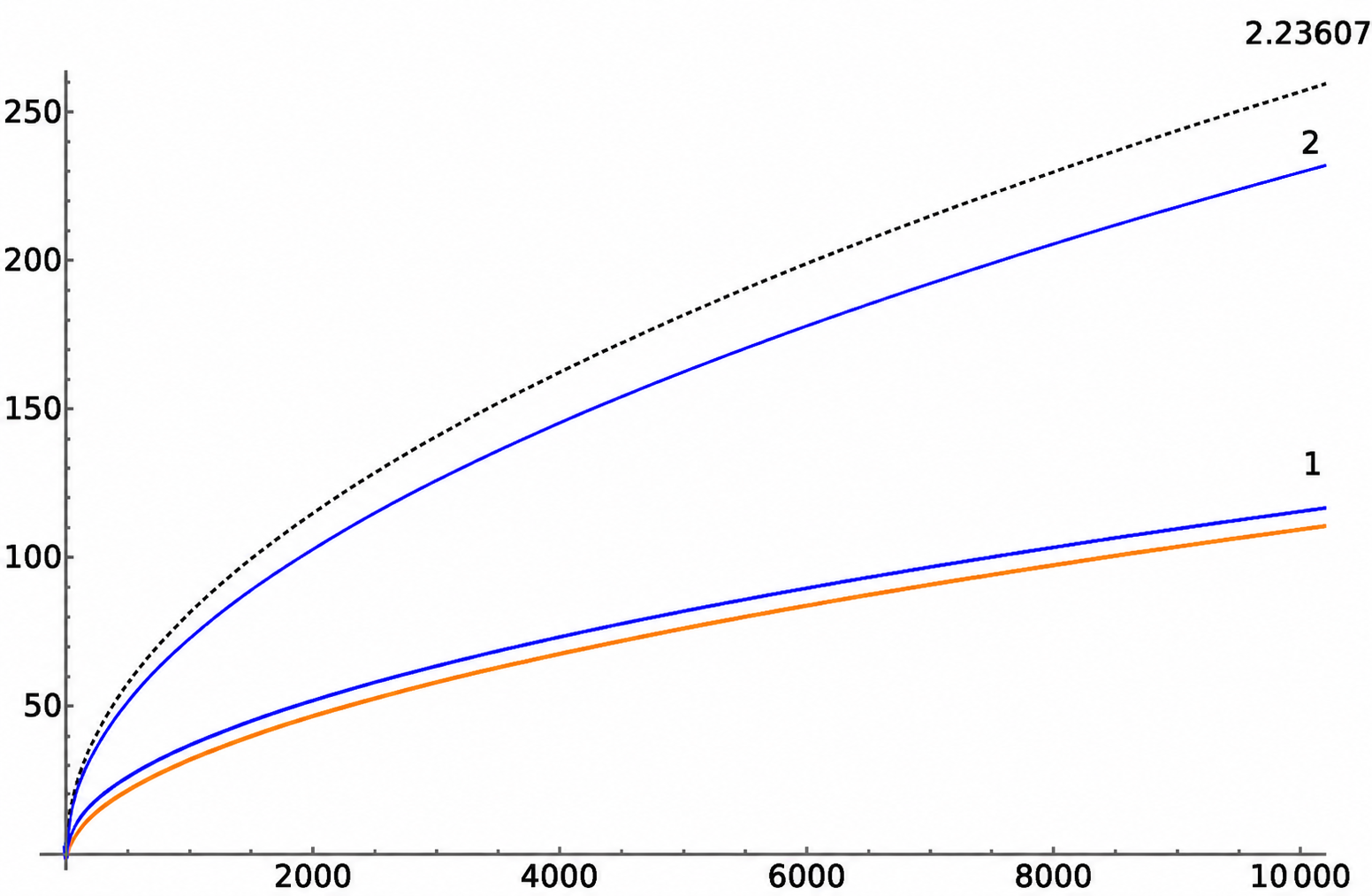}
 \caption{$\ceff$ data for $\overline{\Sigma(2,3,5)}$: $\log |a_n|$ as a~function of $n$ (\textcolor{orange}{orange}) producing a~lower bound estimate~${m_{\min}\approx 0.9}$, reference curves \smash{$\sqrt{\left(4\pi\right)^2\frac{m^2}{120}n}$} for $m\in\{1,2\}$ (\textcolor{blue}{blue}), the upper bound \smash{$\sqrt{\frac{2\pi^2}{3}n}$} (dashed black). The exact value $m=1$ was obtained in~\cite{Gukov:2023cog} using the mock modular property of $\zhat$. In this case, $\ceff$ corresponds to the smallest numerical Chern--Simons invariant of a~non-abelian flat connection: $\rm{CS}(\alpha_2) = \frac{1}{120}$, Table~\ref{tab:Invariants of flat connections on Poincare sphere}. Therefore, the conjectural relation~\eqref{coj:z-hat at leading order} holds in this example.}
 \label{fig:ceff235}
\end{figure}

\begin{table}\renewcommand{\arraystretch}{1.2}
\centering
\begin{tabular}{||c|c||}
 \hline
 $\alpha$ & $\rm{CS}(\alpha)$ \\
 \hline\hline
 $\alpha_0$ & $[0]$ \\
 \hline
 $\alpha_1$ & $[\frac{49}{120}]$ \\
 \hline
 $\alpha_2$ & $[\frac{1}{120}]$ \\
 \hline
\end{tabular}
\caption{Invariants of flat connections on $\brs{2}{3}{5}$.
}
\label{tab:Invariants of flat connections on Poincare sphere}
\end{table}

There are two blue curves corresponding to $m=1$ and $m=2$ below the upper bound (black dashed curve). The estimated values of $m$ for $N\in\big[1,10^4\big]$ is shown as a~thick orange line, and it provides the lower bound for $c_{\rm eff}$. We can see both blue lines appearing within the bounds, however, asymptotically $m=1$ is getting much closer. This suggests $m=1$, and the $c_{\rm eff}$-connection is the one with the smallest Chern--Simons number.
In this case, we can access this result analytically using the modular properties of $\zhat\bigl(\overline{\Sigma(2,3,5)};q\bigr)$, as shown in~\cite{Gukov:2023cog}.

We will conduct a~more thorough study of \smash{$\widehat{Z}$} for \smash{$\brs{s}{t}{rst+1}$}
using its mixed modular properties in the next section, as well as the detailed analysis of various examples $\brs{s}{t}{rst\pm1}$ in Section~\ref{sec:Case studies}.

\section{Mixed mock modularity and resurgence}\label{sec:Mixed mock modularity and resurgence}
Despite being defined as half-indices of three-dimensional $\mathcal{N}=2$ theories on a~solid torus, $\zhat$~invariants often retain striking remnants of the full modular symmetry enjoyed by ordinary two-dimensional indices on the torus boundary.
Exact modular symmetry is broken, but the surviving structure turns out to be surprisingly rich.

For Seifert homology spheres with three singular fibres (in plumbing language: star-shaped graphs with three legs), the $\widehat{Z}$ invariants were shown to be quantum modular forms~\cite{cheng20193d,zagier2010quantum}~-- functions of the cusps of the upper-half complex plane whose obstruction to modularity can be extended to a~real analytic function.
In the same cases, upon orientation reversal that obstruction to modularity further trivializes~\cite{cheng20193d} and the invariant becomes a~{mock modular form}: the non-modular holomorphic part of a~modular non-holomorphic object which is obtained by the addition of an explicit, non-holomorphic Eichler integral of a~cusp form. A~curious empirical relation between $\widehat{Z}$-series of manifolds related by orientation reversal culminated in the false--mock conjecture~\cite{cheng2020three}, predicting a~precise pairing between the $\zhat$ invariants of manifolds related by orientation reversal in terms of their modular structure.

In this section, we leverage a~proposal~\cite{Cheng:2024vou} for $\zhat$ invariants on positive Brieskorn spheres of the type $\Sigma(s,t,str+1)$ in order to compute the $\ceff$.
We do this by constructing a~modular completion for the $\zhat$ and use its modular transformations to re-express $\zhat$ in terms of $\tilde{q}$ series.
We also demonstrate key differences that arise when studying these $\tilde{q}$ series as obtained from the proposal of~\cite{Cheng:2024vou}, which we use, and a~recent proposal for the same object deriving from resurgence analysis.

This section has the following structure:
In Section~\ref{sec:ceffmod}, we identify the completion of the~$\zhat$ invariant and use it to re-write the invariant in terms of a~sum of three $q$ and $\tilde{q}$ series, $W\bigl(-\tfrac{1}{\tau}\bigr)$, $ZG^{\mathsf{c}_1}(\tau,\overline{\tau})$ and $ZG^{\mathsf{c}_2}(\tau,\overline{\tau})$.
In Section~\ref{sec:asympt}, we analyze each of the three terms. We construct polynomials with the same asymptotic divergence as $W\bigl(-\tfrac{1}{\tau}\bigr)$ and $ZG^{\mathsf{c}_2}(\tau,\overline{\tau})$ and we show that~$ZG^{\mathsf{c}_1}(\tau,\overline{\tau})$ does not diverge as $\tau={\rm i}\alpha\to0^{+}$, so it does not contribute to the $\ceff$.
Lastly, in Section~\ref{sec:compare}, we compare our results with recent developments from resurgence and demonstrate that the modular $S$-transformation leads to terms which do not have an analogue in that formalism.

\subsection[ceff from modularity]{$\boldsymbol{\ceff}$ from modularity}\label{sec:ceffmod}

To find the $\ceff$ of $\zhat\bigl(\overline{\Sigma(s,t,r s t +1)};q\bigr)$, we need to study the asymptotics of, $\widehat{Z}\bigl(\overline{\Sigma(s, t, str +1)};\allowbreak q\bigr)/\eta(q)$ as $q \rightarrow 1$, where $\eta(\tau)$ is the familiar Dedekind $\eta$ function.
We can get the asymptotics by looking at the modular $S$-transform of \smash{$\widehat{Z}\bigl(\overline{\Sigma(s, t, str +1)};q\bigr)$}. To obtain the modular $S$-transform of the \smash{$\zhat$}, we will introduce the (non-holomorphic) modular completion of the \smash{$\widehat{Z}$} invariant by taking the modular completions of each summand in~\eqref{Zhatindefinitetheta}
\[
 \mathring{Z}_{b_0}\bigl(\overline{\Sigma(p_1,p_2,p_3)};\tau\bigr) := \frac{C_{\Gamma}\bigl(q^{-1}\bigr)}{2f_{2r+1,1}(\tau)} \sum_{\substack{\hat\epsilon=(\epsilon_1,\epsilon_2,\epsilon_3)\\\in (\ZZ/2\ZZ)^{\otimes 3}}}
(-1)^{\hat\epsilon} \mathrm{e}^{-2 \pi {\rm i} B(\mathsf{a}_{\hat\epsilon},\mathsf{b})} \widehat\Theta_{A,\mathsf{a}_{\hat\epsilon},\mathsf{b},\mathsf{c}_1,\mathsf{c}_2}(\tau).
\]
Since $2\mathsf{b} \in L^{*}$ and $\mathsf{a}_{\hat\epsilon} - \mathsf{b} \in L^{*}$, we can use the modular $S$-transformation given in equation~\eqref{modularStransformindivtheta} to write
\begin{align*}
 \frac{2 \mathring{Z}_{b_0}\bigl(\overline{\Sigma(p_1,p_2,p_3)};\tau\bigr)}{C_{\Gamma}\bigl(q^{-1}\bigr)} ={}&
 \frac{- \tau^{-1} }{\sqrt{|\det(A)|}} \frac{\mathrm{e}^{-2 \pi {\rm i} B(\mathsf{b},\mathsf{b})}}{f_{2r+1,1}(\tau)} \sum_{\substack{\hat\epsilon=(\epsilon_1,\epsilon_2,\epsilon_3)\\\in (\ZZ/2\ZZ)^{\otimes 3}}} \sum_{ \lambda \in L^{*} }
(-1)^{\hat\epsilon} \mathrm{e}^{- 2\pi {\rm i} B(\mathsf{a}_{\hat\epsilon}, \lambda + \mathsf{b})} \\
 & \times\mathrm{e}^{-2 \pi {\rm i} B( \lambda, \mathsf{b})} \widehat{\Theta}_{A,\lambda + \mathsf{b}, \mathsf{b}, c_{1}, c_{2}} \left(-\frac{1}{\tau},- \frac{1}{\bar{\tau}}\right).
\end{align*}
We can simplify the above expression by summing over $\hat{\epsilon}$. We pick the following representatives of~$L^{*}$:
\[
 L^{*} = \left\{ \left(\frac{\ell_{1}}{2 s t (r s t+1)}, \frac{\ell_{2}}{2r+1}\right)\, \Bigg|\, 0 \leq \ell_{1} < 2 s t(r s t+1),\, 0 \leq \ell_{2}< 2r+1 \right\}.
\]
We can simplify the sum over $\hat{\epsilon}$ with the following identity:
 \begin{gather*}
 \frac{1}{(2{\rm i})^{3}}\sum_{\hat{\epsilon}\in (\ZZ/2\ZZ)^{\otimes 3}} (-1)^{\hat{\varepsilon}}\mathrm{e}^{-2\pi {\rm i} B(\mathsf{b}, \mathsf{b})}\mathrm{e}^{-2\pi {\rm i} B(\mathbf{a}_{\hat{\varepsilon}},\lambda + \mathsf{b})}\mathrm{e}^{-2\pi {\rm i} B(\lambda,\mathsf{b})}\\
 \qquad=(-1)^{\ell_1}\sin(\frac{\ell_1 \pi}{s})\sin(\frac{\ell_1 \pi}{t})\sin(\frac{\ell_1 \pi}{str+1}).
 \end{gather*}
We can write
\begin{align*}
 -\frac{ \tau \sqrt{|\det(A)|} }{C_{\Gamma}\bigl(q^{-1}\bigr)} \mathring{Z}_{b_0}\bigl(\overline{\Sigma(p_1,p_2,p_3)};\tau\bigr) ={}&
 \frac{4 {\rm i} }{f_{2r+1,1}(\tau)} \sum_{\ell_{1}=0}^{2 s t(r s t+1)-1} (-1)^{\ell_{1}} \sin(\frac{ \ell_{1} \pi }{s}) \sin(\frac{ \ell_{1} \pi }{t})\\
 &\times\sin(\frac{ \ell_{1} \pi }{r s t+1})
 \sum_{\ell_{2}=0}^{2r}\widehat{\Theta}_{A,\lambda + \mathsf{b}, \mathsf{b}, c_{1}, c_{2}} \left(-\frac{1}{\tau},- \frac{1}{\bar{\tau}}\right).
\end{align*}
Just as the completion of indefinite theta series that comprise that of $\widehat{Z}$ can be written as a~sum of three terms, we can write the completion of $\widehat{Z}$ as a~sum of three terms, $\widehat{Z}$, a~non-holomorphic term that depends on $c_{1}$ and a~non-holomorphic term that depends on $c_{2}$. We will use this grouping to write $\widehat{Z}$ as a~sum of three terms
\begin{equation}\label{zhatassumofthreeterms}
 \widehat{Z} = W(\tau) + ZG^{c_{1}}(\tau, \bar{\tau}) - ZG^{c_{2}}(\tau, \bar{\tau}),
\end{equation}
where
\begin{gather}
 W(\tau) = - \frac{C_{\Gamma}\bigl(q^{-1}\bigr) 4 {\rm i} }{ \tau \sqrt{|\det(A)|} f_{2r+1,1}(\tau) } \sum_{\ell_{1}=0}^{2 s t(r s t+1)-1} (-1)^{\ell_{1}} \sin(\frac{ \ell_{1} \pi }{s}) \sin(\frac{ \ell_{1} \pi }{t}) \sin(\frac{ \ell_{1} \pi }{r s t+1})\nonumber \\
 \phantom{ W(\tau) = }{}\times \sum_{\ell_{2}=0}^{2r} \Theta_{A,\lambda + \mathsf{b}, \mathsf{b}, c_{1}, c_{2}} \left(-\frac{1}{\tau},- \frac{1}{\bar{\tau}}\right),\label{eq:Wqtilde}
\\
 ZG^{c}(\tau, \bar{\tau}) = \frac{C_{\Gamma}\bigl(q^{-1}\bigr)}{2 f_{2r+1,1}(\tau)} \Bigg[ \frac{8 {\rm i} }{ (-\tau) \sqrt{|\det(A)|} } \nonumber\\
\phantom{ZG^{c}(\tau, \bar{\tau}) = }{} \times \sum_{\ell_{1}=0}^{2 s t(r s t+1)-1} (-1)^{\ell_{1}} \sin(\frac{ \ell_{1} \pi }{s}) \sin(\frac{ \ell_{1} \pi }{t}) \sin(\frac{ \ell_{1} \pi }{r s t+1}) \nonumber\\
\phantom{ZG^{c}(\tau, \bar{\tau}) = }{} \times \sum_{\ell_{2}=0}^{2r} G_{A, \lambda + \mathsf{b}, \mathsf{b},c}\left(-\frac{1}{\tau}, -\frac{1}{\bar{\tau}}\right)
 - \sum_{\substack{\hat\epsilon=(\epsilon_1,\epsilon_2,\epsilon_3)\\ \in (\ZZ/2\ZZ)^{\otimes 3}}}
(-1)^{\hat\epsilon} \mathrm{e}^{-2 \pi {\rm i} B(\mathsf{a}_{\hat\epsilon},\mathsf{b})} G_{A,\mathsf{a}_{\hat\epsilon},\mathsf{b},\mathsf{c}}(\tau, \bar{\tau})\Bigg].\nonumber
\end{gather}

\subsection{Asymptotic structure of the modular completion}\label{sec:asympt}

We will determine the leading order of divergence of the $\zhat$ invariant by analyzing each of the three terms of equation~\eqref{zhatassumofthreeterms}.
Due to the technical nature of the proofs, here we shall summarize the main results, and relegate their proofs to Appendix~\ref{apx:proofs}.
The main strategy will be that of isolating the leading terms from the $\tilde{q}$ series of $W(\tau)$ and $ZG^{\mathsf{c}_2}$ to construct polynomials with the same asymptotic properties.
We will also show that $ZG^{\mathsf{c}_1}$ does not diverge and thus does not contribute to $\ceff$.
Our main result will be the identification of a~finite polynomial whose leading order divergence is the same as the $\zhat$ invariant proposal, as summarized in the following lemma.\looseness=1

\begin{Lemma}\label{lem:asympt}
 Let $\zhat\bigl(\overline{\Sigma(s,t,str+1)};\mathrm{e}^{2\pi {\rm i} \tau}\bigr)$ be the $\zhat$ invariant for a~positive Brieskorn sphere as described in Section~{\rm\ref{sec:preliminaries}}, then
\smash{$\zhat\bigl(\overline{\Sigma(s,t,str+1)};\mathrm{e}^{2\pi {\rm i} \tau}\bigr)\sim \zeta(s,t,str+1;\tau)$} as \smash{$q=\mathrm{e}^{2\pi {\rm i}\tau}\to 1$},
 where
\begin{gather*}
 \zeta(s,t,str+1;\tau)
 \\
 \qquad= - \frac{C_{\Gamma}\bigl(q^{-1}\bigr) 4 {\rm i} }{ \tau \sqrt{|\det(A)|} f_{2r+1,1}(\tau) }\sum_{\ell_{1}=0}^{2 s t(r s t+1)-1} (-1)^{\ell_{1}} \sin(\frac{ \ell_{1} \pi }{s}) \sin(\frac{ \ell_{1} \pi }{t}) \sin(\frac{ \ell_{1} \pi }{r s t+1})\\
 \phantom{\qquad=}{}\times\sum_{\ell_{2}=0}^{2r} \mathcal{P}(s,t,r,\lambda;\tau)+ \frac{4 {\rm i} C_{\Gamma}\bigl(q^{-1}\bigr) \sqrt{- {\rm i} \tau} \sqrt{(2r+1)(s t -2)}}{(- \tau) \sqrt{8 s t} \sin(\frac{\pi r}{2r+1}) }\\
 \phantom{\qquad=}{} \times
 \sum_{\ell_{1}=-L_{1}}^{L_{1}} \sum_{\ell_{2}=-L_{2}}^{L_{2}} (-1)^{\ell_{1}} \sin(\frac{ \ell_{1} \pi }{s}) \sin(\frac{ \ell_{1} \pi }{t}) \sin(\frac{ \ell_{1} \pi }{r s t+1}) \\
 \phantom{\qquad=}{} \times\frac{ \mathrm{e}^{\frac{{\rm i} \pi (2 \ell_{2}+1)}{2(2r+1)}} \operatorname{sgn}\bigl(B\bigl(\mathsf{c}_{2}, \bigl(\frac{\ell_{1}}{2 s t (r s t+1)}, \frac{2 \ell_{2}+1}{2(2r+1)}\bigr) \bigr)\bigr) \tilde{q}^{ Q_{\mathsf{c}_{2}} \bigl( \frac{\ell_{1}}{2 s t (r s t+1)}, \frac{2 \ell_{2}+1}{2(2r+1)} \bigr) - \frac{1}{8(2r+1)} }}{\pi \big|B\bigl(\mathsf{c}_{2}, \bigl(\frac{\ell_{1}}{2 s t (r s t+1)}, \frac{2 \ell_{2}+1}{2(2r+1)}\bigr) \bigr)\big| } \\
 \phantom{\qquad=\times}{}
 +\frac{ {\rm i} C_{\Gamma}\bigl(q^{-1}\bigr) \sqrt{2r+1} }{4 \pi \sin(\frac{\pi r}{2r+1}) \sqrt{ s t (s t-2) } }
 \sum_{\substack{\hat\epsilon=(\epsilon_1,\epsilon_2,\epsilon_3)\\\in (\ZZ/2\ZZ)^{\otimes 3}}}
(-1)^{\hat\epsilon} \mathrm{e}^{-2 \pi {\rm i} B(\mathsf{a}_{\hat\epsilon},\mathsf{b})} \sum_{I \in P_{2}} \sum_{k = L_{3} }^{L_{4}} \tilde{q}^{ \frac{(2k+ s t)^{2}}{8 s t (s t -2)} - \frac{1}{8(2r+1)} } \endline
 \phantom{\qquad=}{} \times \log( \mathrm{e}^{ {\rm i} \pi ( \frac{ B(I+\mathsf{a}_{\hat\epsilon},\mathsf{c}_{2}) }{(2r+1)(s t-2)(r s t +1)} +1 ) } ) \mathrm{e}^{- {\rm i} \pi \frac{B(I+\mathsf{a}_{\hat\epsilon},\mathsf{c}_{2})}{(2r+1)(s t -2)} - \frac{ 2 \pi {\rm i} k B(I+\mathsf{a}_{\hat\epsilon}, c_{2, \perp}) }{s t (s t -2)} },
\end{gather*}
where $\lambda=(\ell_1,\ell_2)$,
\begin{equation*}
\mathcal{P}(s,t,r,\mu;\tau)
=
\begin{cases}
\displaystyle
\sum_{(\ell_1,\ell_2)\in\{(0,0),(-1,-1)\}}
q^{\frac12 R(\mu,s,t,r;\ell_1,\ell_2)},
& \delta\neq 0,\vspace{1mm}\\
\displaystyle
\sum_{(\ell_1,\ell_2)\in\{(0,0),(-1,0)\}}
q^{\frac12 R(\mu,s,t,r;\ell_1,\ell_2)},
& \delta=0
\end{cases}
 \end{equation*}
 for
 \begin{align*}
 \begin{split}
 R(\alpha_1,\alpha_2,r,s,t;\ell_1,\ell_2)={}&(2 r+1) \left(-\left\lfloor \frac{1}{2} \left(\frac{2 \alpha_2+1}{2 r+1}-\frac{\alpha_1}{r s t+1}\right)\right\rfloor +\ell_1 s t+\ell_2+\frac{2 \alpha_2+1}{4 r+2}\right)^2\\
 &-2 s t (r s t+1) \left(\ell_1+\frac{\alpha_1}{2r s^2 t^2+2 s t}\right)^2
 \end{split}
 \end{align*}
 and
\smash{$
 \delta =\frac{2\alpha_2 + 1}{2(2r+1)}- \frac{\alpha_1}{2(rst+1)}$}.
 The quadratic form $Q_{\mathsf{c}_{2}}$ is given by
\begin{equation*}
 Q_{\mathsf{c}_{2}} (n) \coloneq \frac{|n|^{2}}{2} - \frac{(B(\mathsf{c}_{2},n))^{2}}{|\mathsf{c}_{2}|^{2}}.
\end{equation*}
The set $P_{2}$ is defined by $P_2\coloneq\{ (0,x)\mid 1 \leq x \leq s t-2\} $ and the bounds of summation $L_1$, $L_2$, $L_3$ and $L_4$ are given by
\begin{gather*}
 L_{1} = \left\lceil \frac{\sqrt{2 s t(r s t+1)(4 r s t+ s t+ 2)}}{2\sqrt{(2r+1)(s t -2)}} \right\rceil ,\qquad L_{2}= \left\lceil \frac{\sqrt{4 r s t+ s t+2}}{2\sqrt{s t-2}} - \frac{1}{2} \right\rceil,\\
 L_{3} = \left\lceil -\frac{s t}{2} - \sqrt{\frac{s t(s t -2)}{4(2r+1)}} \right\rceil,\qquad L_{4} = \left\lfloor -\frac{s t}{2} + \sqrt{\frac{s t(s t -2)}{4(2r+1)}} \right\rfloor.
\end{gather*}
\end{Lemma}

\subsubsection*{Asymptotics of $\boldsymbol{ W(\tau)}$}

To capture the leading divergence of $W(\tau)$ as $\tau\to0^{+}$ \big($\tilde q=\mathrm{e}^{2\pi {\rm i}/\tau}\to0$\big), we replace every indefinite theta function that appears in its definition by a~{two-term polynomial} in $\tilde q$.
The construction proceeds in two steps.
Within the summation cone~\eqref{eq:indef-theta-cone}, the $\tilde q$-exponent is monotone: it increases with the distance of the lattice point $n$ from the origin.
Hence only the two points that lie closest to the origin~-- one in each connected component of the cone~-- govern the smallest power of $\tilde q$.
So keeping just those two contributions yields a~polynomial of two monomials whose leading $\tilde q$-power coincides with that of the full theta function.
Substituting these two-term approximants for every indefinite theta series in $W$ furnishes a~single polynomial that reproduces the asymptotic divergence of $W\!\bigl(\tau\bigr)$. The precise form of the approximant for an individual theta function is stated in the following lemma.

\begin{Lemma}
\label{lem:leadingorderofW}
 Using the same conventions as Section~{\rm\ref{sec:preliminaries}}, the leading order of $\Theta_{A,\mu+b,b,\mathsf{c}_1,\mathsf{c}_2}$, where
 \[
 \mu = \left( \begin{matrix} \frac{\alpha_1}{2st(str+1)} \\ \frac{\alpha_2}{2r+1} \end{matrix} \right),
 \]
 with $1\leq \alpha_1 < 2st(str+1)$, $0\leq \alpha_2 < 2r+1 $, is the same as
 \[
\mathcal{P}(s,t,r,\mu;\tau)
=
\begin{cases}
\displaystyle
\sum_{(\ell_1,\ell_2)\in\{(0,0),(-1,-1)\}}
q^{\frac12 R(\mu,s,t,r;\ell_1,\ell_2)},
& \delta\neq 0,\vspace{1mm}\\
\displaystyle
\sum_{(\ell_1,\ell_2)\in\{(0,0),(-1,0)\}}
q^{\frac12 R(\mu,s,t,r;\ell_1,\ell_2)},
& \delta=0,
\end{cases}
 \]
 where
 \begin{align*}
 R(\alpha_1,\alpha_2,r,s,t;\ell_1,\ell_2)={}&(2 r+1) \left(-\left\lfloor \frac{1}{2} \left(\frac{2 \alpha_2+1}{2 r+1}-\frac{\alpha_1}{r s t+1}\right)\right\rfloor +\ell_1 s t+\ell_2+\frac{2 \alpha_2+1}{4 r+2}\right)^2\\
& -2 s t (r s t+1) \left(\ell_1+\frac{\alpha_1}{2r s^2 t^2+2 s t}\right)^2
 \end{align*}
 and
\smash{$
 \delta =\frac{2\alpha_2 + 1}{2(2r+1)}- \frac{\alpha_1}{2(rst+1)}$}.
\end{Lemma}
Note that we are not considering the case where $\alpha_1=0 \mod 2st(str+1)$.
This is sufficient for our goals, as if $\alpha_1$ were zero, then the sines in~\eqref{eq:Wqtilde} vanish, so the $\tilde{q}$-power should not be considered.

\subsubsection*{Asymptotics of $\boldsymbol{ZG^{\textsf{c}_1}}$}
The term $ZG^{\mathsf{c}_1}$ can be shown to be finite, thus not contributing to the $\ceff$.
The first step to do so is to perform a~rewriting of $ZG^{\mathsf{c}_1}$, to
\begin{align}
 ZG^{c_{1}}(\tau, \bar{\tau} ) = {}& \frac{C_{\Gamma}\bigl(q^{-1}\bigr)}{2} \nonumber\\
 &\times\sum_{\substack{\hat\epsilon=(\epsilon_1,\epsilon_2,\epsilon_3)\\\in (\ZZ/2\ZZ)^{\otimes 3}}}
(-1)^{\hat\epsilon} \left[ \frac{\rm i}{\sqrt{2 s t (r s t+1)} } \int_{0}^{{\rm i} \infty} \frac{ \theta^{1}_{s t (r s t+1),2s t (r s t+1)(\mathsf{a}_{\hat\epsilon})_{1} }(y) }{\sqrt{-{\rm i}(y + \tau)}} \dd y \right].\label{eq:zgc1rewriting}
\end{align}
Then, after substituting in the definition of the unary theta function $\theta^{1}_{m,r}(\tau)$, its modular properties can be leveraged to show that the integral in~\eqref{eq:zgc1rewriting} converges.

\subsubsection*{Asymptotics of $\boldsymbol{ZG^{\textsf{c}_2}}$}
The asymptotic order of divergence of $ZG^{\textsf{c}_2}$ is obtained in a~similar way to that of $W(\tau)$.
First, the sum is separated into two terms
\begin{gather*}
 ZG^{\textsf{c}_2}_{1}(\tau, \bar{\tau}) \coloneq \frac{4{\rm i} C_{\Gamma}\bigl(q^{-1}\bigr) }{ (-\tau) \sqrt{|\det(A)|} f_{2r+1,1}(\tau) }
 \sum_{\ell_{1}=0}^{2 s t(r s t+1)-1} (-1)^{\ell_{1}} \sin(\frac{ \ell_{1} \pi }{s}) \sin(\frac{ \ell_{1} \pi }{t}) \\
\phantom{ ZG^{\textsf{c}_2}_{1}(\tau, \bar{\tau}) \coloneq }{}
 \times\sin(\frac{ \ell_{1} \pi }{r s t+1}) \sum_{\ell_{2}=0}^{2r} G_{A, \lambda + \mathsf{b}, \mathsf{b},\mathsf{c}_{2}}\left(-\frac{1}{\tau}, -\frac{1}{\bar{\tau}}\right),
\\
 ZG^{\textsf{c}_2}_{2}(\tau, \bar{\tau}) \coloneq \frac{C_{\Gamma}\bigl(q^{-1}\bigr)}{2 f_{2r+1,1}(\tau)} \sum_{\substack{\hat\epsilon=(\epsilon_1,\epsilon_2,\epsilon_3)\\
 \in (\ZZ/2\ZZ)^{\otimes 3}}}
(-1)^{\hat\epsilon} \mathrm{e}^{-2 \pi {\rm i} B(\mathsf{a}_{\hat\epsilon},\mathsf{b})} G_{A,\mathsf{a}_{\hat\epsilon},\mathsf{b},\mathsf{c}_{2}}(\tau, \bar{\tau})
\end{gather*}
and each term in the sum is considered independently.
We show that the leading order divergence of the first term $ZG^{\mathsf{c}_2}_{1}(\tau,\bar{\tau})$ is
\begin{align*} 
 ZG^{\textsf{c}_2}_{1}(\tau, \bar{\tau})
 \approx{} & \frac{4{\rm i} C_{\Gamma}\bigl(q^{-1}\bigr) \sqrt{-{\rm i} \tau} \sqrt{(2r+1)(s t -2)}}{(- \tau) \sqrt{8 s t} \sin(\frac{\pi r}{2r+1}) } \\
 &
 \times \sum_{\ell_{1}=-L_{1}}^{L_{1}} \sum_{\ell_{2}=-L_{2}}^{L_{2}} (-1)^{\ell_{1}} \sin(\frac{ \ell_{1} \pi }{s}) \sin(\frac{ \ell_{1} \pi }{t}) \sin(\frac{ \ell_{1} \pi }{r s t+1}) \\
 &
 \times \frac{ \mathrm{e}^{\frac{{\rm i} \pi (2 \ell_{2}+1)}{2(2r+1)}} \operatorname{sgn}\bigl(B\bigl(\mathsf{c}_{2}, \bigl(\frac{\ell_{1}}{2 s t (r s t+1)}, \frac{2 \ell_{2}+1}{2(2r+1)}\bigr) \bigr)\bigr) \tilde{q}^{ Q_{\mathsf{c}_{2}} \bigl( \frac{\ell_{1}}{2 s t (r s t+1)}, \frac{2 \ell_{2}+1}{2(2r+1)} \bigr) - \frac{1}{8(2r+1)} }}{\pi \big|B\bigl(\mathsf{c}_{2}, \bigl(\frac{\ell_{1}}{2 s t (r s t+1)}, \frac{2 \ell_{2}+1}{2(2r+1)}\bigr) \bigr)\big| },
\end{align*}
where the bounds $L_{1}$, and $L_{2}$ are given by
\begin{align*}
 L_{1} = \left\lceil \frac{\sqrt{2 s t(r s t+1)(4 r s t+ s t+ 2)}}{2\sqrt{(2r+1)(s t -2)}} \right\rceil,\qquad L_{2}= \left\lceil \frac{\sqrt{4 r s t+ s t+2}}{2\sqrt{s t-2}} - \frac{1}{2} \right\rceil.
\end{align*}
Similarly, the leading order divergence of $ZG^{\mathsf{c}_2}_2(\tau,\bar{\tau})$ is
\begin{align*}
 ZG^{\textsf{c}_2}_{2}(\tau, \bar{\tau})
\approx {}& \frac{{\rm i} C_{\Gamma}\bigl(q^{-1}\bigr) \sqrt{2r+1} }{4 \pi \sin\bigl(\frac{\pi r}{2r+1}\bigr) \sqrt{ s t (s t-2) } } \sum_{\substack{\hat\epsilon=(\epsilon_1,\epsilon_2,\epsilon_3)\\\in (\ZZ/2\ZZ)^{\otimes 3}}}
(-1)^{\hat\epsilon} \mathrm{e}^{-2 \pi {\rm i} B(\mathsf{a}_{\hat\epsilon},\mathsf{b})} \sum_{I \in P_{2}} \sum_{k = L_{3} }^{L_{4}} \tilde{q}^{ \frac{(2k+ s t)^{2}}{8 s t (s t -2)} - \frac{1}{8(2r+1)} } \\
 & \times \log( \mathrm{e}^{ {\rm i} \pi ( \frac{ B(I+\mathsf{a}_{\hat\epsilon},\mathsf{c}_{2}) }{(2r+1)(s t-2)(r s t +1)} +1 ) } ) \mathrm{e}^{- {\rm i} \pi \frac{B(I+\mathsf{a}_{\hat\epsilon},\mathsf{c}_{2})}{(2r+1)(s t -2)} - \frac{ 2 \pi {\rm i} k B(I+\mathsf{a}_{\hat\epsilon}, c_{2, \perp}) }{s t (s t -2)} },
\end{align*}
where
\[
 L_{3} = \left\lceil -\frac{s t}{2} - \sqrt{\frac{s t(s t -2)}{4(2r+1)}} \right\rceil,\qquad L_{4} = \left\lfloor -\frac{s t}{2} + \sqrt{\frac{s t(s t -2)}{4(2r+1)}} \right\rfloor.
\]

\subsection{Resurgence and mock modularity}\label{sec:compare}

Along with the surgery formula from~\cite{park2021inverted} and modularity~\cite{Cheng:2024vou}, one can obtain $\widehat{Z}$ on the other side using resurgence techniques. This approach was carried out in~\cite{Adams:2025aad,Costin:2023kla}. In this section, we will briefly summarize the resurgence approach and compare it to the approach using regularized surgery formula for $\widehat{Z}$.

In~\cite{Costin:2023kla}, a~new $q$-series, $\Psi^{\lor}_{a,p}(q)$, called dual false theta functions, were introduced. $\Psi^{\lor}_{a,p}(q)$ is obtained by performing the $q \rightarrow q^{-1}$ operation on a~false theta function $\Psi_{a,p}$ using resurgence techniques. The dual false theta functions are expected to satisfy the following equation:
\begin{equation}\label{resurgancedualfalseStransform}
 \Psi^{\lor}_{a,p}(q) = \frac{ {\rm i}}{\sqrt{2 {\rm i} \tau}} \sum_{b =1}^{\lfloor \frac{p}{2} \rfloor} S_{a b} \tilde{q}^{\tilde{\Delta}_{b}} W_{b}(\tilde{q}) - \frac{1}{\pi} \sqrt{\frac{2 {\rm i} p}{\tau}} \int_{0}^{\infty} \mathrm{d} u \mathrm{e}^{\frac{p u^{2}}{2 \pi {\rm i} \tau}} \frac{\sinh((p-a)u)}{\sinh( p u)}.
\end{equation}
For further details, see \cite[equation~(4.95)]{Costin:2023kla}. Just as $\widehat{Z}$ for negative definite Brieskorn spheres is a~linear combination of false theta functions, according to the proposal in~\cite{Adams:2025aad}, the $\widehat{Z}$ for positive definite Brieskorn spheres is the same linear combination of dual false theta functions. Therefore, the $\widehat{Z}$ for positive definite Brieskorn spheres $\overline{\Sigma(p_{1}, p_{2}, p_{3})}$ satisfies a~linear combination of the equation above~\eqref{resurgancedualfalseStransform}. Schematically, one can write
\begin{align}
 \widehat{Z}\bigl(\overline{\Sigma(p_1,p_2,p_3)};q\bigr) ={}&\frac{\rm i}{\sqrt{2 {\rm i} \tau}} \sum_{ a\in \mathcal{M}_{nAb} } S_{ 0 a} \tilde{q}^{\tilde{\Delta}_{a}} W_{b}(\tilde{q}) \nonumber\\
 &- \frac{4}{\pi} \sqrt{ \frac{2 {\rm i} p }{\tau}} \int_{0}^{\infty} \mathrm{d} u \mathrm{e}^{\frac{p u^{2}}{2 \pi {\rm i} \tau}} \frac{\sinh(p_{1} p_{2} u )\sinh(p_{2} p_{3} u) \sinh(p_{1} p_{3} u)}{ \sinh(p u)},\label{resurganceZhatStransform}
\end{align}
where $p = p_{1} p_{2} p_{3}$, $\mathcal{M}_{nAb}$ is the set of non-abelian $SL(2,\mathbb{C})$ flat connections on $\overline{\Sigma(p_{1}, p_{2}, p_{3})}$, $W_{b}(\tilde{q})$ is a~$\tilde{q}$-series, and $\tilde{\Delta}_{b}$ are rational numbers that are integer lifts of Chern--Simons values of non-abelian flat connections. Here, the product of three hyperbolic sines in the integrand of~\eqref{resurganceZhatStransform} arises from the underlying linear combination of dual false theta functions through the identity
\[
\sum_{\epsilon_{1}, \epsilon_{2} \in \mathbb{Z}/2\mathbb{Z} } \epsilon_{1}\epsilon_{2} \sinh((p_{2}p_{3}+\epsilon_{1}p_{1}p_{2} +\epsilon_{2}p_{1}p_{3})u) = 4 \sinh(p_{1}p_{2}u) \sinh(p_{1}p_{3}u) \sinh(p_{2}p_{3}u).
\]

The above equation~\eqref{resurganceZhatStransform} is the modular $S$-transformation of the $\widehat{Z}$ for positive definite Brieskorn spheres obtained from resurgence techniques. To compare it with $\widehat{Z}$ for positive definite Brieskorn spheres obtained from the surgery formula, we have to look at the modular $S$-transform of $\widehat{Z}$ given in equation~\eqref{zhatassumofthreeterms}. Comparing the two equations~\eqref{zhatassumofthreeterms} and~\eqref{resurganceZhatStransform}, we see that the functions $W(\tau)$ and $ZG^{\mathsf{c}_{1}}(\tau, \bar{\tau})$ have the following analogues in the world of resurgence techniques:
\begin{gather*}
 W(\tau) \longleftrightarrow \frac{\rm i}{\sqrt{2 {\rm i} \tau}} \sum_{ a \in \mathcal{M}_{nAb} } S_{ 0 a} \tilde{q}^{\tilde{\Delta}_{a}} W_{b}(\tilde{q}), \\
 ZG^{\mathsf{c}_{1}}(\tau, \bar{\tau}) \longleftrightarrow - \frac{4}{\pi} \sqrt{ \frac{2 {\rm i} p }{\tau}} \int_{0}^{\infty} \mathrm{d} u \mathrm{e}^{\frac{p u^{2}}{2 \pi {\rm i} \tau}} \frac{\sinh(p_{1} p_{2} u )\sinh(p_{2} p_{3} u) \sinh(p_{1} p_{3} u)}{ \sinh(p u)}.
\end{gather*}
However, the term $ZG^{\mathsf{c}_{2}}(\tau, \bar{\tau})$ has no analogue in the resurgence framework. The term $ZG^{\mathsf{c}_{2}}(\tau, \bar{\tau})$ is associated with the \textit{mixed} mock modularity of $\widehat{Z}$.

Further, the number of $\tilde{q}$-series in
\[
\sum_{ a \in \mathcal{M}_{nAb} } S_{ 0 a} \tilde{q}^{\tilde{\Delta}_{a}} W_{b}(\tilde{q})
\]
 is equal to the number of non-abelian $SL(2,\mathbb{C})$ flat connections on $\overline{\Sigma(p_{1}, p_{2}, p_{3})}$. However, in the case of $\overline{\Sigma(s, t, s t r +1)}$, for~${r>1}$, the number of $\tilde{q}$-series in the term $ W(\tau)$ exceeds the number of non-abelian flat connections.

\section[ceff on the negative side]{$\boldsymbol{\ceff}$ on the negative side}\label{sec:ceff on the negative side}

So far we have been concerned with finding $\ceff$ for the family of Brieskorn spheres $\overline{\Sigma(s, t, rst \pm 1)}$ which are positive definite plumbed manifolds.
In this section, we take a~slight detour and consider negative definite plumbed manifolds which aren't necessarily integer homology spheres but whose $\widehat{Z}$ invariants admit a~precise mathematical description in terms of the plumbing data.
Similarly to Section~\ref{sec:Bounds on ceff}, our goal is to bound the growth of coefficients of $\widehat{Z}$.
We will prove that~${\ceff \leq 1}$ for this class of manifolds.

We begin with recalling the following definition and proposition from~\cite{akhmechet2023lattice}.

\begin{Definition}
 Consider the family of functions \smash{$\widehat{F} = \big\{\widehat{F}_n \mid \mathbb{Z} \to \mathbb{Q}\big\}_{n \geq 0}$} defined by setting $\widehat{F}_n(r)$ to be the coefficient of $z^{-r}$ in \smash{$\bigl(z-z^{-1}\bigr)^{2-n}$} for $ 0 \leq n \leq 2$. For $n \geq 3$, we define
 \begin{equation}\label{eq:def-f-n-r}
 \widehat{F}_n(r) = \begin{cases}
 \frac{1}{2}\operatorname{sgn}(r)^n \binom{\frac{n+|r|}{2}-2} {n-3} & \text{if } |r| \geq n-2, \text{ and } r \equiv n \mod 2 , \\
 0 & \text{otherwise}.
 \end{cases}
 \end{equation}
In the above $\operatorname{sgn}(r) \in \{-1, 1\}$ denotes the sign of $r$. Further, if $\Gamma$ is a~weighted tree with $L$ vertices and framing matrix $M$ and $\vec{b} \in 2\mathbb{Z}^L + \vec{\delta}$, then define $\widehat{F}_{\Gamma, \vec{b}} \colon \mathbb{Z}^L \to \mathbb{Q}$ by
\[
 \widehat{F}_{\Gamma, \vec{b}}(\vec{x}) = \prod_{v \in V(\Gamma)}\widehat{F}_{\deg(v)}((2M\vec{x}+\vec{b})_v),
 \]
 where $(-)_v$ denotes the $v$-th entry, $V(\Gamma)$ is the set of vertices and $\delta$ is the degree vector.
\end{Definition}

\begin{Proposition}
 Let $M_3$ be a~negative definite plumbed manifold and $b \in \operatorname{Spin}^c(M_3)$. Suppose~$M_3$ is obtained by plumbing along a~weighted tree $\Gamma$ on $L$ vertices with framing matrix $M$. Let~${\vec{b} \in 2\mathbb{Z}^L + \vec{\delta}}$ be a~$\operatorname{Spin}^c$ representative for $b$. Then
 \begin{equation}\label{z-hat}
 \widehat{Z}_b(M_3;q) = q^{-\frac{3L + \operatorname{Tr}(M)+\vec{b}^{\mathsf{T}}M^{-1}\vec{b}}{4}}\sum_{n \in \mathbb{Z}}\bigg(\sum_{\substack{\vec{x}\in D_n}} \widehat{F}_{\Gamma, \vec{b}} (\vec{x})\bigg)q^n,
 \end{equation}
 where\footnote{In~\cite{akhmechet2023lattice}, the modified prefactor \smash{$-\frac{1}{4}\left(3L + \operatorname{Tr}(M)+\vec{b}^{\mathsf{T}}M^{-1}\vec{b}\right)$} was denoted by $\Delta_b$ as a~result of a~different choice of notation. Note however this prefactor is not the minimal $q$-exponent in $\widehat{Z}_b$.}
 $D_n \coloneq \big\{\vec{x} \in \mathbb{Z}^L \mid 2\chi_{\vec{b}}(\vec{x}) = n\big\}$ and $\chi_{\vec{b}} \colon \mathbb{Z}^L \to \mathbb{Z}$ is defined by $\chi_{\vec{b}}(x) \coloneq - \frac{1}{2}\bigl( \vec{x}^{\mathsf{T}}M\vec{x} + \vec{x}^{\mathsf{T}}\vec{b}\bigr)$.
\end{Proposition}

\begin{Remark}\label{delta-shift}
 The reader may notice that the prefactor appearing in~\eqref{z-hat} differs from the one seen in~\cite{gukov2021two}. To avoid confusion in the indexing of coefficients, let $n_0$ be the smallest integer such that \smash{$\sum_{\substack{\vec{x}\in D_{n_0}}} \widehat{F}_{\Gamma, \vec{b}} (\vec{x}) \neq 0$}. Then the smallest $q$-exponent of $\widehat{Z}_b$ is $\Delta_b = -\frac{1}{4}\bigl(3L + \operatorname{Tr}(M)+\smash{\vec{b}^{\mathsf{T}}M^{-1}\vec{b}}\bigr) + n_0$. As a~result, the coefficient $a_n$ associated to \smash{$q^{\Delta_b+n}$} in \smash{$\widehat{Z}_b$} is given by~\smash{$\sum_{\substack{\vec{x}\in D_{n+n_0}}} \widehat{F}_{\Gamma, \vec{b}} (\vec{x})$}.
\end{Remark}

We continue with a~quick lemma inspired by experimentation with~\cite{petercode}.

\begin{Lemma}\label{B-size}
Let $n \in \mathbb{Z}$. If $n < \frac{1}{4}\vec{b}^{\mathsf{T}}M^{-1}\vec{b}$, then $D_n = \varnothing$. Otherwise,
\[
|D_n| \leq \left(2\sqrt{\frac{n-\frac{1}{4}\mathbf{b}^{\top}M^{-1}\mathbf{b}}{\lambda_{\min}}}+1\right)^L,
\]
where $\lambda_{\min}$ is the minimum eigenvalue of~$-M$.
\end{Lemma}

\begin{proof}
We want to find $\big|\big\{\vec{x}\in \mathbb{Z}^L \mid 2\chi_{\vec{b}}(\vec{x}) = n\big\}\big|$. Observe that $2\chi_{\vec{b}}(\vec{x}) = n$ holds if and only~if
 \begin{gather}\label{eq:ellipsoid_E}
 \left(\vec{x}+ \frac{1}{2}M^{-1}\vec{b}\right)^{\mathsf{T}}(-M)\left(\vec{x}+ \frac{1}{2}M^{-1}\vec{b}\right) = n - \frac{1}{4}\vec{b}^{\mathsf{T}}M^{-1}\vec{b}.
\end{gather}
Note that $\frac{1}{4}\vec{b}^{\mathsf{T}}M^{-1}\vec{b}$ does not depend on $\vec{x}$ and so it is a~constant. Since $M$ is negative definite so is $M^{-1}$. On the other hand, $-M$ is positive definite as a~result. This implies that $-\frac{1}{4}\vec{b}^{\mathsf{T}}M^{-1}\vec{b} > 0$ and $\lambda_{\max}>0$. Considering $\chi_{\vec{b}}$ as a~function from $\mathbb{R}^L$ to $\mathbb{R}$ a~linear algebra computation shows that $\min 2\chi_{\vec{b}} = \frac{1}{4}\vec{b}^{\mathsf{T}}M^{-1}\vec{b}$. This implies that~(i) the right-hand side of~\eqref{eq:ellipsoid_E} is always greater or equal to zero and (ii) that if $n < \frac{1}{4}\vec{b}^{\mathsf{T}}M^{-1}\vec{b}$, then $D_n = \varnothing$. Consider $E:=\bigl\{\mathbf{x}\in \mathbb{R}^{L}\mid 2\chi_{\mathbf{b}}(x)\leq n \bigr\}$.
From \eqref{eq:ellipsoid_E}, one sees that $E$ forms a solid elipsoid centered at $v:=-\frac{1}2M^{-1}\mathbf{b}$ with diameter
\[
2\sqrt{\frac{n-\frac14\mathbf{b}^\top M^{-1}\mathbf{b}}{\lambda_{\min}}},
\]
where $\lambda_{\min}$ is the minimum eigenvalue of $-M$, which contains at most
\[
\left(2\sqrt{\frac{n-\frac14\mathbf{b}^\top M^{-1}\mathbf{b}}{\lambda_{\min}}}+1\right)^{L}
\]
points of $\mathbb{Z}^L$.
Since $D_n$ is contained in this set the result follows.
\end{proof}

Before continuing, let us note that Definition~\ref{def:polynomially-bounded} can be extended to define the notion of a~sequence of rational numbers being polynomially bounded. This is relevant since by Remark~\ref{rem:discrepancies} one sees that the coefficients $a_n$ of $\widehat{Z}_b$ take values in the set $2^{-c}\mathbb{Z}$ for some integer constant $c \geq 0$. With this extension at hand we can state and prove the following.

\begin{Proposition}\label{prop-neg-def-coefficient-growth}
Let $M_3$ be a~negative definite plumbed manifold, $b \in\operatorname{Spin}^c(M_3)$ and $n \geq 0$ be an integer. If $a_n$ is the coefficient associated to $q^{\Delta_b+n}$ in $\widehat{Z}_b(M_3;q)$, then $|a_n|$ is polynomially bounded.
\end{Proposition}

\begin{proof}
Suppose $M_3$ is obtained by plumbing along a~weighted tree $\Gamma$ on $L$ vertices with (negative-definite) framing matrix $M$. Let $\vec{b} \in 2\mathbb{Z}^L + \vec{\delta}$ be a~$\operatorname{Spin}^c$ representative for $b$ and consider~\eqref{z-hat}. Let $n_0$ be the smallest integer such that $\sum_{\substack{\vec{x}\in D_{n_0}}} \widehat{F}_{\Gamma, \vec{b}} (\vec{x}) \neq 0$. By Remark~\ref{delta-shift}, we see that
\smash{$
 a_n = \sum_{\substack{\vec{x}\in D_{n+n_0}}} \widehat{F}_{\Gamma, \vec{b}} (\vec{x})$}.
Let $d$ denote the maximal degree of a~vertex of $\Gamma$. If $d \leq 2$, then $\widehat{Z}_b$ is a~polynomial in $q$. One then sees $|a_n| \leq \max_i |a_i|$ and the result follows trivially. Suppose now $d \geq 3$. Define
\begin{equation}\label{eq-def-m}
 m \coloneq \max_{\vec{x} \in D_{n+n_0},\, v \in V(\Gamma)}|(2M\vec{x}+\vec{b})_v|
\end{equation}
 and $G \colon \mathbb{R} \to \mathbb{R}$ by
\smash{$
 G(y) = 2\binom{d+y}{d-3} = \frac{(d+y)(d+y-1) \cdots (y+2)}{(d-3)!}$}.
Comparing with~\eqref{eq:def-f-n-r}, we observe for all~${\vec{x} \in D_{n+n_0}}$ and $v \in V(\Gamma)$ that
$
 \big|\widehat{F}_{\operatorname{deg}(v)}((2M\vec{x}+\vec{b})_v)\big| \leq G(m)$.
 This implies for all $\vec{x} \in D_{n+n_0}$ that
 \begin{equation}\label{eq:s-power}
 \big|\widehat{F}_{\Gamma, \vec{b}}(\vec{x})\big| \leq G(m)^L.
 \end{equation}
Let $\lambda_{\max}$, $\lambda_{\min}$ denote the largest (resp.\ smallest) eigenvalue of $-M$. Define the constant $C_{\vec{b}} :=n_0 - \frac{1}{4}\vec{b}^{\mathsf{T}}M^{-1}\vec{b}$ and define
\[
h(n):= 4\lambda_{\max} \sqrt{\frac{n+C_{\mathbf{b}}}{\lambda_{\min}}} + \|\mathbf b\|\left(1+\frac{\lambda_{\max}}{\lambda_{\min}}\right),
\]
where $||\cdot||$ is the usual $\mathbb{R}^{L}$ norm.
We claim that
\begin{equation}\label{eq:G(m)}
 G(m) < G(h(n)).
\end{equation}
Let $\vec{x} \in D_{n+n_0}$. Since $M$ and $M^{-1}$ are symmetric, we find $||M|| = \lambda_{\max}$ and $\bigl\|M^{-1}\bigr\| = 1/\lambda_{\min}$ (here $||\cdot||$ denotes the spectral norm). Using the fact $||M\vec{x}||\leq ||M||\cdot||\vec{x}||$, we find
\begin{equation}\label{eq:estimate_constituent}
 |(2M\vec{x}+\vec{b})_v| \leq 2\lambda_{\max}\cdot||\vec{x}|| + ||\vec{b}||.
\end{equation}
Further, from the proof of Lemma~\ref{B-size}, one sees that
\[
\mathbf{x} = \mathbf x_0-\frac12M^{-1}\mathbf b, \qquad \mathbf{x}_0 \in S\left(0,2\sqrt{\frac{n+C_{\mathbf b}}{\lambda_{\mathrm{min}}}}\right),
\]
where $S(0,R)$ is the solid ball centered at $0$ with radius $R$.
Therefore, by the triangle inequality%
\begin{equation}\label{eq:estimate-x}
    \|\mathbf x \| \leq 2\sqrt{\frac{n+C_{\mathbf b}}{\lambda_{\mathrm{min}}}}+\frac{1}{2}\bigl\|M^{-1}\mathbf b\bigr\|.
\end{equation}
Since $
\bigl\|M^{-1}\mathbf b\bigr\|
\leq
\bigl\|M^{-1}\bigr\| \|\mathbf b\|
=\|\mathbf b\|/\lambda_{\min}
$,
combining \eqref{eq:estimate_constituent} and \eqref{eq:estimate-x} gives
\[
m
\leq
4\lambda_{\max}
\sqrt{\frac{n+C_{\mathbf b}}{\lambda_{\min}}}
+
\|\mathbf b\|
\left(
1+\frac{\lambda_{\max}}{\lambda_{\min}}
\right)
=
h(n).
\]
Since $G$ is monotonically increasing, the claim follows. Equations~\eqref{eq:s-power} and the claim (now proven) imply
\[
\bigl|\widehat F_{\Gamma,\mathbf b}(\mathbf x)\bigr|
\leq G(h(n))^L.
\]
By Lemma~\ref{B-size},
\[
|D_{n+n_0}|
\leq
\left(
2\sqrt{\frac{n+C_{\mathbf b}}{\lambda_{\min}}}+1
\right)^L.
\]
Therefore,
\begin{equation}
|a_n|
\leq
\left(
2\sqrt{\frac{n+C_{\mathbf b}}{\lambda_{\min}}}+1
\right)^L
G(h(n))^L.
\tag{6.9}
\end{equation}
Since $h(n)=O\bigl(\sqrt n\bigr)$ and $G$ is a polynomial of degree $d-3$,
we have
\smash{$
G(h(n))=O\bigl(n^{\frac{d-3}{2}}\bigr)$}.
Consequently,
\smash{$
|a_n|
=
O\bigl(n^{\frac{(d-2)L}{2}}\bigr)$}.
In particular, choosing any integer
\smash{$
k\geq \bigl\lceil\frac{(d-2)L}{2}\bigr\rceil$},
there exists a constant $C>0$ such that
$
|a_n|\leq C(n+1)^k$
for all $n\geq0$. Hence $|a_n|$ is polynomially bounded.
\end{proof}

\begin{Corollary}
Suppose $M_3$ is a~negative definite plumbed manifold and $b \in \operatorname{Spin}^c(M_3)$, then
\smash{$
\ceff\bigl(\widehat{Z}_b(M_3;q)/(q)_{\infty}\bigr) \leq 1$}.
\end{Corollary}

\begin{proof}
 From Proposition~\ref{prop-neg-def-coefficient-growth} and the discussion preceding it, one sees that up to an overall factor of $2^{-c}$ for some integer $c \geq 0$, the absolute values of the coefficients of $\widehat{Z}_b(M_3;q)$ form a~polynomially bounded sequence of integers. After reviewing the conventions of Remark~\ref{rem:discrepancies}, the proof then follows by an application of Lemma~\ref{helper-lemma}.
\end{proof}

\section{Case studies}\label{sec:Case studies}

We conclude the analysis of $\ceff$ derived from the conjectural $+1/r$-surgery formula for $\widehat{Z}$ invariants on the positively plumbed Brieskorn spheres with case studies. We focus on a~simpler case of $+1$-surgeries, but also include one example of $+\frac{1}{2}$-surgery, namely $\overline{\Sigma(2,3,13)}$. Comparing the resulting $\ceff$ with the Chern--Simons invariants given by~\eqref{eq:CS invariants for Brieskorn spheres} could shed light on the conjectural relation~\eqref{coj:z-hat at leading order} between the two. In turn, this will allow us to produce a~set of counterexamples confirming our main Theorem~\ref{thm:conjecture fails}, as well as differences of the two approaches (modular and resurgent) to $\widehat{Z}$ on the other side.
Our data for $\ceff$ is presented in numerical tables which contain the following:
\begin{itemize}
\itemsep=0pt
 \item Numerical values on the lower bounds of $\ceff$ for $\brs{s}{t}{st\pm 1}$, following the analysis in Section~\ref{sec:Numerical analysis of ceff}.
 \item Exact values of $\ceff$ for $\brs{s}{t}{st+1}$ from the analysis of Section~\ref{sec:Mixed mock modularity and resurgence}.
 \item Exact values of $\ceff$ coming from resurgence, based on~\cite{ACDGO} as well as asymptotic results from Section~\ref{sec:asympt}.
 \item In Table~\ref{tab:exact values of ceff from modularity}, we also present $\ceff$ for the half-index of the corresponding 3d $\mathcal{N}=2$ theory~$T[M_3]$.
\end{itemize}

\subsection[sigma]{$\boldsymbol{\overline{\Sigma(s,t,st+1)}}$}

The numerical estimates for the lower bound of $m$ are shown in Table~\ref{tab:numerics for m, case 2}.

\begin{table}[!ht]
\centering
\resizebox{16cm}{!}{
\begin{tabular}{||c|ccccccccccccccccc||}
\hline
\backslashbox{$s$}{$t$} & 3 & 4 &. &. &. &. &. &. &. &. &. &. &. &. &. &. &. \\
\hline
 2 & 0.937 &. &. &. &. &. &. &. &. &. &. &. &. &. &. &. &. \\
 3 &. & 4.66 &. &. &. &. &. &. &. &. &. &. &. &. &. &. &. \\
. & 2.90 & 6.02 & 7.88 &. &. &. &. &. &. &. &. &. &. &. &. &. &. \\
. &. &. &. & 11.8 &. &. &. &. &. &. &. &. &. &. &. &. &. \\
. & 4.84 & 8.37 & 11.0 & 13.8 & 16.4 &. &. &. &. &. &. &. &. &. &. &. &. \\
. &. & 9.48 &. & 15.7 &. & 21.8 &. &. &. &. &. &. &. &. &. &. &. \\
. & 6.63 &. & 14.1 & 17.6 &. & 24.4 & 27.8 &. &. &. &. &. &. &. &. &. &. \\
. &. & 11.8 &. &. &. & 27.0 &. & 34.3 &. &. &. &. &. &. &. &. &. \\
. & 8.32 & 13.0 & 17.2 & 21.5 & 25.5 & 29.6 & 33.6 & 37.5 & 41.2 &. &. &. &. &. &. &. &. \\
. &. &. &. & 23.3 &. & 32.2 &. &. &. & 48.4 &. &. &. &. &. &. &. \\
. & 9.96 & 15.3 & 20.2 & 25.2 & 30.0 & 34.7 & 39.2 & 43.6 & 47.8 & 51.8 & 55.6 &. &. &. &. &. &. \\
. &. & 16.5 &. & 27.1 &. &. &. & 46.5 &. & 55.1 &. & 62.7 &. &. &. &. &. \\
. & 11.6 &. & 23.3 &. &. & 39.6 & 44.6 &. &. & 58.2 &. & 66.0 & 69.5 &. &. &. &. \\
. &. & 18.8 &. & 30.8 &. & 42.0 &. & 52.1 &. & 61.2 &. & 69.1 &. & 75.8 &. &. &. \\
. & 13.2 & 20.0 & 26.3 & 32.6 & 38.6 & 44.3 & 49.7 & 54.8 & 59.6 & 64.0 & 68.2 & 72.0 & 75.5 & 78.7 & 81.6 &. &. \\
. &. &. &. & 34.4 &. & 46.6 &. &. &. & 66.8 &. & 74.7 &. &. &. & 86.8 &. \\
. & 14.7 & 22.3 & 29.3 & 36.2 & 42.7 & 48.9 & 54.6 & 60.0 & 64.9 & 69.4 & 73.5 & 77.2 & 80.6 & 83.7 & 86.5 & 89.0 & 91.3 \\
. &. & 23.4 &. &. &. & 51.1 &. & 62.4 &. & 71.9 &. & 79.7 &. &. &. & 91.1 &. \\
 \hline
\end{tabular}
}
\caption{The lower bounds on $m$ derived from the numerical analysis; \smash{$M_3=\overline{\Sigma(s,t,st+1)}$}. The first~$10^4$ terms in \smash{$\widehat{Z}$} are being used. Using this, $\ceff$ for $\widehat{Z}$ invariant can be estimated by the conjectural formula~\smash{$\ceff = \frac{m^2}{4st(st+1)}$}.}
\label{tab:numerics for m, case 2}
\end{table}

\begin{table}[!ht]
\centering
\resizebox{16cm}{!}{
\begin{tabular}{||c|cccccccccccccccccc||}
\hline
\backslashbox{$s$}{$t$} & 3 & 4 &. &. &. &. &. &. &. &. &. &. &. &. &. &. &. &. \\
\hline
2 & 1.13 &. &. &. &. &. &. &. &. &. &. &. &. &. &. &. &. &. \\
3 &. & 1.84 &. &. &. &. &. &. &. &. &. &. &. &. &. &. &. &. \\
. & 1.32 & 1.80 & 1.89 &. &. &. &. &. &. &. &. &. &. &. &. &. &. &. \\
. &. &. &. & 1.90 &. &. &. &. &. &. &. &. &. &. &. &. &. &. \\
. & 1.41 & 1.70 & 1.78 & 1.86 & 1.89 &. &. &. &. &. &. &. &. &. &. &. &. &. \\
. &. & 1.66 &. & 1.82 &. & 1.89 &. &. &. &. &. &. &. &. &. &. &. &. \\
. & 1.44 &. & 1.73 & 1.79 &. & 1.86 & 1.88 &. &. &. &. &. &. &. &. &. &. &. \\
. &. & 1.63 &. &. &. & 1.83 &. & 1.86 &. &. &. &. &. &. &. &. &. &. \\
. & 1.45 & 1.62 & 1.70 & 1.76 & 1.79 & 1.81 & 1.83 & 1.84 & 1.83 &. &. &. &. &. &. &. &. &. \\
. &. &. &. & 1.74 &. & 1.79 &. &. &. & 1.80 &. &. &. &. &. &. &. &. \\
. & 1.45 & 1.60 & 1.67 & 1.72 & 1.76 & 1.78 & 1.78 & 1.79 & 1.78 & 1.77 & 1.76 &. &. &. &. &. &. &. \\
. &. & 1.59 &. & 1.72 &. &. &. & 1.76 &. & 1.74 &. & 1.71 &. &. &. &. &. &. \\
. & 1.45 &. & 1.66 &. &. & 1.74 & 1.75 &. &. & 1.72 &. & 1.68 & 1.65 &. &. &. &. &. \\
. &. & 1.58 &. & 1.69 &. & 1.73 &. & 1.72 &. & 1.69 &. & 1.65 &. & 1.60 &. &. &. &. \\
. & 1.44 & 1.58 & 1.64 & 1.68 & 1.71 & 1.71 & 1.71 & 1.70 & 1.68 & 1.66 & 1.64 & 1.62 & 1.59 & 1.56 & 1.54 &. &. &. \\
. &. &. &. & 1.68 &. & 1.70 &. &. &. & 1.64 &. & 1.59 &. &. &. & 1.48 &. &. \\
. & 1.44 & 1.57 & 1.63 & 1.67 & 1.68 & 1.69 & 1.68 & 1.66 & 1.64 & 1.62 & 1.59 & 1.56 & 1.53 & 1.51 & 1.48 & 1.45 & 1.43 &. \\
. &. & 1.56 &. &. &. & 1.67 &. & 1.64 &. & 1.59 &. & 1.54 &. &. &. & 1.42 &. & 1.38 \\
 \hline
\end{tabular}
}
\caption{The numerical estimates of the lower bound on $\ceff$ for the half-index of \smash{$T\big[\overline{\Sigma(s,t,st+1)}\big]$}. Note that these values are obtained from the above Table~\ref{tab:numerics for m, case 2} via \smash{$\ceff^{T[M_3]} := 1+24 \frac{m^2}{4st(st+1)} $}, see~\cite{Gukov:2023cog} for the details.}
\label{tab:numerics for ceff, case 2b}
\end{table}

This class of examples admits a~mixed mock modular description (see Section~\ref{sec:Mixed mock modularity and resurgence}). The analysis from that section allows us comparison the numerical data with $\ceff$ obtained by using modular completion of $\widehat{Z}$, as shown in Table~\ref{tab:exact values of m from modularity}. One can see that it gives an excellent agreement with Table~\ref{tab:numerics for m, case 2}. By carefully examining the values of $m$ in Table~\ref{tab:exact values of m from modularity}, we see that the majority of cases (where $m$ is not an integer) fail to relate to any Chern--Simons invariant, even with $l\neq 0$. This confirms Theorem~\ref{thm:conjecture fails}.

\begin{table}[!ht]
\centering
\resizebox{12.5cm}{!}{
\begin{tabular}{||c|cccccccccc||}
\hline
 \backslashbox{$s$}{$t$} & 3 & 4 &. &. &. &. &. &. &. &. \\
 \hline
 2 & 1 &. &. &. &. &. &. &. &. &. \\
 3 &. & $3 \sqrt{\frac{13}{5}}$ &. &. &. &. &. &. &. &. \\
. & 3 & $16 \sqrt{\frac{2}{13}}$ & $\sqrt{\frac{203}{3}}$ &. &. &. &. &. &. &. \\
. &. &. &. & $\sqrt{\frac{2139}{14}} $&. &. &. &. &. &. \\
. & 5 & $22 \sqrt{\frac{3}{19}} $& $\sqrt{\frac{5191}{39}} $& $\frac{48}{\sqrt{11}} $& $\frac{\sqrt{\frac{5977}{5}}}{2} $&. &. &. &. &. \\
. &. & $5 \sqrt{\frac{43}{11}} $&. & $\sqrt{\frac{15457}{57}}$ &. & $\frac{\sqrt{4769}}{3}$ &. &. &. &. \\
. & $\sqrt{\frac{95}{2}}$ &. & $\sqrt{\frac{3737}{17}}$ & $46 \sqrt{\frac{7}{43}}$ &. & $64 \sqrt{\frac{10}{61}}$ & $\sqrt{\frac{30587}{35}}$ &. &. &. \\
. &. & $\sqrt{\frac{2139}{14}}$ &. &. &. & $\sqrt{\frac{84277}{102}}$ &. & $\frac{\sqrt{\frac{59969}{11}}}{2}$ &. &. \\
. & $7 \sqrt{\frac{23}{15}}$ & $34 \sqrt{\frac{5}{31}}$ & $3 \sqrt{\frac{255}{7}}$ & $56 \sqrt{\frac{26}{159}}$ & $\frac{3 \sqrt{\frac{2613}{2}}}{4}$ & $\frac{26 \sqrt{37}}{5}$ &
 $\sqrt{\frac{168121}{129}}$ & $\frac{400}{\sqrt{97}}$ & $\frac{\sqrt{\frac{36593}{2}}}{3}$ &. \\
. &. &. &. & $17 \sqrt{\frac{61}{29}}$ &. & $\sqrt{\frac{48705}{41}}$ &. &. &. & $\sqrt{\frac{190057}{65}}$ \\
 \hline
\end{tabular}
}
\caption{The exact values of $m$ for $\overline{\Sigma(s,t,st+1)}$, obtained from the analysis of the modular completion of $\widehat{Z}$. Note that only three cases, namely $\overline{\Sigma(2,3,7)}$, $\overline{\Sigma(2,5,11)}$ and $\overline{\Sigma(2,7,15)}$ correspond to integer values of $m$. In all other cases the conjectural relation~\eqref{coj:z-hat at leading order} fails, as the value $\ceff$ is not related to any Chern--Simons invariant.}
\label{tab:exact values of m from modularity}
\end{table}

\begin{table}[!ht]
\centering
\resizebox{12.5cm}{!}{
\begin{tabular}{||c|cccccccccc||}
\hline
\backslashbox{$s$}{$t$} & 3 & 4 &. &. &. &. &. &. &. &. \\
 \hline
2 &. &. &. &. &. &. &. &. &. &. \\
 3 & 1.00 &. &. &. &. &. &. &. &. &. \\
. &. & 4.84 &. &. &. &. &. &. &. &. \\
. & 3.00 & 6.28 & 8.23 &. &. &. &. &. &. &. \\
. &. &. &. & 12.4 &. &. &. &. &. &. \\
. & 5.00 & 8.74 & 11.5 & 14.5 & 17.3 &. &. &. &. &. \\
. &. & 9.89 &. & 16.5 &. & 23.0 &. &. &. &. \\
. & 6.89 &. & 14.8 & 18.6 &. & 25.9 & 29.6 &. &. &. \\
. &. & 12.4 &. &. &. & 28.7 &. & 36.9 &. &. \\
. & 8.67 & 13.7 & 18.1 & 22.6 & 27.1 & 31.6 & 36.1 & 40.6 & 45.1 &. \\
. &. &. &. & 24.7 &. & 34.5 &. &. &. & 54.1 \\
 \hline
\end{tabular}
}
 \caption{The three-digit approximation of the numbers from Table~\ref{tab:exact values of m from modularity}.}
 \label{tab:exact values of m from modularity approximated}
\end{table}

\begin{table}[t]
\centering
\resizebox{12.5cm}{!}{
\renewcommand{\arraystretch}{1.2}
\begin{tabular}{||c|cccccccccc||}
\hline
\backslashbox{$s$}{$t$} & 3 & 4 &. &. &. &. &. &. &. &. \\
 \hline
2 & $\frac{8}{7}$ &. &. &. &. &. &. &. &. &. \\
3 &. & $\frac{19}{10}$ &. &. &. &. &. &. &. &. \\
. & $\frac{82}{55}$ & $\frac{129}{65}$ & $\frac{59}{30}$ &. &. &. &. &. &. &. \\
. &. &. &. & $\frac{139}{70}$ &. &. &. &. &. &. \\
. & $\frac{12}{7}$ & $\frac{265}{133}$ & $\frac{361}{182}$ & $\frac{769}{385}$ & $\frac{279}{140}$ &. &. &. &. &. \\
. &. & $\frac{87}{44}$ &. & $\frac{757}{380}$ &. & $\frac{503}{252}$ &. &. &. &. \\
. & $\frac{11}{6}$ &. & $\frac{203}{102}$ & $\frac{1289}{645}$ &. & $\frac{2561}{1281}$ & $\frac{839}{420}$ &. &. &. \\
. &. & $\frac{139}{70}$ &. &. &. & $\frac{2377}{1190}$ &. & $\frac{1319}{660}$ &. &. \\
. & $\frac{104}{55}$ & $\frac{681}{341}$ & $\frac{307}{154}$ & $\frac{5827}{2915}$ & $\frac{703}{352}$ & $\frac{3849}{1925}$ & $\frac{3781}{1892}$ & $\frac{6401}{3201}$ & $\frac{1979}{990}$ &. \\
. &. &. &. & $\frac{579}{290}$ &. & $\frac{1147}{574}$ &. &. &. & $\frac{2859}{1430}$ \\
 \hline
\end{tabular}
}
\caption{The exact values of $\ceff$ for $T\big[\overline{\Sigma(s,t,st+1)}\big]$, obtained from the analysis of the modular completion of $\widehat{Z}$. Note that only three cases, namely $\overline{\Sigma(2,3,7)}$, $\overline{\Sigma(2,5,11)}$ and $\overline{\Sigma(2,7,15)}$ correspond to integer values of $m$. In all other cases, the conjectural relation~\eqref{coj:z-hat at leading order} fails, as the value $\ceff$ is not related to any Chern--Simons invariant.}
\label{tab:exact values of ceff from modularity}
\end{table}

We can also see how fast the coefficients grow -- some examples are shown in Figure~\ref{fig:stst+1_combined}.

\begin{figure}[!ht]
 \centering
 \includegraphics[width=0.85\linewidth]{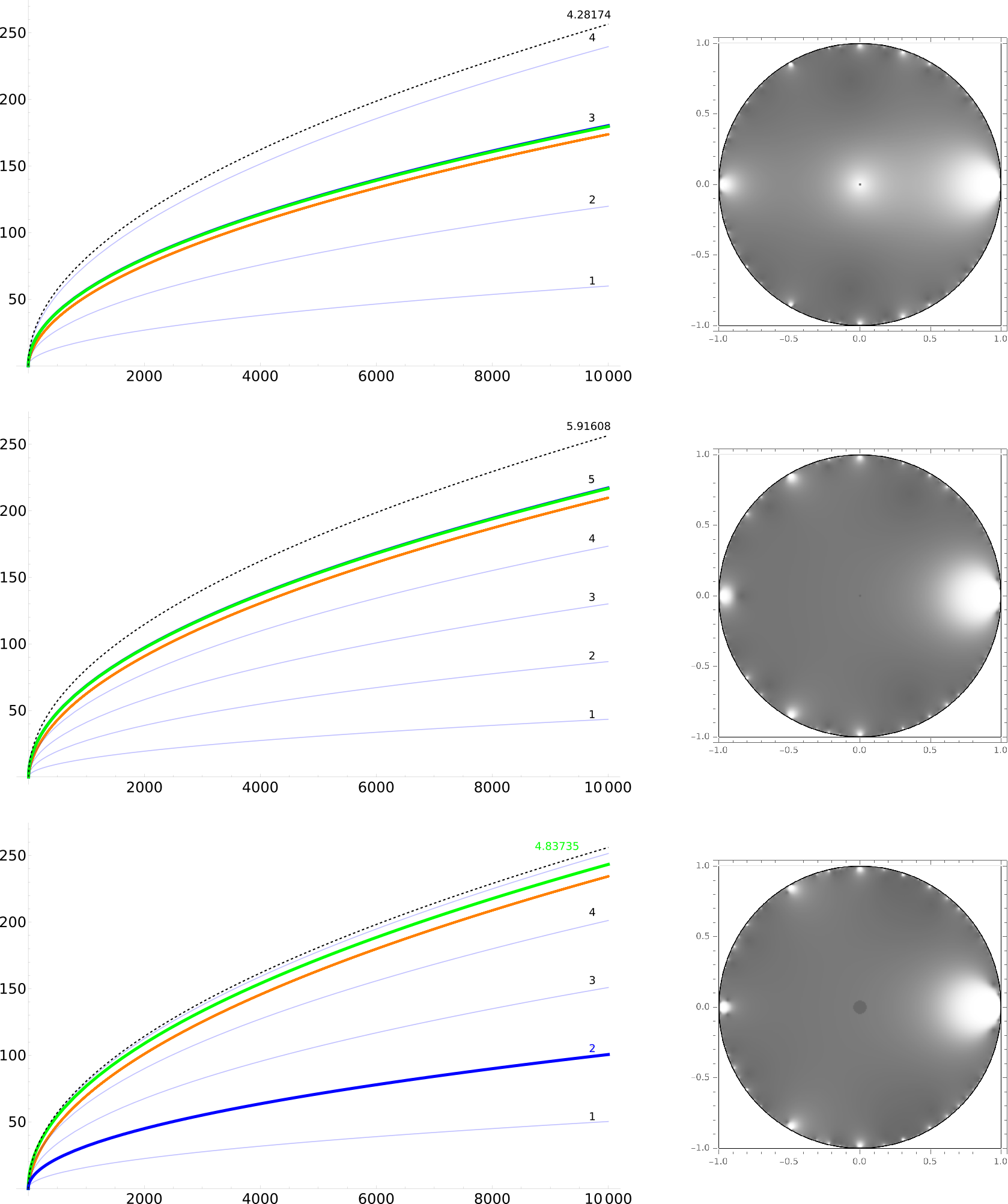}
 \caption{$c_{\rm eff}$ data for $\overline{\Sigma(2,5,11)}$, $\overline{\Sigma(2,7,15)}$ and $\overline{\Sigma(3,4,13)}$ (from top to bottom). The numerical estimates for $m$ are shown by \textcolor{orange}{orange} lines.
 The exact value of $m$ from mixed mock modularity is shown in \textcolor{green}{green}, and the value of $m$ obtained from resurgence is shown in \textcolor{blue}{blue}.
 In the first two cases with $m=3$ and $m=5$, the green and blue lines coincide. For $\brs{3}{4}{13}$ the value from resurgence is significantly smaller than the actual $\ceff$. We can also see that $\brs{3}{4}{13}$ provides a~counterexample to the conjectural relation~\eqref{coj:z-hat at leading order} since $m$ is an irrational number (\smash{$m=3\sqrt{\frac{13}{5}}\sim 4.8375$}, see Table~\ref{tab:exact values of m from modularity}) and is not related to any Chern--Simons invariant of a~flat connection.} \label{fig:stst+1_combined}\vspace{-2mm}
\end{figure}

\begin{figure}[!ht]
 \centering
 \includegraphics[width=0.86\linewidth]{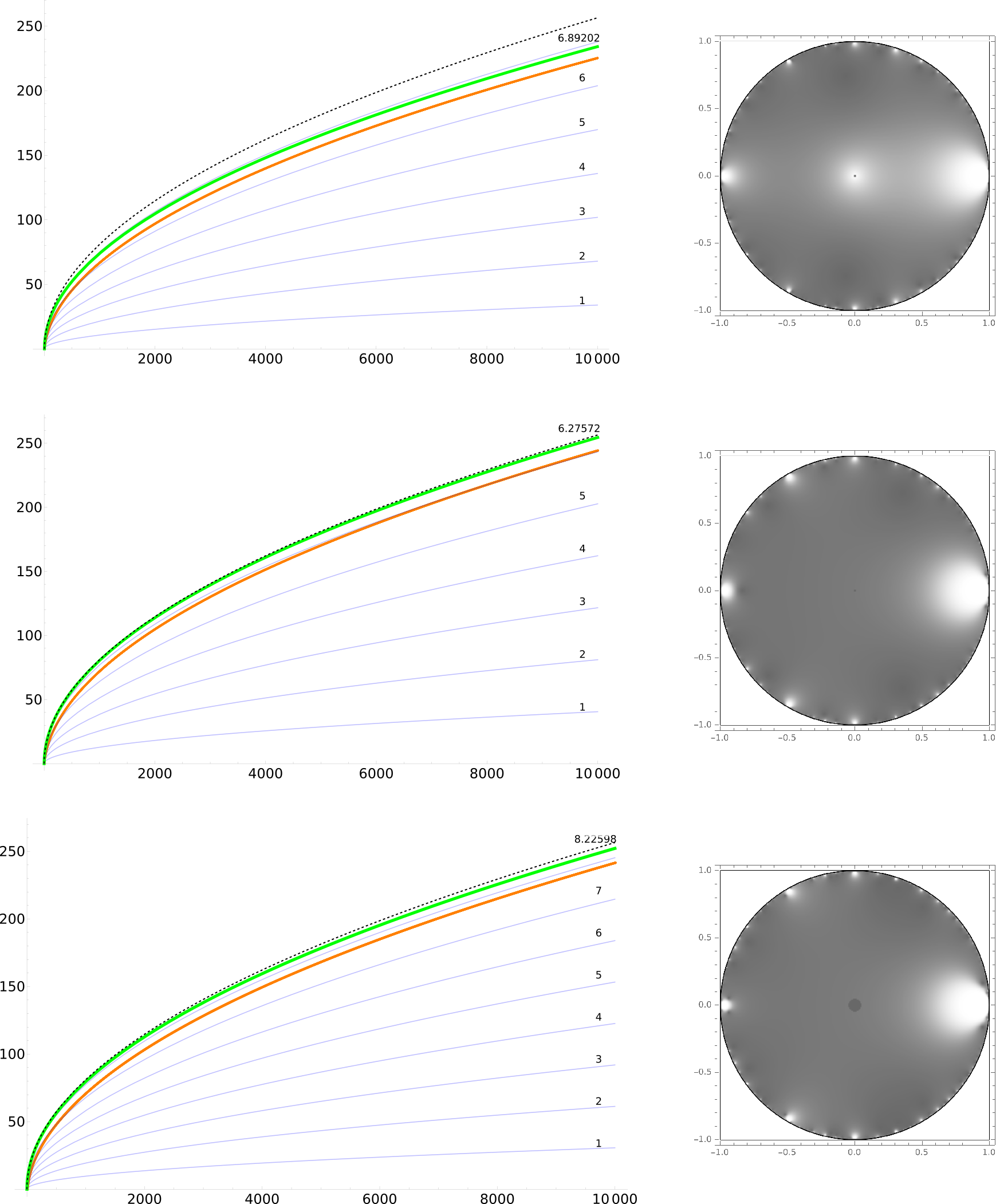}
 \caption{$c_{\rm eff}$ data for $\overline{\Sigma(2,9,19)}$, $\overline{\Sigma(3,5,16)}$ and $\overline{\Sigma(4,5,21)}$ (from top to bottom). The numerical estimates for $m$ are shown by \textcolor{orange}{orange} lines.
 The exact value of $m$ from mixed mock modularity is shown in \textcolor{green}{green}.
 All cases show an excellent agreement of the mixed modular analysis and the numerics. On the other hand, they provide another set of counterexamples to the conjectural relation~\eqref{coj:z-hat at leading order} since $m$ is an irrational number and is not related to any Chern--Simons invariant of a~flat connection, see Table~\ref{tab:exact values of m from modularity}.} \label{fig:stst+1_combined2}
\end{figure}

\subsection[sigma]{$\boldsymbol{\overline{\Sigma(s,t,st-1)}}$}

For this class of examples, a~mock modular description is not known, and we are only left to rely on the numerical estimates. Note that the resurgence analysis has been performed independently in~\cite{ACDGO}, and we can also compare our estimates with the values found there (although based on~$\overline{\Sigma(s,t,st+1)}$, we know that the value coming from resurgence analysis may not produce the actual $\ceff$). Table~\ref{tab:numerics for m, case 1} shows the numerical estimates of $m$.
We further compare the numerical results for the regularized +1 surgery formula with the resurgence -- the detailed comparative results for some small values of $s$, $t$ are shown in Figure~\ref{fig:stst-1_combined}.

\begin{table}[!ht]
\centering
\resizebox{14cm}{!}{
\begin{tabular}{||c|ccccccccccccccccc||}
\hline
\backslashbox{$s$}{$t$} & 3 & 4 &. &. &. &. &. &. &. &. &. &. &. &. &. &. &. \\
 \hline
 2 & 0.951 &. &. &. &. &. &. &. &. &. &. &. &. &. &. &. &. \\
 3 &. & 4.34 &. &. &. &. &. &. &. &. &. &. &. &. &. &. &. \\
. & 2.90 & 5.63 & 7.50 &. &. &. &. &. &. &. &. &. &. &. &. &. &. \\
. &. &. &. & 11.4 &. &. &. &. &. &. &. &. &. &. &. &. &. \\
. & 4.70 & 7.97 & 10.6 & 13.4 & 16.0 &. &. &. &. &. &. &. &. &. &. &. &. \\
. &. & 9.09 &. & 15.3 &. & 21.4 &. &. &. &. &. &. &. &. &. &. &. \\
. & 6.38 &. & 13.7 & 17.2 &. & 24.0 & 27.3 &. &. &. &. &. &. &. &. &. &. \\
. &. & 11.4 &. &. &. & 26.6 &. & 33.8 &. &. &. &. &. &. &. &. &. \\
. & 8.02 & 12.6 & 16.8 & 21.0 & 25.1 & 29.2 & 33.1 & 36.9 & 40.7 &. &. &. &. &. &. &. &. \\
. &. &. &. & 22.9 &. & 31.7 &. &. &. & 47.8 &. &. &. &. &. &. &. \\
. & 9.62 & 14.9 & 19.8 & 24.8 & 29.5 & 34.2 & 38.7 & 43.0 & 47.2 & 51.2 & 55.0 &. &. &. &. &. &. \\
. &. & 16.1 &. & 26.6 &. &. &. & 45.9 &. & 54.4 &. & 62.0 &. &. &. &. &. \\
. & 11.2 &. & 22.9 &. &. & 39.1 & 44.0 &. &. & 57.5 &. & 65.2 & 68.7 &. &. &. &. \\
. &. & 18.4 &. & 30.3 &. & 41.4 &. & 51.5 &. & 60.5 &. & 68.3 &. & 75.0 &. &. &. \\
. & 12.8 & 19.5 & 25.9 & 32.1 & 38.1 & 43.8 & 49.1 & 54.2 & 58.9 & 63.3 & 67.4 & 71.2 & 74.6 & 77.8 & 80.7 &. &. \\
. &. &. &. & 33.9 &. & 46.0 &. &. &. & 66.0 &. & 73.9 &. &. &. & 85.9 &. \\
. & 14.4 & 21.8 & 28.8 & 35.7 & 42.2 & 48.3 & 54.0 & 59.3 & 64.1 & 68.6 & 72.7 & 76.4 & 79.8 & 82.8 & 85.6 & 88.1 & 90.4 \\
. &. & 23.0 &. &. &. & 50.5 &. & 61.7 &. & 71.1 &. & 78.8 &. &. &. & 90.2 &. \\
 \hline
\end{tabular}
}

\caption{The lower bounds on $m$ derived from the numerical analysis; $M_3=\overline{\Sigma(s,t,st-1)}$.}
\label{tab:numerics for m, case 1}
\end{table}

\begin{figure}[!ht]
 \centering
 \includegraphics[width=0.85\linewidth]{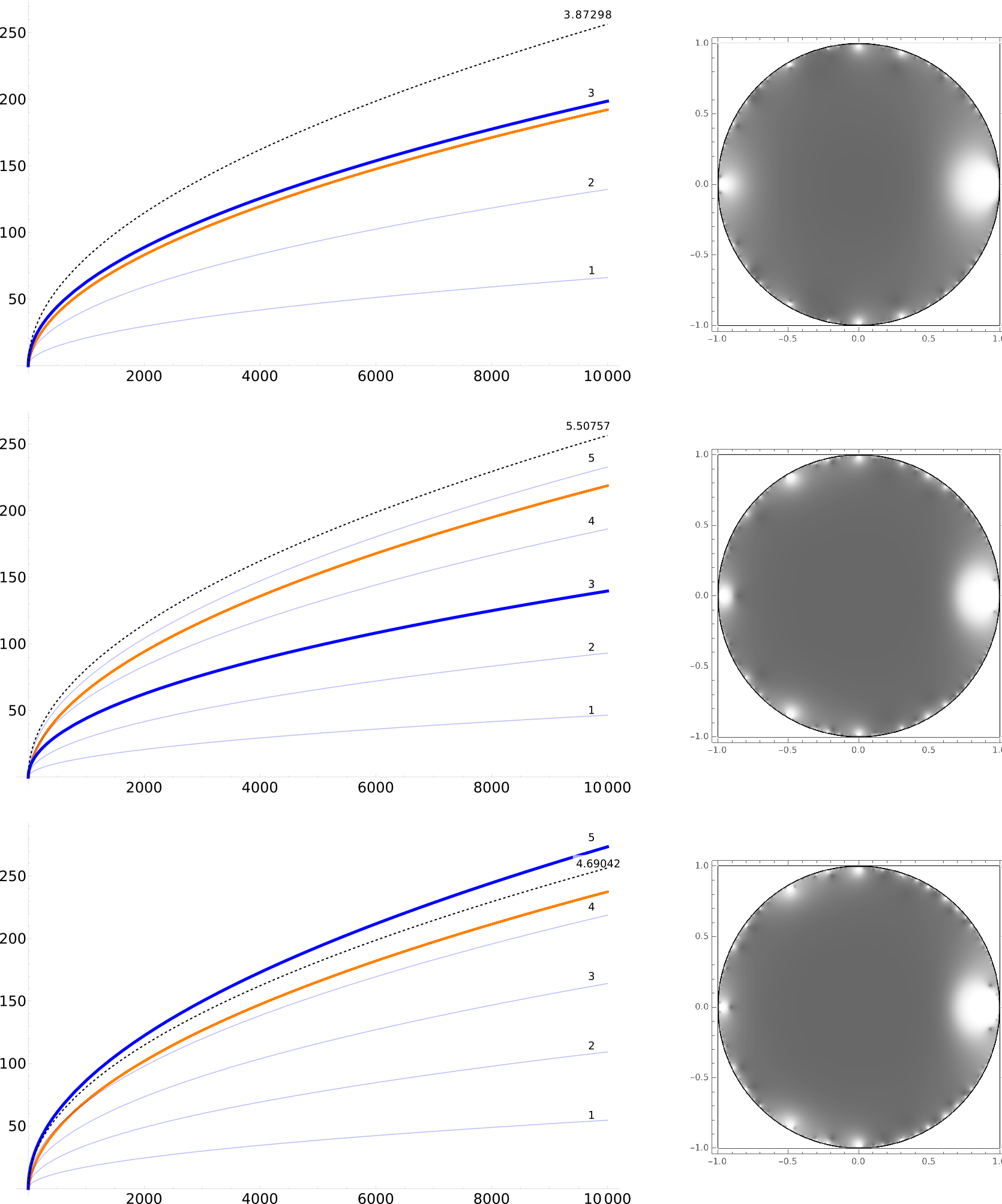}
 \caption{$c_{\rm eff}$ data for $\overline{\Sigma(2,5,9)}$, $\overline{\Sigma(2,7,13)}$ and $\overline{\Sigma(3,4,11)}$ (from top to bottom). The numerical estimates for $m$ are shown by \textcolor{orange}{orange} lines, and the value of $m$ obtained from resurgence is shown in \textcolor{blue}{blue}. The first case gives a~nice agreement with the resurgence, while the other two show important discrepancies. For $\overline{\Sigma(2,7,13)}$, the resurgence value lies below the numerical estimate, while for $\overline{\Sigma(3,4,11)}$ it lies above the upper bound and thus cannot be consistent with the regularized $+1$-surgery formula~\eqref{eq:Zhat regularized +1 surgery}. $\overline{\Sigma(3,4,11)}$ is another counterexample to the conjectural relation~\eqref{coj:z-hat at leading order}, since
 there are no integer values of $m$ between the numerical estimate and the upper bound.}
 \label{fig:stst-1_combined}\vspace{-2mm}
\end{figure}

\subsection[sigma]{$\boldsymbol{\overline{\Sigma(2,3,13)}}$}

We also include one example with dominant singularity not located at $q=1$. Numerical analysis shows that in general, for higher values of $r$ the position of the dominant singularity will change as a~function of $r$. The simplest case is $\overline{\Sigma(2,3,13)}$. In this case we find the dominant singularities at \smash{$q=\mathrm{e}^{-2\pi {\rm i}\frac{2}{5}}$} and \smash{$q=\mathrm{e}^{-2\pi {\rm i} \frac{3}{5}}$}, similarly to the case shown in Figure~\ref{fig:singular examples}, left.
We can bring the dominant singularity to $q=1$ by dividing the $\widehat{Z}\bigl(\overline{\Sigma(2,3,13)},\tau\bigr)$ by a~Dedekind eta function $\eta(q)$. Using modularity we see that the asymptotic behavior of \smash{$\widehat{Z}\bigl(\overline{\Sigma(2,3,13)};\mathrm{e}^{2\pi {\rm i} \tau}\bigr)/ \eta(q)$} is given by
\[
 \frac{\widehat{Z}\bigl(\overline{\Sigma(2,3,13)};\mathrm{e}^{2\pi {\rm i} \tau}\bigr)}{\eta(q)} \sim \mathrm{e}^{\frac{7 {\rm i} \pi }{78 \tau}}\qquad \text{as} \quad q=\mathrm{e}^{2\pi {\rm i} \tau} \rightarrow 1.
\]
It implies that the asymptotic behavior of the coefficients of the $q$-series \smash{$\widehat{Z}\bigl(\overline{\Sigma(2,3,13)};\mathrm{e}^{2\pi {\rm i} \tau}\bigr)/  \eta(q)$}
 is given by
\smash{$a_{n} \sim \exp\bigl(\raisebox{-2pt}{$\sqrt{\vphantom{\frac{2 a}{3}} \smash{\frac{2 \pi^{2}}{3} \frac{14}{13} n}}$}\bigr)$}.
 Therefore, the $\ceff$ is given by
$\ceff = \frac{14}{13}$.
\begin{figure}[!ht]
 \centering
 \includegraphics[width=0.45\linewidth]{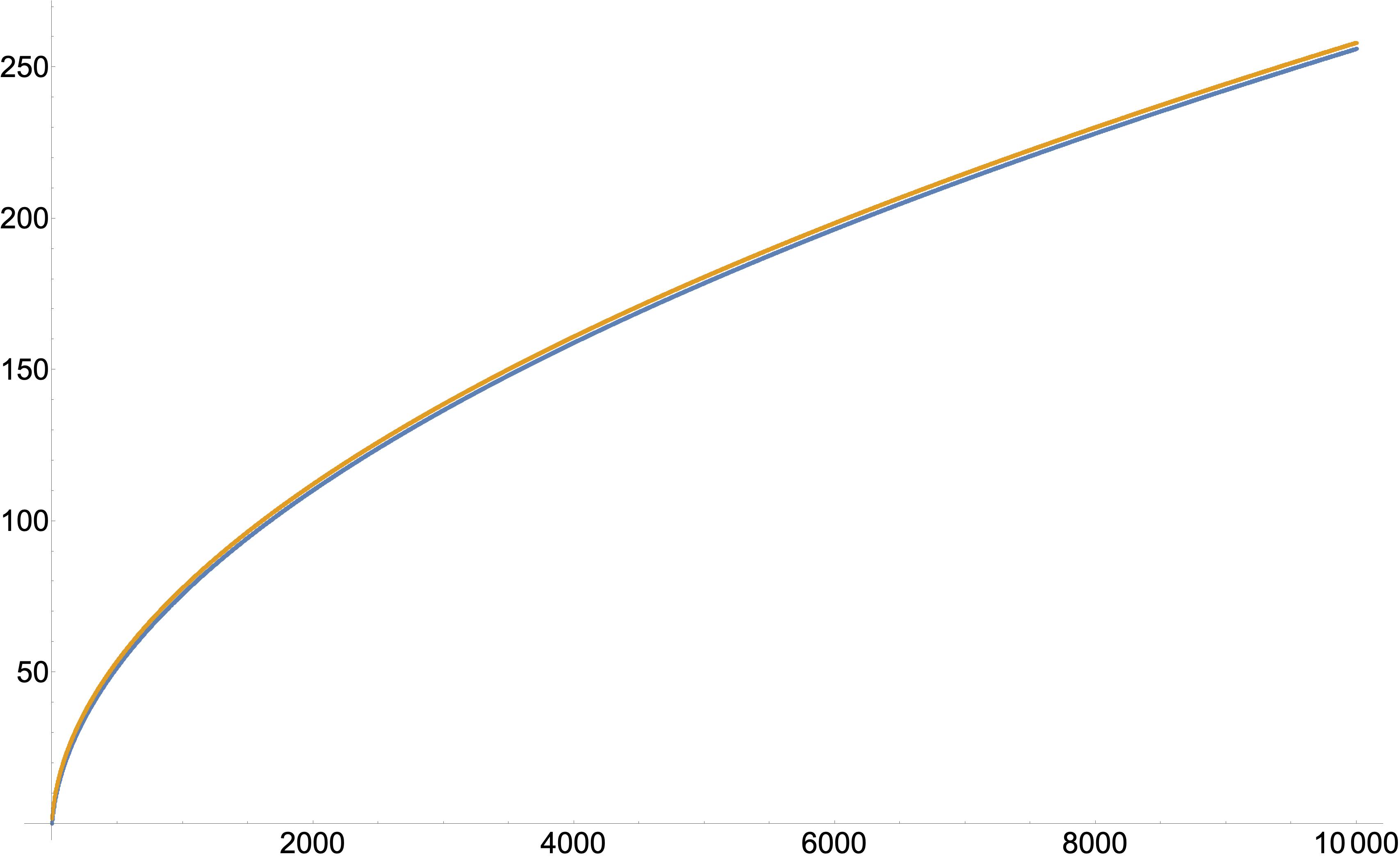}
 \caption{Comparison of coefficients of $\widehat{Z}\bigl(\overline{\Sigma(2,3,13)};q\bigr)/ \eta(q)$ (in blue) and \smash{$ \mathrm{e}^{\sqrt{ \frac{2 \pi^{2}}{3} \frac{14}{13} n }} $} (in orange).}
 \label{ceff2313byeta}\vspace{-2mm}
\end{figure}

\section{Discussion of the results}

In this paper, we focused on the conjectural form of $\zhat$ invariants for Brieskorn spheres $\overline{\Sigma (s,t, }\allowbreak\overline{rst\pm1)}$ from the regularized $+1/r$-surgery formula (the same approach to $\zhat$ was considered, for example, in~\cite{Cheng:2024vou}). Having studied the asymptotic growth of its coefficients encoded by the effective central charge $\ceff\in\mathbb{R}$, we were able
to highlight certain distinctions with another, independent proposal~\cite{ACDGO} rooted in resurgence. The more subtle analysis of the mock modular structure of the regularized surgery showed the mutual incompatibility of the two proposals; as a~result, they also produce very different values of $\ceff$ for the same 3-manifold. This is a~very interesting (and somewhat surprising) result -- we anticipate that these findings are an important step towards a~comprehensive definition of $\widehat{Z}$ for a~general positive definite plumbed 3-manifold. Among the examples where the two proposals agree on $\ceff$ are surgeries on the trefoil knot, as well as $\overline{\Sigma(2,5,9)}$, $\overline{\Sigma(2,5,11)}$ and $\overline{\Sigma(2,7,15)}$. On the contrary, for $\overline{\Sigma(3,4,11)}$ the resurgent method produces $\ceff$ which exceeds the upper bound using our approach, while for $\brs{3}{4}{13}$ the corresponding values of $\ceff$ differ significantly. Many other case studies produce similar inconsistencies.

Apart from all that, we revealed numerous counterexamples to the expected relation between~$\ceff$ and Chern--Simons invariants of flat connections -- assuming, of course, the surgery prescription for $\widehat{Z}$ which we use throughout the paper. Among all such cases, the first one happens for $\overline{\Sigma(3,4,11)}$, shown on the bottom plot of Figure~\ref{fig:stst-1_combined}. In this case there are no integer values of $m$ lying between the lower and upper bounds, indicating that $\ceff$ is not related to any Chern--Simons invariant.
The same phenomenon happens for $\overline{\Sigma(2,2k+1,4k+1)}$, starting from~${k=5}$. As for the $+1$-surgery on a~negative torus knot, the first such example is $\overline{\Sigma(3,5,16)}$.
Another notable example is $\overline{\Sigma(4,5,21)}$, where
the only integer value of $m$ within the bounded region is $m=8$.
However, it is divisible by $4$ and therefore not a~Chern--Simons invariant. It is plausible that in these cases $\ceff$ may correspond to ``flat connections at infinity'' \cite{Gukov:2024vbr}; it would be interesting to investigate this in future work.

In order to confirm these numerical suspicions with the exact values of $\ceff$, we used the mock modular properties of the conjectural form of $\widehat{Z}$ for $\brs{s}{t}{rst+1}$. By using the technique of modular completions, we were able to derive asymptotics of all terms in the $S$-transformation of~$\widehat{Z}$, approximated by an appropriate regularized theta series -- in agreement with the numerical estimates for $\ceff$ for this class of manifolds.
Interestingly, we also found that $\widehat{Z}$ decomposes into a~number of summands, not all of which correspond to non-abelian flat connections.

Summing up, our findings put a~number of current conjectures about $\widehat{Z}$ invariants of positive definite plumbed manifolds to a~strong test. We keep their interpretation open and leave it for future much-anticipated research.

\appendix
\section[Modular S-transformation of Theta]{Modular $\boldsymbol{S}$-transformation of $\boldsymbol{\widehat{\Theta}_{A, \mu + b, b, c_{1}, c_{2} }}$}\label{stransformofnonholomorphicsec}
Suppose $A$ is a~$2 \times 2$ indefinite matrix, $\mu \in L^{*}$, $2 b \in L^{*}$, $B(x,y)= x^{\mathsf{T}} A~y$, $|n|^{2}= B(n,n)$ and~${c_{1}, c_{2} \in \ZZ^{2} }$ such that $|c_{i}|<0$. The modular $S$-transform of $\widehat{\Theta}_{A, \mu + b, b, c_{1}, c_{2} }(\tau, \bar{\tau})$ is given in~\cite{Zwegers2008MockTF}
\[
 \widehat{\Theta}_{A, \mu + b, b, c_{1}, c_{2} }(\tau, \bar{\tau}) = \frac{- \tau^{-1} \exp(2 \pi {\rm i} B(\mu + b, b))}{\sqrt{|\det(A)|}} \sum_{ \lambda \in L^{*} } \widehat{\Theta}_{A, \lambda + b, -\mu - b, c_{1}, c_{2}}\left(-\frac{1}{\tau},- \frac{1}{\bar{\tau}}\right).
\]
In this appendix, we will simplify this equation. Let us look at the term on the right-hand side of above equation,
\begin{gather*}
 \exp(2 \pi {\rm i} B(\mu + b, b)) \sum_{ \lambda \in L^{*} } \widehat{\Theta}_{A, \lambda + b, -\mu - b, c_{1}, c_{2}}\left(-\frac{1}{\tau},- \frac{1}{\bar{\tau}}\right) \\
 \qquad= \exp(2 \pi {\rm i} B(\mu + b, b)) \sum_{ \lambda \in L^{*} } \sum_{n \in \ZZ^{2}} \hat{\rho}^{c_{1}, c_{2}}( n + \lambda + b, \tau, \bar{\tau}) \\
\phantom{\qquad = }{}\times\exp(2 \pi {\rm i} B(n + \lambda+ b, - \mu - b)) \tilde{q}^{\frac{|n+ \lambda + b |^{2}}{2}}\\
 \qquad = \sum_{ \lambda \in L^{*} } \sum_{n \in \ZZ^{2}} \hat{\rho}^{c_{1}, c_{2}}( n + \lambda+ b, \tau, \bar{\tau}) \exp(2 \pi {\rm i} B(n + \lambda, - \mu - b)) \tilde{q}^{\frac{|n+ \lambda + b |^{2}}{2}} \\
 \qquad = \exp(- 2 \pi {\rm i} B(b,b)) \sum_{ \lambda \in L^{*} } \exp(- 2\pi {\rm i} B(\lambda, \mu + 2 b)) \\
\phantom{\qquad = }{}\times \sum_{n \in \ZZ^{2}} \hat{\rho}^{c_{1}, c_{2}}( n + \lambda+ b, \tau, \bar{\tau}) \exp(2 \pi {\rm i} B(n + \lambda + b, b)) \tilde{q}^{\frac{|n+ \lambda + b |^{2}}{2}} \\
 \qquad = \sum_{ \lambda \in L^{*} } \exp(- 2\pi {\rm i} B(\mu + b, \lambda + b)) \exp(2 \pi {\rm i} B(\mu - \lambda, b)) \widehat{\Theta}_{A,\lambda + b, b, c_{1}, c_{2}} \left(-\frac{1}{\tau},- \frac{1}{\bar{\tau}}\right).
\end{gather*}
Therefore,
\begin{gather*}
 \mathrm{e}^{- 2\pi {\rm i} B(\mu, b )} \widehat{\Theta}_{A, \mu + b, b, c_{1}, c_{2} }(\tau, \bar{\tau}) \\
 \qquad= \frac{- \tau^{-1} }{\sqrt{|\det(A)|}} \sum_{ \lambda \in L^{*} } \mathrm{e}^{- 2\pi {\rm i} B(\mu + b, \lambda + b)} \mathrm{e}^{-2 \pi {\rm i} B( \lambda, b)} \widehat{\Theta}_{A,\lambda + b, b, c_{1}, c_{2}} \left(-\frac{1}{\tau},- \frac{1}{\bar{\tau}}\right).
\end{gather*}

\section{Asymptotic structure, proofs}\label{apx:proofs}
In this appendix, we will provide a~more detailed analysis of the asymptotics of the $\zhat$ invariant on positive Brieskorn spheres when the proposal of Conjecture~\ref{conj:regform} is utilized.
Our aim will be to give a~proof for Lemma~\ref{lem:asympt}.
This appendix is divided into three sections, mimicking Section~\ref{sec:asympt}.

\subsection[Asymptotics of W(~q)]{Asymptotics of $\boldsymbol{W(\tau)}$}

We shall compute the asymptotic order of divergence of $W(\tau)$ as $q\to 1$, or equivalently as $\tau\to0$ for the indefinite theta functions which form its building blocks.
Since \smash{$\tilde{q}=\mathrm{e}^{\frac{2\pi {\rm i}}{\tau}}$} this corresponds to the leading $\tilde{q}$-power of these series.

The strategy that we will use is the following:
we will demonstrate that within the cone of summation~\eqref{eq:indef-theta-cone} of the indefinite theta series the $\tilde{q}$-power increases as the index of summation moves further away from the origin.
With this information, we will construct a~finite polynomial with two terms which will have the same leading $\tilde{q}$-power as the indefinite theta function.
Each polynomial will then be summed over in the same way as in the definition of $W(\tau)$, resulting in a~polynomial with the same leading $\tilde{q}$-order as the original sum of indefinite theta series.
We~now give a proof of Lemma~\ref{lem:leadingorderofW}.

\begin{proof}
 The indefinite theta series is given by
\[
\Theta_{A,\mu+b,b,c_1,c_2}(\tau)
=
\sum_{n\in \mathbb{Z}^2}
\rho^{c_1,c_2}(\mu+b+n)
e^{2\pi i B(\mu+b+n,b)}
q^{\frac12 B(\mu+b+n,\mu+b+n)}.
\]
To analyze the leading $q$-power, we shall look at the derivatives of the $q$-exponent along the rays of the cone.
Let
\[
 \ell_1 = n_1 + \left\lfloor \frac{\alpha_1}{2st(str+1)}\right\rfloor = n_1, \qquad \ell_2 = n_2 -n_1st + \left\lfloor \frac{2\alpha_2 + 1}{2(2r+1)}- \frac{\alpha_1}{2(rst+1)}\right\rfloor,
\]
then, for $\delta =\frac{2\alpha_2 + 1}{2(2r+1)}- \frac{\alpha_1}{2(rst+1)}$, we have $B\left(\mu+b+n,\mu+b+n\right)=R(\alpha_1,\alpha_2,r,s,t;\ell_1,\ell_2)$ and
 \begin{gather*}
 \frac{{\rm d}}{{\rm d} \ell_1}R(\alpha_1,\alpha_2,r,s,t;\ell_1,\ell_2) = 2st(2r+1)\left[\frac{st-2}{2r+1}\left(\ell_1 + \frac{\alpha_1}{2st(str+1)}\right)+\ell_2 + \delta - \left\lfloor \delta \right\rfloor\right],\\
 \frac{{\rm d}}{{\rm d} \ell_2}R(\alpha_1,\alpha_2,r,s,t;\ell_1,\ell_2) = (2r+1)\left[st\left(\ell_1 + \frac{\alpha_1}{2st(str+1)}\right)+\ell_2 + \delta - \left\lfloor \delta \right\rfloor\right],
 \end{gather*}
which are non-negative and negative for $\ell_1,\ell_2\geq 0$ and $\ell_1,\ell_2<0$, respectively.
Furthermore, we can expand the kernel of the sum $\rho^{\mathsf{c}_1,\mathsf{c}_2}$
\[
 \rho^{\mathsf{c}_1, \mathsf{c}_2}= - \operatorname{sgn}\left(\ell_1 + \frac{\alpha_1}{2st(rst+1)} \right)-\operatorname{sgn}\left(\ell_2+ \delta-\lfloor\delta\rfloor\right)
\]
therefore, for $\ell_1\geq 0$, resp.\ $\ell_1<0$, the first sign is positive, resp.\ negative.
If $\delta\neq 0$ for $\ell_2\geq 0$, resp.\ $\ell_2<0$, the second sign is positive, resp.\ negative.
If $\delta=0$, then the second sign reduces to $\operatorname{sgn}(\ell_2)$.

For $\delta\neq 0$, the theta function reduces to
 \begin{align*}
\Theta_{A,\mu+b,b,c_1,c_2}(\tau)
={}&
2\sum_{\ell_1=-\infty}^{-1}
\sum_{\ell_2=-\infty}^{-1}
e^{2\pi iB(\mu+b+n,b)}
q^{\frac12R(\alpha_1,\alpha_2,r,s,t;\ell_1,\ell_2)}
\\
&-
2\sum_{\ell_1=0}^{\infty}
\sum_{\ell_2=0}^{\infty}
e^{2\pi iB(\mu+b+n,b)}
q^{\frac12R(\alpha_1,\alpha_2,r,s,t;\ell_1,\ell_2)}.
 \end{align*}
Hence, given the derivatives of the $R(\alpha_1,\alpha_2,r,s,t;\ell_1,\ell_2)$ the leading $q$-power is given for either~${(\ell_1,\ell_2)=(0,0)}$ or $(\ell_1,\ell_2)=(-1,-1)$.

If $\delta = 0$, the theta function becomes
 \begin{align*}
\Theta_{A,\mu+b,b,c_1,c_2}(\tau)
={}&
2\sum_{\ell_1=-\infty}^{-1}
\sum_{\ell_2=-\infty}^{-1}
e^{2\pi iB(\mu+b+n,b)}
q^{\frac12R(\alpha_1,\alpha_2,r,s,t;\ell_1,\ell_2)}
\\
&-
2\sum_{\ell_1=0}^{\infty}
\sum_{\ell_2=1}^{\infty}
e^{2\pi iB(\mu+b+n,b)}
q^{\frac12R(\alpha_1,\alpha_2,r,s,t;\ell_1,\ell_2)}
\\
&+
\sum_{\ell_1=-\infty}^{-1}
e^{2\pi iB(\mu+b+n,b)}
q^{\frac12R(\alpha_1,\alpha_2,r,s,t;\ell_1,0)}
\\
&-
\sum_{\ell_1=0}^{\infty}
e^{2\pi iB(\mu+b+n,b)}
q^{\frac12R(\alpha_1,\alpha_2,r,s,t;\ell_1,0)}.
 \end{align*}
Here
\begin{gather*}
\begin{split}
 & R(\alpha_1,\alpha_2,r,s,t;-1,-1)-R(\alpha_1,\alpha_2,r,s,t;-1,0)\\
 & \qquad= 2r+1+\frac{2r+1}{str+1}\left(2st(str+1)-\alpha_1\right)>0
 \end{split}
\end{gather*}
and
\[
 R(\alpha_1,\alpha_2,r,s,t;1,0)-R(\alpha_1,\alpha_2,r,s,t;0,0)= \frac{(st-2)(st(str+1)+\alpha_1)}{str+1}>0.
\]
Hence the leading $q$-power is given by $(\ell_1,\ell_2)=(0,0)$ or $(\ell_1,\ell_2)=(-1,0)$.
\end{proof}

\subsection[Asymptotics of ZGc1]{Asymptotics of $\boldsymbol{ZG^{\textsf{c}_1}}$} \label{borelmordellintegral}
In this appendix, we will show that $ZG^{c_{1}}(\tau, \bar{\tau})$ is convergent. Recall that
\begin{align*}
 ZG^{c_{1}}(\tau, \bar{\tau} )
 ={}&
 \frac{C_{\Gamma}\bigl(q^{-1}\bigr)}{2 f_{2r+1,1}(\tau)}
 \sum_{\substack{\hat\epsilon=(\epsilon_1,\epsilon_2,\epsilon_3)\\\in (\ZZ/2\ZZ)^{\otimes 3}}}
(-1)^{\hat\epsilon}\\
 & \times\bigg[ \frac{ \sum_{\ell_{1}=0}^{2 s t(r s t+1)-1} \sum_{\ell_{2}=0}^{2r} \mathrm{e}^{- 2\pi {\rm i} B(\mathsf{a}_{\hat\epsilon}, \lambda+ \mathsf{b})} \mathrm{e}^{- 2 \pi {\rm i} B(\lambda + \mathsf{b}, \mathsf{b})} G_{A, \lambda + \mathsf{b}, \mathsf{b},\mathsf{c}_{1}}\!\bigl(-\frac{1}{\tau}, -\frac{1}{\bar{\tau}}\bigr) }{ (-\tau) \sqrt{|\det(A)|} } \\
 &\phantom{\times}
 - \mathrm{e}^{-2 \pi {\rm i} B(\mathsf{a}_{\hat\epsilon},\mathsf{b})} G_{A,\mathsf{a}_{\hat\epsilon},\mathsf{b},\mathsf{c}_{1}}(\tau, \bar{\tau})\big].
\end{align*}
For $\mu = \bigl(\frac{m_{1}}{2 s t (r s t+1)}, \frac{m_{2}}{2r+1}\bigr) \in L^{*}$,
\begin{align*}
 \mathrm{e}^{- 2 \pi {\rm i} B(\mu + b, b)} G_{A, \mu+ b,b,c_{1}}(\tau, \bar{\tau}) = & \frac{-{\rm i} }{\sqrt{2 s t(r s t+1)}} f_{2r+1,2m_{2}+1} (\tau) \int_{- \bar{\tau}}^{{\rm i} \infty} \frac{\theta^{1}_{s t (r s t+1), m_{1}}(y)}{\sqrt{-{\rm i} (y + \tau)}}.
\end{align*}
Therefore, for $\mu = \bigl(\frac{m_{1}}{2 s t (r s t+1)}, \frac{m_{2}}{2r+1}\bigr) \in L^{*}$, $\lambda = \bigl(\frac{\ell_{1}}{2 s t (r s t+1)}, \frac{\ell_{2}}{2r+1}\bigr) \in L^{*}$,
\begin{gather}
 \frac{1}{(- \tau)\sqrt{|\det(A)|} } \sum_{\ell_{1}=0}^{2 s t(r s t+1)-1} \sum_{\ell_{2}=0}^{2r} \mathrm{e}^{- 2 \pi {\rm i} B(\mu + b, \lambda+ b)} \mathrm{e}^{- 2 \pi {\rm i} B(\lambda + b, b)} G_{A, \lambda+ b, b, c} \left(-\frac{1}{\tau}, -\frac{1}{\bar{\tau}}\right) \nonumber\\
 \qquad = \frac{1}{(- \tau)\sqrt{|\det(A)|} } \left( \sum_{\ell_{1}=0}^{2 s t(r s t+1)-1} \frac{-{\rm i} \mathrm{e}^{2 \pi {\rm i} \frac{m_{1} \ell_{1}}{2 s t (r s t +1)}}}{\sqrt{2 s t(r s t+1)}} \int_{ \frac{1}{\bar{\tau}} }^{{\rm i} \infty} \frac{\theta^{1}_{s t (r s t+1), \ell_{1}}(y)}{\sqrt{-{\rm i}(y- \frac{1}{\tau})}} \dd y \right)\nonumber \\
 \phantom{\qquad = }{}\times \left( \sum_{\ell_{2}=0}^{2r} \mathrm{e}^{-{\rm i} \pi \frac{(2 \ell_{2}+1)(2 m_{2}+1)}{2(2r+1)} } f_{2r+1, 2 \ell_{2}+1} \left(-\frac{1}{\tau}\right) \right).\label{BMStransform1}
\end{gather}
The term in the second bracket is the $S$-transformation of the theta function $f_{2r+1,2m_{2}+1}$,
\begin{align}\label{BMcperpterm}
 \sum_{\ell_{2}=0}^{2r} \mathrm{e}^{-{\rm i} \pi \frac{(2 \ell_{2}+1)(2 m_{2}+1)}{2(2r+1)} } f_{2r+1, 2 \ell_{2}+1} \left(-\frac{1}{\tau}\right) = \sqrt{2r+1} \sqrt{- {\rm i} \tau} f_{2r+1,2m_{2}+1}(\tau).
\end{align}
Now let us look at the first bracket,
\begin{gather}
 \sum_{\ell_{1}=0}^{2 s t(r s t+1)-1} \frac{-{\rm i} \mathrm{e}^{2 \pi {\rm i} \frac{m_{1} \ell_{1}}{2 s t (r s t +1)}}}{\sqrt{2 s t(r s t+1)}} \int_{ \frac{1}{\bar{\tau}} }^{{\rm i} \infty} \frac{\theta^{1}_{s t (r s t+1), \ell_{1}}(y)}{\sqrt{-{\rm i}(y- \frac{1}{\tau})}} \dd y \endline
 \qquad = \sqrt{\tau} \int_{0}^{- \bar{\tau}} \frac{ y^{- \frac{3}{2}} }{\sqrt{-{\rm i}(y + \tau)}} \sum_{\ell_{1}=0}^{2 s t(r s t+1)-1} \frac{ \mathrm{e}^{2 \pi {\rm i} \frac{m_{1} \ell_{1}}{2 s t (r s t +1)}}}{\sqrt{2 s t(r s t+1)}} \theta^{1}_{s t(r s t+1), \ell_{1}}\left(\frac{-1}{y}\right) \dd y \endline
 \qquad= {\rm i}^{\frac{3}{2}} \sqrt{\tau} \int_{0}^{- \bar{\tau}} \frac{ \theta^{1}_{s t (r s t+1), - m_{1}}(y) }{\sqrt{-{\rm i}(y + \tau)}} \dd y
 \endline
 \qquad = - {\rm i}^{\frac{3}{2}} \sqrt{\tau} \int_{0}^{- \bar{\tau}} \frac{ \theta^{1}_{s t (r s t+1), m_{1}}(y) }{\sqrt{-{\rm i}(y + \tau)}} \dd y.\label{BMcparallelterm}
\end{gather}
Here we have changed the integration variable $y \rightarrow -1/y$ in the first step, and then used the $S$-transformation for unary theta functions in the second step, and finally we used the property of unary theta function
\[
- \theta^{1}_{s t (r s t+1), m_{1}}(y) = \theta^{1}_{s t (r s t+1), - m_{1}}(y).
\]
 Combining~\eqref{BMStransform1}, \eqref{BMcperpterm}, \eqref{BMcparallelterm}, we~get\looseness=-1
\begin{gather*}
 \frac{1}{(- \tau)\sqrt{|\det(A)|} } \sum_{\ell_{1}=0}^{2 s t(r s t+1)-1} \sum_{\ell_{2}=0}^{2r} \mathrm{e}^{- 2 \pi {\rm i} B(\mu + b, \lambda+ b)} \mathrm{e}^{- 2 \pi {\rm i} B(\lambda + b, b)} G_{A, \lambda+ b, b, c} \left(\frac{-1}{\tau}, \frac{-1}{\bar{\tau}}\right) \\
 \qquad= \frac{{\rm i}}{\sqrt{2 s t (r s t+1)} } f_{2r+1,2m_{2}+1}(\tau) \int_{0}^{- \bar{\tau}} \frac{ \theta^{1}_{s t (r s t+1), m_{1}}(y) }{\sqrt{-{\rm i}(y + \tau)}} \dd y.
\end{gather*}
This implies, for $\mu = \bigl(\frac{m_{1}}{2 s t (r s t+1)}, \frac{m_{2}}{2r+1}\bigr) \in L^{*}$, $\lambda = \bigl(\frac{\ell_{1}}{2 s t (r s t+1)}, \frac{\ell_{2}}{2r+1}\bigr) \in L^{*}$,
\begin{gather*}
 \frac{1}{(- \tau)\sqrt{|\det(A)|} } \sum_{\ell_{1}=0}^{2 s t(r s t+1)-1} \sum_{\ell_{2}=0}^{2r} \mathrm{e}^{- 2 \pi {\rm i} B(\mu + b, \lambda+ b)} \mathrm{e}^{- 2 \pi {\rm i} B(\lambda + b, b)} G_{A, \lambda+ b, b, c_{1}} \left(\frac{-1}{\tau}, \frac{-1}{\bar{\tau}}\right) \endline
 \qquad- \mathrm{e}^{- 2 \pi {\rm i} B(\mu + b, b)} G_{A, \mu+ b,b,c_{1}}(\tau, \bar{\tau}) \endline
\qquad\phantom{-}{}= \frac{\rm i}{\sqrt{2 s t (r s t+1)} } f_{2r+1,2m_{2}+1}(\tau) \int_{0}^{{\rm i} \infty} \frac{ \theta^{1}_{s t (r s t+1), m_{1}}(y) }{\sqrt{-{\rm i}(y + \tau)}} \dd y.
\end{gather*}
Here the integral over the contour $0$ to ${\rm i}\infty$ is obtained by concatenating the contour $0$ to $-\bar{\tau}$ coming from the $S$-transformed term with the contour $-\bar{\tau}$ to ${\rm i}\infty$ appearing in the original integral defining $G_{A,\mu+b,b,c_{1}}(\tau,\bar{\tau})$. Therefore,
\begin{align*}
 ZG^{c_{1}}(\tau, \bar{\tau} ) = & \frac{C_{\Gamma}\bigl(q^{-1}\bigr)}{2} \sum_{\substack{\hat\epsilon=(\epsilon_1,\epsilon_2,\epsilon_3)\\\in (\ZZ/2\ZZ)^{\otimes 3}}}
(-1)^{\hat\epsilon} \Bigg[ \frac{\rm i}{\sqrt{2 s t (r s t+1)} } \int_{0}^{{\rm i} \infty} \frac{ \theta^{1}_{s t (r s t+1),2s t (r s t+1)(\mathsf{a}_{\hat\epsilon})_{1} }(y) }{\sqrt{-{\rm i}(y + \tau)}} \dd y \Bigg],
\end{align*}
where $(\mathsf{a}_{\hat\epsilon})_{1}$ denotes the first component of $\mathsf{a}_{\hat\epsilon}$.
For $\tau= {\rm i} \alpha$, we can show that each integral converges with the following lemma.

\begin{Lemma}\label{lem:alpha_limit}
Let $m\in\mathbb{N}$ and let $r>0$ be real. Define, for $0\le \alpha<1$,
\begin{equation}
 I(\alpha) \coloneq
 \int_{0}^{\infty}
 \frac{
 \sum_{n\in\mathbb{Z}}
 \bigl(n+\frac{r}{2m}\bigr)
 \exp\left[- \frac{2\pi z}{4m} (2mn+r)^{2}\right]
 }{\sqrt{z+\alpha}}
 \mathrm dz.\label{eq:Iint}
\end{equation}
Then the integral $I(\alpha)$ converges absolutely for $\alpha\in [0,1)$.
\end{Lemma}

\begin{proof}
Put
\[
a=2m,\qquad \delta=\frac r{2m},\qquad S(z)=\sum_{n\in\mathbb Z}(n+\delta){\rm e}^{-\pi az(n+\delta)^2}.
\]
Differentiating the Jacobi identity
\[
\sum_{n\in\mathbb Z}{\rm e}^{-\pi t(n+b)^2} = t^{-1/2}\sum_{k\in\mathbb Z} {\rm e}^{-\pi k^2/t}e^{2\pi {\rm i}kb}
\]
with respect to $b$, and setting $t=az$, gives
\[
S(z) = -{\rm i}(az)^{-3/2} \sum_{k\in\mathbb Z} k {\rm e}^{-\pi k^2/(az)}{\rm e}^{2\pi {\rm i}k\delta}.
\]
Hence, for $0<z\le1$,
\[
|S(z)| \le C z^{-3/2}{\rm e}^{-c/z}
\]
for suitable $C,c>0$.

On the other hand, let
$d=\min\{|n+\delta| \mid n\in\mathbb Z,\, n+\delta\ne0\}>0$.
For $z\ge1$,
\begin{align*}
|S(z)|  \le \sum_{n\in\mathbb Z}|n+\delta| {\rm e}^{-\pi az(n+\delta)^2} \le {\rm e}^{-\frac{\pi a}{2}d^2z} \sum_{n\in\mathbb Z}|n+\delta| {\rm e}^{-\frac{\pi a}{2}(n+\delta)^2} \le C'{\rm e}^{-c'z}
\end{align*}
for suitable $C',c'>0$.

Since $(z+\alpha)^{-1/2}\le z^{-1/2}$ for $\alpha\ge0$,
\[
\frac{|S(z)|}{\sqrt{z+\alpha}} \le Cz^{-2}{\rm e}^{-c/z}\mathbf1_{(0,1]}(z) + C'z^{-1/2}{\rm e}^{-c'z}\mathbf1_{[1,\infty)}(z),
\]
and the right-hand side is integrable on $(0,\infty)$.
Thus $I(\alpha)$ converges absolutely.
Moreover, the same majorant and dominated convergence give $I(\alpha)\to I(0)$ as $\alpha\to0^+$.
\end{proof}

\subsection[Asymptotics of ZGc2]{Asymptotics of $\boldsymbol{ZG^{\textsf{c}_2}}$}
In this appendix, we will look at the asymptotics of $ZG^{\textsf{c}_2}$. We can split $ZG^{\textsf{c}_2}$ into two parts, $ZG^{\textsf{c}_2}(\tau, \bar{\tau})=ZG^{\textsf{c}_2}_{1}(\tau, \bar{\tau}) - ZG^{\textsf{c}_2}_{2}(\tau, \bar{\tau})$, where
\begin{gather*}
 ZG^{\textsf{c}_2}_{1}(\tau, \bar{\tau}) \coloneq \frac{4{\rm i} C_{\Gamma}\bigl(q^{-1}\bigr) }{ (-\tau) \sqrt{|\det(A)|} f_{2r+1,1}(\tau) }
 \sum_{\ell_{1}=0}^{2 s t(r s t+1)-1} (-1)^{\ell_{1}} \sin(\frac{ \ell_{1} \pi }{s}) \sin(\frac{ \ell_{1} \pi }{t}) \\
 \phantom{ ZG^{\textsf{c}_2}_{1}(\tau, \bar{\tau}) \coloneq}{}\times \sin(\frac{ \ell_{1} \pi }{r s t+1}) \sum_{\ell_{2}=0}^{2r} G_{A, \lambda + \mathsf{b}, \mathsf{b},\mathsf{c}_{2}}\left(-\frac{1}{\tau}, -\frac{1}{\bar{\tau}}\right),
\\
 ZG^{\textsf{c}_2}_{2}(\tau, \bar{\tau}) \coloneq \frac{C_{\Gamma}\bigl(q^{-1}\bigr)}{2 f_{2r+1,1}(\tau)} \sum_{\substack{\hat\epsilon=(\epsilon_1,\epsilon_2,\epsilon_3)\\\in (\ZZ/2\ZZ)^{\otimes 3}}}
(-1)^{\hat\epsilon} \mathrm{e}^{-2 \pi {\rm i} B(\mathsf{a}_{\hat\epsilon},\mathsf{b})} G_{A,\mathsf{a}_{\hat\epsilon},\mathsf{b},\mathsf{c}_{2}}(\tau, \bar{\tau}).
\end{gather*}
Using
\smash{$
 \beta(x)= \frac{\mathrm{e}^{-\pi x}}{\pi\sqrt{x}}$}, as $x\to\infty$,
we find
\begin{align*}
 \beta\left( \frac{2 B(c_{2},n)^{2} \operatorname{Im}(-1/\tau) }{-|c_{2}|^{2}} \right) \approx {}& \frac{\sqrt{(s t -2)(2r+1)(r s t+1)}}{\pi |B(c_{2},n)| \sqrt{\operatorname{Im}(-1/\tau)} } {\rm e}^{-2 \pi \bigl( \frac{(B(c_{2},n))^{2}}{-|c_{2}|^{2}} \bigr) \operatorname{Im}(-1/\tau) } \\
 \approx{} & \frac{\sqrt{-{\rm i} \tau }\sqrt{(s t -2)(2r+1)(r s t+1)}}{\pi |B(c_{2},n)| } \tilde{q}^{- \frac{(B(c_{2},n))^{2}}{|c_{2}|^{2}} }.
\end{align*}
Therefore, we can approximate \smash{$G_{A, \lambda + \mathsf{b}, \mathsf{b},c}\left(-\frac{1}{\tau}, -\frac{1}{\bar{\tau}}\right)$} near $\tau=0$ as follows\footnote{In equation~\eqref{eq:Gapproxc2}, when $B(\mathsf{c}_{2},n) = 0 $ we need to take \smash{$\frac{\operatorname{sgn}(B(\mathsf{c}_{2},n))}{|B(\mathsf{c}_{2},n))|} =0$}. }
\begin{gather}
 G_{A, \lambda + \mathsf{b}, \mathsf{b},c}\left(-\frac{1}{\tau}, -\frac{1}{\bar{\tau}}\right)\nonumber\\
 \qquad
 \approx - \sum_{n \in \lambda + \mathsf{b} + \ZZ^{2} } \frac{ \sqrt{- {\rm i} \tau} \mathrm{e}^{2 \pi {\rm i} B(n,\mathsf{b})} }{ \pi } \frac{\operatorname{sgn}(B(\mathsf{c}_{2},n))}{|B(\mathsf{c}_{2},n)|} \sqrt{(s t-2)(2r+1)(r s t+1)} \tilde{q}^{ Q_{\mathsf{c}_{2}} (n) },\label{eq:Gapproxc2}
\end{gather}
with $Q_{\mathsf{c}_{2}}$ is the quadratic form given by,
\smash{$
 Q_{\mathsf{c}_{2}} (n) \coloneq \frac{|n|^{2}}{2} - \frac{(B(\mathsf{c}_{2},n))^{2}}{|\mathsf{c}_{2}|^{2}}$}.
By \cite[Lemma~2.5]{Zwegers2008MockTF}, we can see that $Q_{\mathsf{c}_{2}}$ is a~positive definite quadratic form.

Near $\tau =0$, using modular $S$-transform of $f_{2r+1,1}$ we can approximate $1/f_{2r+1,1}(\tau)$ as
\[
 \frac{1}{f_{2r+1,1}(\tau)} \approx \frac{\sqrt{-{\rm i} \tau } \sqrt{2r+1} }{2 \sin\big(\frac{\pi r}{2r+1}\big)} \tilde{q}^{- \frac{1}{8(2r+1)}}.
\]
We get exponential singularity from $ ZG^{\textsf{c}_2}_{1}(\tau, \bar{\tau})$ if $ZG^{\textsf{c}_2}_{1}(\tau, \bar{\tau}) \approx \tilde{q}^{s} $ for some $s<0$. Therefore, we can ignore the $n$'s from the sum in~\eqref{eq:Gapproxc2} for which \smash{$Q_{\mathsf{c}_{2}}(n) > \frac{1}{8(2r+1)}$}.

Let us now look at which $n$'s can be ignored. For a~fixed $n_{2}$, minima of $Q_{\mathsf{c}_{2}}(n) $ are at \smash{$n_{1} = \frac{2n_{2}(2r+1)}{4 r s t + s t + 2}$}. At this minimum,
\[
Q_{\mathsf{c}_{2}}(n) > \frac{1}{8(2r+1)} \qquad \text{if} \quad |n_{2}| > \frac{\sqrt{4 r s t+ s t+2}}{2(2r+1)\sqrt{s t-2}}.
\]
Therefore, if
$
|n_{2}| > \smash{\frac{\sqrt{4 r s t+ s t+2}}{2(2r+1)\sqrt{s t-2}}}$,
 then \smash{$Q_{\mathsf{c}_{2}}(n) > \frac{1}{8(2r+1)}$} for all $n_{1}$.
We can argue similarly that~if
\[
|n_{1}| >\frac{\sqrt{4 r s t+ s t+2}}{\sqrt{8 s t (r s t+1)(2r+1)(st-2)}},
\]
 then \smash{$Q_{\mathsf{c}_{2}}(n) > \frac{1}{8(2r+1)}$} for all $n_{2}$. Therefore, we can restrict the sum in~\eqref{eq:Gapproxc2} to
\begin{gather}
|n_{1}| \leq \frac{\sqrt{4 r s t+ s t+2}}{\sqrt{8 s t (r s t+1)(2r+1)(st-2)}} \qquad \text{and}\qquad |n_{2}|\leq \frac{\sqrt{4 r s t+ s t+2}}{2(2r+1)\sqrt{s t-2}},\label{ZGc21singpoly2}
\\
 ZG^{\textsf{c}_2}_{1}(\tau, \bar{\tau})
 \approx \frac{4 {\rm i} C_{\Gamma}\bigl(q^{-1}\bigr) \sqrt{-{\rm i} \tau} \sqrt{(2r+1)(s t -2)}}{(- \tau) \sqrt{8 s t} \sin\bigl(\frac{\pi r}{2r+1}\bigr) } \sum_{\ell_{1}=-L_{1}}^{L_{1}} \sum_{\ell_{2}=-L_{2}}^{L_{2}} (-1)^{\ell_{1}} \sin(\frac{ \ell_{1} \pi }{s}) \sin(\frac{ \ell_{1} \pi }{t}) \nonumber\\
 \phantom{\approx }{} \times \sin(\frac{ \ell_{1} \pi }{r s t+1})
  \frac{ \mathrm{e}^{\frac{{\rm i} \pi (2 \ell_{2}+1)}{2(2r+1)}} \operatorname{sgn}\bigl(B\bigl(\mathsf{c}_{2}, \bigl(\frac{\ell_{1}}{2 s t (r s t+1)}, \frac{2 \ell_{2}+1}{2(2r+1)}\bigr) \bigr)\bigr) \tilde{q}^{ Q_{\mathsf{c}_{2}} ( \frac{\ell_{1}}{2 s t (r s t+1)}, \frac{2 \ell_{2}+1}{2(2r+1)} ) - \frac{1}{8(2r+1)} }}{\pi |B\bigl(\mathsf{c}_{2}, \bigl(\frac{\ell_{1}}{2 s t (r s t+1)}, \frac{2 \ell_{2}+1}{2(2r+1)}\bigr) \bigr)\big| },\nonumber
\end{gather}
where the bounds $L_{1}$, and $L_{2}$ are given by
\[
 L_{1} = \left\lceil \frac{\sqrt{2 s t(r s t+1)(4 r s t+ s t+ 2)}}{2\sqrt{(2r+1)(s t -2)}} \right\rceil ,\qquad L_{2}= \left\lceil \frac{\sqrt{4 r s t+ s t+2}}{2\sqrt{s t-2}} - \frac{1}{2} \right\rceil.
\]
The leading order singularity of the $\tilde{q}$-polynomial in~\eqref{ZGc21singpoly2} gives us the leading order singularity of $ZG^{\textsf{c}_2}_{1}(\tau, \bar{\tau})$.

Now let us examine the asymptotics of $ZG^{\textsf{c}_2}_{2}(\tau, \bar{\tau})$. Using equation~\eqref{nonholparteichint}, we can write $G_{A,\mathsf{a}_{\hat\epsilon},\mathsf{b},\mathsf{c}_{2}}(\tau, \bar{\tau})$ as{\samepage
\begin{align}
 G_{A,\mathsf{a}_{\hat\epsilon},\mathsf{b},\mathsf{c}_{2}}(\tau, \bar{\tau})
 = & -{\rm i} \sqrt{-|\mathsf{c}_{2}|^{2}} \sum_{I \in P_{2}} \Bigg( \sum_{\ell \in (I+\mathsf{a}_{\hat\epsilon})^{\perp} + c_{2,\perp}\ZZ } \mathrm{e}^{2 \pi {\rm i} B(\ell, b^{\perp})} q^{\frac{|\ell|^{2}}{2}} \Bigg)\nonumber\\
 &\times\int_{-\bar{\tau}}^{{\rm i} \infty} \frac{g_{\frac{B(I+\mathsf{a}_{\hat\epsilon},\mathsf{c}_{2})}{|\mathsf{c}_{2}|^{2}},B(\mathsf{c}_{2},b)}\bigl(-|\mathsf{c}_{2}|^{2}y\bigr)}{\sqrt{-{\rm i}(y+ \tau)}} \dd y \endline
 = & -{\rm i}\sqrt{-|\mathsf{c}_{2}|^{2}} \sum_{I \in P_{2}} \mathrm{e}^{{\rm i} \pi \frac{B(I+\mathsf{a}_{\hat\epsilon},c_{2,\perp}) }{ (s t-2)}} \left( \sum_{k \in \ZZ } (-1)^{k s t} q^{ \frac{s t (s t-2)}{2} \bigl( k + \frac{2B(I+\mathsf{a}_{\hat\epsilon}, c_{2, \perp})}{2 s t (s t-2)} \bigr)^{2} } \right) \endline & \times \int_{-\bar{\tau}}^{{\rm i} \infty} \frac{g_{-\frac{B(I+\mathsf{a}_{\hat\epsilon},\mathsf{c}_{2})}{ 2 (1 + 2 r) ( s t- 2) ( r s t+1)},(r s t +1)}(2 (1 + 2 r) ( s t- 2) ( r s t+1)y)}{\sqrt{-{\rm i}(y+ \tau)}} \dd y,\label{GinZGc2}
\end{align}
where we can choose $P_{2}=\{ (0,x)| 1 \leq x \leq s t-2 \}$, and take the generator to be $c_{\perp}=(1, s t)$.}

Using Poisson summation, we can write the integrand as
\begin{gather*}
 g_{-\frac{B(I+\mathsf{a}_{\hat\epsilon},\mathsf{c}_{2})}{ 2 (1 + 2 r) ( s t- 2) ( r s t+1)},(r s t +1)}(2 (1 + 2 r) ( s t- 2) ( r s t+1)y) \endline
 \qquad= {\rm i} \sum_{k \in \ZZ } \frac{k \mathrm{e}^{\frac{ - {\rm i} \pi (k+ r s t+1) B(I+\mathsf{a}_{\hat\epsilon},\mathsf{c}_{2}) }{(2r+1)(s t-2)(r s t +1)}} \mathrm{e}^{ - \frac{ {\rm i} \pi k^{2} }{2 (2r+1)(s t-2)(r s t +1) y } } }{(2 (2r+1)(s t-2)(r s t+1))^{\frac{3}{2}} (-{\rm i} y)^{\frac{3}{2}}}.
\end{gather*}
Therefore,
\begin{gather*}
 \int_{-\bar{\tau}}^{{\rm i} \infty} \frac{g_{-\frac{B(I+\mathsf{a}_{\hat\epsilon},\mathsf{c}_{2})}{ 2 (1 + 2 r) ( s t- 2) ( r s t+1)},(r s t +1)}(2 (1 + 2 r) ( s t- 2) ( r s t+1)y)}{\sqrt{-{\rm i}(y+ \tau)}} \dd y \endline
 \qquad= {\rm i} \sum_{k \in \ZZ } \frac{k \mathrm{e}^{\frac{ - {\rm i} \pi (k+ r s t+1) B(I+\mathsf{a}_{\hat\epsilon},\mathsf{c}_{2}) }{(2r+1)(s t-2)(r s t +1)}}}{(2 (2r+1)(s t-2)(r s t+1))^{\frac{3}{2}}} \int_{-\bar{\tau}}^{{\rm i} \infty} \frac{ \mathrm{e}^{ - \frac{ \pi k^{2} }{2 (2r+1)(s t-2)(r s t +1) (-{\rm i} y) } }}{(-{\rm i} y)^{\frac{3}{2}} \sqrt{-{\rm i}(y+ \tau)}} \dd y.
\end{gather*}
To write down the asymptotics of the integral above, we use the following integral identity valid for $\operatorname{Re}(\alpha) >0$ and $\operatorname{Re}(\beta) >0$:
 \[
 \int_{\alpha}^{\infty} \frac{\mathrm{e}^{-\frac{ \beta}{x}} }{x^{\frac{3}{2}} \sqrt{x + \alpha} } \dd x = \sqrt{\frac{\pi}{\alpha \beta }} \mathrm{e}^{\frac{\beta}{\alpha}} \left( \operatorname{Erfc} \left(\sqrt{\frac{\beta}{\alpha}}\right) - \operatorname{Erfc} \left(\sqrt{\frac{2\beta}{\alpha}}\right) \right)
\]
As $\alpha \rightarrow 0^{+}$,
\[
 \lim_{\alpha \rightarrow 0^{+}} \sqrt{\frac{\pi}{\alpha \beta }} \mathrm{e}^{\frac{\beta}{\alpha}} \left( \operatorname{Erfc} \left(\sqrt{\frac{\beta}{\alpha}}\right) - \operatorname{Erfc} \left(\sqrt{\frac{2\beta}{\alpha}}\right) \right) = \frac{1}{\beta}.
\]
Therefore, near $\tau =0$,
\[
 \int_{-\bar{\tau}}^{{\rm i} \infty} \frac{ \mathrm{e}^{ - \frac{ \pi k^{2} }{2 (2r+1)(s t-2)(r s t +1) (-{\rm i} y) } }}{(-{\rm i} y)^{\frac{3}{2}} \sqrt{-{\rm i}(y+ \tau)}} \dd y
 \approx \frac{2 (2r+1)(s t-2)(r s t +1) {\rm i} }{ \pi k^{2} },
\]
and consequently,
\begin{gather*}
 \int_{-\bar{\tau}}^{{\rm i} \infty} \frac{g_{-\frac{B(I+\mathsf{a}_{\hat\epsilon},\mathsf{c}_{2})}{ 2 (1 + 2 r) ( s t- 2) ( r s t+1)},(r s t +1)}(2 (1 + 2 r) ( s t- 2) ( r s t+1)y)}{\sqrt{-{\rm i}(y+ \tau)}} \dd y \\
 \qquad\approx - \frac{ \mathrm{e}^{\frac{ - {\rm i} \pi B(I+\mathsf{a}_{\hat\epsilon},\mathsf{c}_{2}) }{(2r+1)(s t-2)}} }{ \pi (2 (2r+1)(s t-2)(r s t+1))^{\frac{1}{2}} } \sum_{\substack{k \in \ZZ \\ k \neq 0}} \frac{ \mathrm{e}^{\frac{ - {\rm i} \pi k B(I+\mathsf{a}_{\hat\epsilon},\mathsf{c}_{2}) }{(2r+1)(s t-2)(r s t +1)}} }{ k } \\
 \qquad \approx - \frac{ \mathrm{e}^{\frac{ - {\rm i} \pi B(I+\mathsf{a}_{\hat\epsilon},\mathsf{c}_{2}) }{(2r+1)(s t-2)}} \log\bigl( \mathrm{e}^{ {\rm i} \pi ( \frac{ B(I+\mathsf{a}_{\hat\epsilon},\mathsf{c}_{2}) }{(2r+1)(s t-2)(r s t +1)} +1 ) } \bigr) }{ \pi \sqrt{2 (2r+1)(s t-2)(r s t+1)} }.
\end{gather*}
Using Poisson summation, we can write the theta function in the equation~\eqref{GinZGc2} as
\[
 \sum_{k \in \ZZ } (-1)^{k s t} q^{ \frac{s t (s t-2)}{2} \bigl( k + \frac{2B(I+\mathsf{a}_{\hat\epsilon}, c_{2, \perp})}{2 s t (s t-2)} \bigr)^{2} } = \frac{\mathrm{e}^{- {\rm i} \pi \frac{ B(I+\mathsf{a}_{\hat\epsilon}, c_{2, \perp}) }{ (s t -2)}} }{\sqrt{- {\rm i} s t (s t-2) \tau}} \sum_{k \in \ZZ} \mathrm{e}^{- \frac{ 2 \pi {\rm i} k B(I+\mathsf{a}_{\hat\epsilon}, c_{2, \perp}) }{s t (s t -2)}} \tilde{q}^{\frac{(2k+ s t)^{2}}{8 s t (s t -2)}}.
\]
Therefore, we can approximate $ZG^{\textsf{c}_2}_{2}(\tau, \bar{\tau})$ by
\begin{align}
 ZG^{\textsf{c}_2}_{2}(\tau, \bar{\tau})
\approx{} & \frac{{\rm i} C_{\Gamma}\bigl(q^{-1}\bigr) \sqrt{2r+1} }{4 \pi \sin\bigl(\frac{\pi r}{2r+1}\bigr) \sqrt{ s t (s t-2) } } \sum_{\substack{\hat\epsilon=(\epsilon_1,\epsilon_2,\epsilon_3)\\\in (\ZZ/2\ZZ)^{\otimes 3}}}
(-1)^{\hat\epsilon} \mathrm{e}^{-2 \pi {\rm i} B(\mathsf{a}_{\hat\epsilon},\mathsf{b})}\endline
& \times \sum_{I \in P_{2}} \sum_{k = L_{3} }^{L_{4}} \tilde{q}^{ \frac{(2k+ s t)^{2}}{8 s t (s t -2)} - \frac{1}{8(2r+1)} } \endline
& \times \log\bigl( \mathrm{e}^{ {\rm i} \pi ( \frac{ B(I+\mathsf{a}_{\hat\epsilon},\mathsf{c}_{2}) }{(2r+1)(s t-2)(r s t +1)} +1 ) } \bigr) \mathrm{e}^{- {\rm i} \pi \frac{B(I+\mathsf{a}_{\hat\epsilon},\mathsf{c}_{2})}{(2r+1)(s t -2)} - \frac{ 2 \pi {\rm i} k B(I+\mathsf{a}_{\hat\epsilon}, c_{2, \perp}) }{s t (s t -2)} },\label{ZGc22singpoly2}
\end{align}
where
\[
 L_{3} = \left\lceil -\frac{s t}{2} - \sqrt{\frac{s t(s t -2)}{4(2r+1)}} \right\rceil ,\qquad L_{4} = \left\lfloor -\frac{s t}{2} + \sqrt{\frac{s t(s t -2)}{4(2r+1)}} \right\rfloor.
\]
By restricting the sum over $k$ to the range $ L_{3} \leq k \leq L_{4} $ we are throwing away the terms that do not contribute to the exponential singularity of $ZG_{2}^{\mathsf{c}_{2}}(\tau, \bar{\tau})$. Therefore, the leading order singularity of the $\tilde{q}$-polynomial in~\eqref{ZGc22singpoly2} yields the leading order singularity of $ZG^{\textsf{c}_2}_{2}(\tau, \bar{\tau})$.

\subsection*{Acknowledgements}
The authors would like to thank Griffen Adams, Miranda Cheng, Ioana Coman, Ovidiu Costin, Gerald V.~Dunne, Sergei Gukov and O\v{g}uz \"Oner for very helpful discussions. We thank the anonymous referees for their valuable suggestions to improve the paper.
While writing this manuscript, the authors learned that~\cite{ACDGO} was soon to be published, and are especially grateful to Griffen Adams, Ovidiu Costin, Gerald V.~Dunne, Sergei Gukov and O\v{g}uz \"Oner for agreeing to coordinate publications.
The authors would like to thank Sergei Gukov again for his mentorship and encouragement to pursue this direction.
The work of DP is supported by a~Sherman Fairchild Postdoctoral Fellowship sponsored by the Walter Burke Institute for Theoretical Physics.
The work of MJ is supported by the Walter Burke Institute for Theoretical
Physics, the U.S. Department of Energy, Office of Science, Office of High Energy Physics,
under Award No. DE-SC0011632 and the Simons Collaboration Grant on "New Structures in Low-Dimensional Topology". The work of SH is supported by the FirstRand FNB 2020 Fund Education Scholarship. We acknowledge the use of Claude Opus 4.8 to simplify certain expressions and generate figures in Section~\ref{sec:Bounds on ceff}, most of which were present in a~previous version of the paper. The authors take full responsibility for the work having reviewed these simplifications.

\pdfbookmark[1]{References}{ref}
\LastPageEnding

\end{document}